\documentclass[11pt]{article}

\usepackage[T1]{fontenc}
\usepackage[utf8]{inputenc}
\usepackage{lmodern}
\usepackage[margin=1in]{geometry}
\usepackage{amsmath,amssymb,amsthm,mathtools}
\usepackage{xspace,graphicx,relsize,bm,breqn,multirow,booktabs}
\usepackage{multicol,adjustbox,tabularx,longtable,tikz,float}
\usetikzlibrary{arrows.meta,calc,fit,positioning,decorations.pathreplacing}
\usepackage{microtype}
\usepackage{physics}
\usepackage{hyperref}
\hypersetup{colorlinks=true,linkcolor=blue,citecolor=blue,urlcolor=blue}
\usepackage{cleveref}
\mathtoolsset{showonlyrefs}
\AtBeginDocument{\NewCommandCopy\origcref\cref
  \NewCommandCopy\origCref\Cref
  \RenewDocumentCommand\cref{s m}{\noeqref{#2}\IfBooleanTF{#1}{\origcref*{#2}}{\origcref{#2}}}\RenewDocumentCommand\Cref{s m}{\noeqref{#2}\IfBooleanTF{#1}{\origCref*{#2}}{\origCref{#2}}}}

\numberwithin{equation}{section}
\newtheorem{theorem}{Theorem}[section]
\newtheorem{lemma}[theorem]{Lemma}
\newtheorem{corollary}[theorem]{Corollary}
\newtheorem{proposition}[theorem]{Proposition}

\theoremstyle{remark}
\newtheorem{remark}[theorem]{Remark}
\theoremstyle{definition}
\newtheorem{definition}[theorem]{Definition}

\newcommand{\Id}{I}

\RenewDocumentCommand{\norm}{m}{\left\lVert #1\right\rVert}
\renewcommand{\rank}{\operatorname{rank}}
\newcommand{\TFD}{\textsf{TFD}}
\renewcommand{\Tr}{\operatorname{Tr}}

\RenewDocumentCommand{\ip}{m m}{\langle #1, #2 \rangle}
\newcommand{\ran}{\textsf{Im}}

\providecommand{\K}{}\renewcommand{\K}{\mathcal K}
\providecommand{\A}{}\renewcommand{\A}{\mathcal A}
\providecommand{\C}{}\renewcommand{\C}{\mathbb C}
\providecommand{\B}{}\renewcommand{\B}{\mathcal B}
\newcommand{\Aimp}{\mathcal A_{\rm imp}}
\newcommand{\calF}{\mathcal F}
\providecommand{\W}{}\renewcommand{\W}{\calW}
\providecommand{\Z}{}\renewcommand{\Z}{\mathcal Z}
\RenewDocumentCommand{\ket}{m}{\lvert#1\rangle}
\RenewDocumentCommand{\bra}{m}{\langle#1\rvert}

\newcommand{\calW}{\mathcal W}

\newcommand{\Span}{\operatorname{span}}
\newcommand{\poly}{\operatorname{poly}}

\title{Quantum impurity models: easy at equilibrium, universal in motion}

\author{
  Srinivasan Arunachalam \quad Sergey Bravyi \quad Anirban Chowdhury\\
  Arkopal Dutt \quad Alexandru Gheorghiu \quad Zhi Li\\[0.5em]
  \small IBM Research
}
\begin{document}

\date{}
\maketitle

\begin{abstract}
A quantum impurity model describes a small interacting subsystem embedded into a large bath of free fermions.
Here we study the computational complexity of  calculating  the ground energy,  thermal equilibrium, and dynamical properties of these models.
Our work reveals a sharp contrast: equilibrium properties can be efficiently approximated by classical means, whereas time evolution can implement a universal quantum computation.
More precisely, let $H$ be the Hamiltonian of an impurity model with $n$ fermionic modes and a constant-size impurity. We show that:
\begin{itemize}
\setlength{\itemsep}{0pt}
\item the ground energy of $H$ can be approximated to additive error $\varepsilon$ by a classical algorithm with runtime $\poly(n,1/\varepsilon)$, improving on the quasi-polynomial runtime of the best previously known algorithm;
\item at inverse temperature $\beta$, the Helmholtz free energy and a classical description of the thermofield double state can be computed to precision $\varepsilon$ in time $\poly(n,\beta,1/\varepsilon)$;
\item simulating the time evolution $e^{-iHt}$ is $\mathsf{BQP}$-complete, for $H$ that is time-independent and has a fixed, constant impurity size. 
\end{itemize}
Our algorithms exploit exponential suppression of multi-particle  bath excitations in a basis organized by energy scale and Krylov depth. 
Our universality construction realizes a stationary quantum processor whose program arrives in a stream of freely propagating fermions.
\end{abstract}

\newpage 
\setcounter{tocdepth}{2}
{ \tableofcontents}

\newpage 

\section{Introduction}

Quantum impurity models provide a natural arena for studying the computational cost of interactions in fermionic Hamiltonians. 
The model consists of a small interacting subsystem, called the impurity, embedded in a large bath of free fermions. 
Without the impurity, ground energies, thermal properties, and 
time evolution of free fermions can be efficiently simulated classically~\cite{terhal2002classical}. With interactions spread throughout the system, ground energy estimation can be $\mathsf{QMA}$-complete, even for variants of the Fermi--Hubbard model and electronic-structure Hamiltonians~\cite{schuch2009computational,ogorman2022electronic}. Impurity models interpolate between these two extreme limits. They restrict 
the support of interactions to a small subsystem of size $m$, which is much smaller than the total system size $n$.
However, the interaction can be arbitrarily strong, so it  
can drastically change properties of the system  and its effect can spread throughout the bath.

Originally introduced to describe magnetic impurities in metals~\cite{anderson1961localized}, these models 
now serve as building blocks for material simulations based on 
the dynamical mean-field theory~\cite{georges1996dynamical,kotliar2006electronic}.  The latter approximates an interacting lattice model by an impurity coupled to a self-consistently determined bath.
The resulting calculation repeatedly calls an impurity solver, typically to evaluate
Green's function describing the propagation of an added or removed fermion. Bauer \emph{et al.}\ proposed delegating  this step to a quantum computer while performing  the surrounding self-consistency loop
classically~\cite{bauer2016hybrid}. This raises a basic question: which properties of  impurity models can be computed efficiently by classical algorithms, and where might a quantum computer offer an advantage?

Classical impurity solvers have achieved considerable success. Numerical renormalization group~\cite{wilson1975renormalization,bulla2008numerical,merker2012numerical}, Monte Carlo methods~\cite{gull2011continuous}, and density matrix renormalization group~\cite{white1992density} provide powerful methods for equilibrium calculations. Modern tensor-network solvers, including fork tensor-product states and methods using Gaussian changes of basis, also access spectral functions and real-time response~\cite{wolf2015imaginary,bauernfeind2017fork,wu2025disentangling}. Long-time evolution is often more demanding
because entanglement growth makes compact tensor-network representations increasingly costly to maintain~\cite{grundner2024complex}.
These observations suggest  a distinction between the cost of equilibrium and dynamical simulations for impurity models, which is the main topic
investigated in the present work.

Bravyi and Gosset~\cite{bravyi2017complexity} gave the first classical algorithm for approximating the ground energy of impurity models with rigorous performance guarantees.
Assuming that the impurity size is a fixed constant, the algorithm of~\cite{bravyi2017complexity} 
approximates the ground energy with an additive error $\varepsilon$ in time  polynomial in  the total system size $n$ and exponential in $\log^3{(1/\varepsilon)}$,
that is, quasi-polynomial in $1/\varepsilon$.
Their analysis also established low-rank Gaussian approximations to low-energy states, motivating practical variational algorithms~\cite{boutin2021quantum,snyman2021efficient} and subsequent work on quantum embedding~\cite{erakovic2025high}. In this work, we will be concerned with two central questions, that were posed in the work of~\cite{bravyi2017complexity}: 
\begin{quote}
    \begin{center}
        \emph{\textbf{1.} Is there a  $\poly(n,1/\varepsilon)$-time classical algorithm for the ground energy problem?}\\[1.5mm]
\emph{\textbf{2.} Can a time-independent Hamiltonian with a constant-sized impurity be universal?}
    \end{center}
\end{quote}
Here universality means that time evolution under the considered class of Hamiltonians can realize an arbitrary quantum computation with at most polynomial slowdown; previously, it was known for time-dependent impurity Hamiltonians~\cite{brod2014computational} and for time-independent models whose impurity size grows with $n$~\cite{pham2026time}.

The main contribution of this work is in answering both questions above in the affirmative. Additionally, we  give polynomial-time classical algorithms for computing a number of thermal equilibrium properties, that could be of independent interest.

\subsection{Main results}
\label{sec:main-results}

\paragraph{Model.}
Consider $n$ fermionic modes with annihilation and creation operators
$a_j$ and $a_j^\dagger$, and a small subset
$\mathcal I\subseteq\{1,\ldots,n\}$ constituting
the impurity.
Each mode can be either empty or occupied by a fermion.
The corresponding $2^n$-dimensional Hilbert space $\mathcal F_n$, known as the Fock
space, is equipped with the occupation number basis
$\ket{x}$, where $x\in \{0,1\}^n$.
The standard number-conserving impurity Hamiltonian is
\[
H=H_0+V,\qquad
H_0=\sum_{j,k=1}^n h_{jk}a_j^\dagger a_k,\qquad
V=\sum_{j,k,\ell,r\in\mathcal I}
v_{jk\ell r}a_j^\dagger a_k^\dagger a_\ell a_r.
\]
Here $h=h^\dagger$ and the interaction coefficients are chosen so that
$V$ is Hermitian.
The quadratic term $H_0$ describes noninteracting fermions on all $n$
modes, including hopping between the 
impurity and the remaining modes.  
We shall refer to $H_0$ as the \emph{bath Hamiltonian} and choose the energy scale such that $\|h\|=O(1)$.
The quartic term $V$ is the \emph{interaction}: it acts only on the
modes in $\mathcal I$ and may have arbitrary strength.
Following the convention of~\cite{bravyi2017complexity}, let us say that the impurity has size $m=2|\mathcal I|$.
Our results also allow pairing terms in $H_0$ and more general
impurity interactions\footnote{The impurity size $m$ counts the number of Majorana modes in $\mathcal I$.}, as defined precisely in
Section~\ref{sec:impurity-model}.

To describe our results, let us first recall some standard terminology pertaining to fermionic systems, see e.g.~\cite{terhal2002classical,bravyi2004lagrangian,bravyi2017complexity}.
 A {\em Bogoliubov transformation} is a unitary operator
$U$ on the Fock space ${\mathcal F}_n$ such that the conjugated action $a_i \mapsto U^\dagger a_i U$ sends each annihilation operator to a linear combination
of annihilation and creation operators $a_j,a_j^\dagger$ with $j=1,\ldots,n$. Such a unitary $U$ can be succinctly described
by $O(n^2)$ coefficients describing its conjugated action on the annihilation operators.  A fermionic  {\em Gaussian state} is a state of the form $U|x\rangle$, where
$U$ is a Bogoliubov transformation and $x\in \{0,1\}^n$ is a Fock basis vector 
(here we only consider pure Gaussian states).

\paragraph{Ground state estimation.} 
Our first computational task is to approximate the global ground
energy of the Hamiltonian $H$, i.e., approximate the quantity
\[
    E_1(H)
    :=
    \min_{\substack{\phi\in\mathcal F_n\\
                    \lVert\phi\rVert=1}}
    \langle\phi|H|\phi\rangle,
\]
up to accuracy $\varepsilon \in (0,1/2)$ and given an input of the compact classical description of \(H\).

Our first result gives both an approximation to the ground energy and
a succinct classical description of a state attaining approximately
that energy.
\begin{theorem}[Ground-energy approximation]
\label{thm:ground-energy-algorithm}
\label{thm:main}
Let $\varepsilon \in (0,1/2)$. Let $H=H_0+V$ be a Hamiltonian describing a quantum impurity model with $n$ fermionic modes and impurity size $m$ with $1\le m\le n$. Let $E_1(H)$ be the smallest eigenvalue of $H$. There is a classical algorithm that computes $\widehat{E}_1$ satisfying
$$
E_1(H) \leq \widehat{E}_1 \leq E_1(H)+\varepsilon.
$$
The algorithm has runtime $\poly(n,1+\|V\|)\, \exp[O(m\log(m/\varepsilon))]$.
The algorithm also computes a Bogoliubov transformation $U$ and a normalized state $\phi \in {\mathcal F}_n$ such that 
$\widehat{E}_1=\langle \phi|U^\dag HU |\phi\rangle$. The state  $\phi$ is succinctly specified as superposition
of $ \exp[O(m\log(m/\varepsilon))]$ Fock basis vectors.
\end{theorem}
It suffices to consider $m\leq n$, since otherwise direct
computation in the full Fock space takes time
$2^{O(n)}\leq2^{O(m)}$, within the claimed bounds.
For a constant impurity size $m=O(1)$ and the interaction Hamiltonian satisfying $\|V\|\le \poly(n)$, we obtain a polynomial-time algorithm i.e., the runtime is $\poly(n,1/\varepsilon)$. For comparison, the previously best known algorithm of~\cite{bravyi2017complexity} has runtime scaling exponentially with $m\log^3{(m/\varepsilon)}$. 
We note that recently Jiang~\emph{et al.}\ announced a polynomial-time {\em quantum} algorithm approximating the ground energy of impurity models~\cite{jiang2026impuritybqp}.
Our classical algorithm borrows some ideas of~\cite{jiang2026impuritybqp} such as the Krylov subspaces, see Section~\ref{sec:ground-energy-proof} for details.

\paragraph{Thermofield-double approximation.}

Our second result concerns the thermal equilibrium properties of quantum impurity models. 
Let $\beta \ge 0$ be the inverse temperature and  $\rho_\beta$ be the thermal Gibbs state
\[
\rho_\beta = \frac{e^{-\beta H}}{ \mathrm{Tr} \left( e^{-\beta H} \right)}
\]
with the canonical purification
\[
    \lvert\mathsf{TFD}_\beta\rangle
    \in\mathcal F_n\otimes\mathcal F_n
    \cong\mathcal F_{2n},
\]
which is often referred to as the \emph{Thermofield Double  state} (see  Section~\ref{sec:tfd-preliminaries} for a formal definition). Our second result shows that this purification can be approximated by a short superposition of fermionic Gaussian states. 
\begin{theorem}[Gaussian spanning set]\label{thm:rank}
Let $\delta \in (0,1)$. Fix a bath Hamiltonian $H_0$ and an impurity support $\cal I$ as above. For any  inverse temperature $\beta\ge 0$, there exists a Bogoliubov unitary  $U_\beta\, : \, \calF_{2n}\to \calF_{2n}$ and a set of Fock configurations
\begin{equation}\label{eq:rank-main}
\mathcal X_{\beta,\delta} \subseteq \{0,1\}^{2n}, \qquad 
 |\mathcal X_{\beta,\delta}|  \le  (2+\beta)^{O(m)}/\delta
\end{equation}
such that 
\begin{equation}\label{eq:TFD_approx}
\sum_{x\in \mathcal X_{\beta,\delta}} |\langle  \mathsf{TFD}_\beta|U_\beta|x\rangle|^2 \ge 1-\delta^2
\end{equation}
for all interaction Hamiltonians $V$ supported on ${\cal I}$. The Bogoliubov transformation $U_\beta$ and the set $\mathcal X_{\beta,\delta}$ can be computed in time $\poly(n,1+\beta)$ and $\poly(n,1+\beta,1/\delta)$ respectively, where the degree of polynomials may depend on $m$.
\end{theorem}
Here $U_\beta$ depends only on $H_0,{\cal I},\beta$, while $\mathcal X_{\beta,\delta}$ additionally depends on the accuracy $\delta$. Neither depends on the interaction Hamiltonian $V$ and no bound on $\norm{V}$ is required. Theorem~\ref{thm:rank} implies that $|\mathsf{TFD}_\beta\rangle$ can be approximated by a superposition of states
$U_\beta|x\rangle$ with $x\in \mathcal X_{\beta,\delta}$ but does not specify the coefficients in this superposition. In the next theorem, we show that these coefficients  can be efficiently computed as well under an additional norm bound $\|V\|\le 1$. This provides an efficient classical algorithm for computing a low-rank superposition of Gaussian states approximating the thermofield double state. 
\begin{theorem}[Constructive TFD approximation]\label{ct:main}
Suppose $\norm{V}\le1$.
For  $\beta\ge0$ and $\delta \in (0,1/2)$,
let $U_\beta\, : \, \calF_{2n}\to \calF_{2n}$ and $\mathcal X_{\beta,\delta}\subseteq \{0,1\}^{2n}$ be the  Bogoliubov transformation and  the set of Fock configurations of Theorem~\ref{thm:rank}. There is a classical algorithm  that computes a set of coefficients $\{c_x \in \C \, : \, x\in \mathcal X_{\beta,\delta}\}$ such that  $\sum_{x\in  \mathcal X_{\beta,\delta}} |c_x|^2=1$ and 
\[
\left\| 
|\mathsf{TFD}_\beta\rangle - \sum_{x\in \mathcal X_{\beta,\delta}} c_x U_\beta |x\rangle \right\| \le C\delta
\]
for a universal constant $C>0$.
The algorithm has runtime
\[
 \poly(n,1+\beta,1/\delta)
 \exp\!\left[O\!\left(m\log\frac{m(1+\beta)}{\delta}\right)\right].
\]
\end{theorem}
Thus, for constant $m$, the above algorithm has runtime polynomial in $n,\beta,1/\delta$. 
Theorem~\ref{ct:main} allows us to approximate Gibbs-state expectation values of observables given as sums of polynomially
many monomials in the 
Majorana operators $a_j+a_j^\dagger$
and $i(a_j-a_j^\dagger)$. Their matrix elements between the
Gaussian states in the TFD approximation can be computed
efficiently using the generalized Wick theorem
of~\cite{bravyi2017complexity}.
For an observable $P$, the resulting additive error in approximating $\mathrm{Tr}(\rho_\beta P)$ is at most $2C\delta \|P\|$, where $\delta$ is the error parameter of  Theorem~\ref{ct:main},
see Section~\ref{ct:thermodynamics}.
Finally, using Theorem~\ref{ct:main},
we can obtain an estimate of the Helmholtz free energy of quantum impurity models.
\begin{corollary}[Free-energy estimation]\label{ct:free-energy}
There is a classical algorithm that takes as input an impurity model Hamiltonian $H=H_0+V$ with
$\norm{V}\le1$, an inverse temperature $\beta>0$, and an
error tolerance $0<\varepsilon\le1/2$, and computes a real
number $F_{\rm out}$ satisfying
\[
 \left|F_{\rm out}-\mathsf F_\beta\right|\le\varepsilon,
 \qquad
 \mathsf F_\beta=-\frac1\beta\log\Tr e^{-\beta H}.
\]
The algorithm has runtime
\[
 \poly(n,1+\beta,\varepsilon^{-1})
 \exp\!\left[
 O\!\left(m\log\frac{m(1+\beta)}{\varepsilon}\right)
 \right],
\]
which is polynomial in $n$, $1+\beta$, and $\varepsilon^{-1}$
for any constant impurity size $m$.
\end{corollary}

\paragraph{Universal quantum computation.}
Our final result demonstrates that 
the time evolution generated by quantum impurity models can simulate a universal quantum computation with at most polynomial slowdown. 
The matchgate constructions of Brod and Childs~\cite{brod2014computational} imply universality for time-dependent impurity Hamiltonians with constant impurity size. Pham and Santos~\cite{pham2026time} obtained a time-independent construction with $m=O(S\log S)$ impurity modes for circuit depth $S$. 
Our construction, by combining the fixed, autonomous update rule of quantum Turing-machine models \cite{Deutsch1985,BernsteinVazirani1997} with fermionic quantum-walk and scattering computation \cite{childs2013universal}, achieves BQP-universal dynamics with a constant-size impurity and no external control.

\begin{theorem}[Universal quantum computation]
\label{thm:main-reduction}
There exists a fixed, number-conserving quartic interaction $V_{\mathcal I}$ supported on a fixed set $\mathcal I$ of $24$ fermionic modes and a fixed output mode $o\in\mathcal I$ with the following property.
Let $\mathcal C$ be a quantum circuit on $k\ge1$ qubits with $g$ gates from the universal gate set $\{H,T,\operatorname{CZ}\}$\footnote{Alternatively, we may use the universal two-qubit gate in \cite{barenco1995universal}, decreasing the impurity size to 16.}, and let
\begin{equation}\label{eq:qc}
p_{\mathcal C}=\expval{\mathcal C^\dagger (\ketbra{0})_{k}\, \mathcal C}{0^k}
\end{equation}
be the probability of measuring $0$ on its output qubit after applying $\mathcal C$ to $\ket{0^k}$.
Then one can construct, by a deterministic classical algorithm running in time polynomial in $k+g$,
a quadratic bath Hamiltonian $H_0=\sum_{j,\ell=1}^n h_{jl}a_j^\dagger a_\ell$ with $h=h^\dagger$ and $\norm{h}\le2$,
an explicitly specified Slater determinant $\ket{\Psi}\in\mathcal F_n$,
and a positive integer evolution time $T$, such that
\begin{equation}
\label{eq:main-reduction}
\Bigl|\langle\Psi|e^{i(H_0+V_{\mathcal I})T}a_o^\dagger a_o\,
e^{-i(H_0+V_{\mathcal I})T}|\Psi\rangle-p_{\mathcal C}\Bigr|\le\frac1{12},
\qquad n,T=O\bigl((k+g)^{4}\bigr).
\end{equation}
\end{theorem}

We note that the interaction $V_{\mathcal I}$ and output mode $o$ are identical for every instance, and the information of the circuit $\mathcal C$ is encoded into the initial state $\ket{\Psi}$.
As a corollary, estimating the occupation of the output mode $o$ after time evolution under a quantum impurity Hamiltonian is $\mathsf{BQP}$-hard, and the associated promise problem is $\mathsf{BQP}$-complete; containment in $\mathsf{BQP}$ follows from standard Hamiltonian simulation.

\begin{corollary}[BQP-hardness]\label{cor:bqp}
For $V_{\mathcal I}$ and $o$ as in \Cref{thm:main-reduction}, consider the following problem:
given $H_0$, $\ket{\Psi}$ and $T$ as above, with $T$ bounded by a polynomial in the input length,
estimate
\begin{equation}
\label{eq:occupation}
\langle\Psi|e^{i(H_0+V_{\mathcal I})T}a_o^\dagger a_o\,
e^{-i(H_0+V_{\mathcal I})T}|\Psi\rangle
\end{equation}
to additive error $1/12$ with success probability at least $2/3$.
This estimation problem is $\mathsf{BQP}$-hard, and the associated promise problem of deciding whether \eqref{eq:occupation} is at least $2/3$ or at most $1/3$ is $\mathsf{BQP}$-complete.
\end{corollary}

\subsection{AI statement and acknowledgements} 
The key ideas and technical work behind the main classical algorithms, as well as the BQP-hardness result, were generated by GPT-5.6 Sol and Astra.
The idea of using Krylov subspaces was obtained from a talk of Jiang~\cite{jiang2026impuritybqp} on a \emph{quantum algorithm} for impurity models, which we used as a starting point in searching for a classical polynomial-time algorithm.  
Subsequent interactions with AI were used to refine and clarify the proof.
The authors independently verified, developed, and simplified the argument presented here and take full responsibility for its correctness, exposition, and attribution. 

AG is grateful to the Simons Institute for the Theory of Computing for hosting him for the 2026 Quantum Summer Cluster at the start of this work.

\paragraph{Concurrent work.} While finalizing the paper, we learned of Jiang et al.'s work~\cite{Jiang2026} that obtains similar and complementary results to ours. In particular, they also construct polynomial-time classical algorithms for equilibria properties but their techniques for finite-temperature properties diverge from ours. They prove computational hardness results for evaluation of certain dynamical quantities\,---\,namely the thermal Green's functions\,---\,with time-dependent Hamiltonians whereas we establish universal quantum computation with time-independent Hamiltonians.

\subsection{Technical overview I: equilibrium algorithms}
In this section we give a brief overview of the classical algorithms for equilibrium.

\subsubsection{Ground energy estimation}
Here we sketch main ideas used to prove Theorem~\ref{thm:ground-energy-algorithm} which gives an efficient classical algorithm for ground-energy estimation of fermionic Hamiltonians with impurities. Given an input impurity Hamiltonian, we identify a small set of Fock configurations whose span contains a state of nearly optimal energy and prove that these configurations can be efficiently enumerated. Diagonalizing the Hamiltonian restricted to this span then gives the desired ground-energy estimate. Our approach shares structural similarities with ideas used in numerical renormalization group and block Lanczos methods for impurity models~\cite{bulla2008numerical,shirakawa2014block}.
We illustrate the associated main ideas in Figure~\ref{fig:impurity-shell-krylov} and describe them below. 

The impurity Hamiltonian can be brought into the following canonical form after applying a Bogoliubov transformation and a known shift by a multiple of the identity (see e.g.,~\cite{terhal2002classical,bravyi2017complexity})
\begin{equation}\label{eq:H0_after_Bogo}
H = H_0 + V, \quad H_0=\sum_{j=1}^n e_j \beta_j^\dagger \beta_j = d\Gamma(h), \quad h = \mathrm{diag}(e_1,\ldots,e_n), \quad 0\le e_j\le 1.    
\end{equation}
Here, $d\Gamma(h)$ denotes the quadratic (free) Hamiltonian associated with the one-particle matrix $h$. For a vector $z\in \C^n$, let $\beta(z) = \sum_{j=1}^n \overline{z}_j \beta_j$ be the annihilation operator of the one-particle direction $z$. 
 Each original mode is a linear combination $a_j=\beta(w_j)+\beta(w_j')^\dagger$ for some vectors $w_j,w_j'\in\C^n$. By assumption, $V$ is a polynomial in $a_j,a_j^\dagger$ with $j\in\mathcal I$, hence an even Hermitian polynomial in $\beta(w)$ and $\beta(w)^\dagger$ with $w$ in the \emph{impurity subspace}
 \begin{equation}
 \label{eq:one-particle-directions}
\Aimp:=\Span{\left\{w_j,w_j' : j\in\mathcal I\right\}}.
 \end{equation}
 By definition, $m=2|\mathcal I|$ and thus
$\dim{\Aimp}\le m$. We can assume without loss of generality that 
\[
\dim{\Aimp}=m
\]
by extending $\Aimp$ arbitrarily to an $m$-dimensional subspace, if necessaary.

Even though $V$ is supported only on one-particle directions from $\Aimp$, the quadratic Hamiltonian $H_0$ can nevertheless couple $\Aimp$ to additional bath directions. The challenge is to determine which additional bath directions and simultaneous occupations to retain so that the resulting subspace contains a state of nearly optimal energy.

\begin{figure}[!htbp]
  \centering
  \includegraphics[width=15cm]{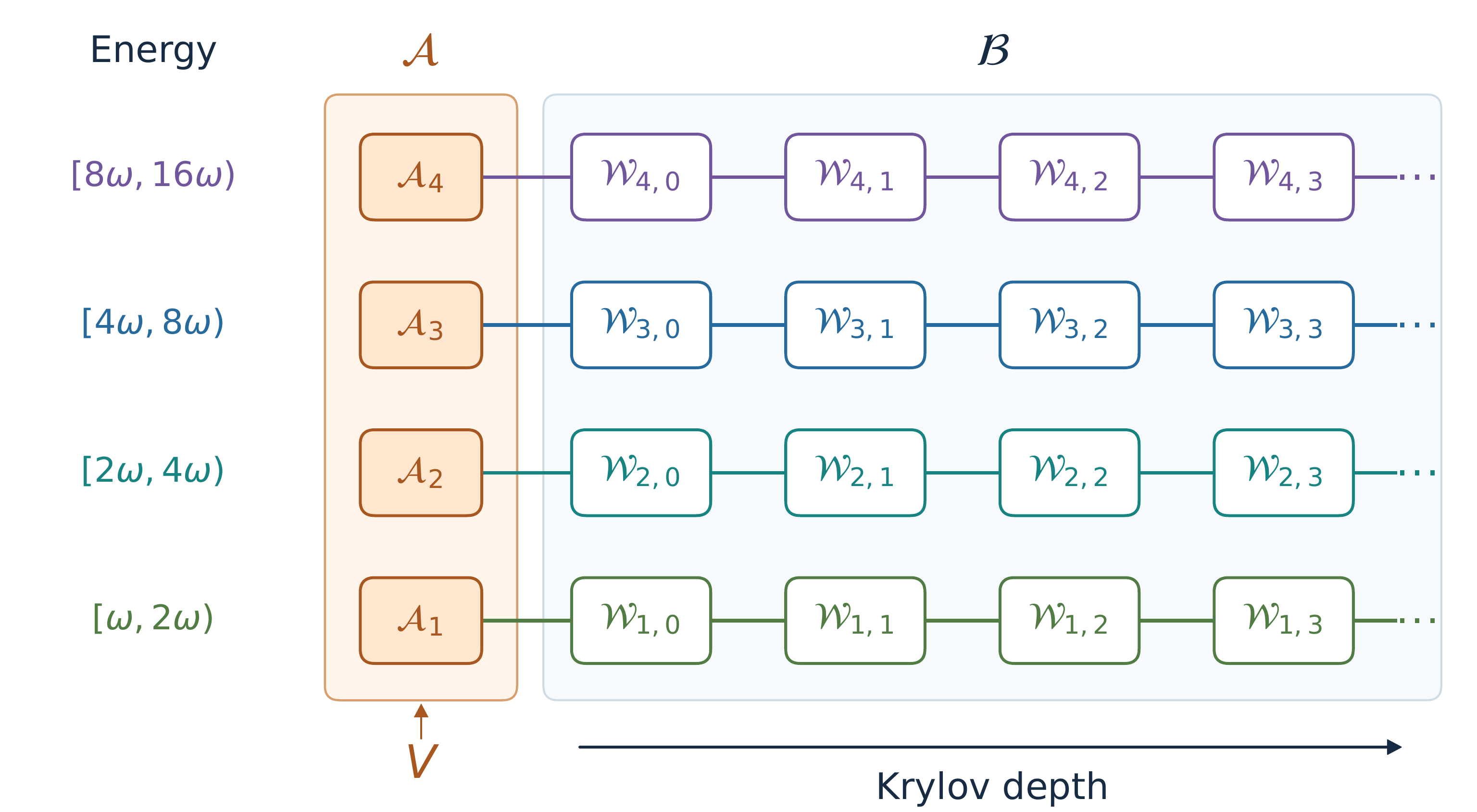}
  \caption{\textbf{Shell Krylov representation of a quantum impurity model.}
  The figure illustrates the energy shells and Krylov layers. After creating a small artificial gap $\omega>0$ in the bath Hamiltonian and changing the one-particle basis, the Hamiltonian is represented as a network of fermionic blocks indicated by boxes.  Each box contains at most $m$ fermionic modes, where $m=O(1)$ is the impurity size.
Each horizontal chain of boxes corresponds to a dyadic  energy shell spanned by bath modes with excitation energies in $[2^p\omega,2^{p+1}\omega)$ for some integer $p\ge 0$.
Blue boxes in each horizontal chain are successive orthogonal Krylov layers. The quadratic bath Hamiltonian couples nearest-neighbor boxes in each horizontal chain.  The orange boxes together form the boundary $\mathcal A$, which contains all fermionic modes supporting the impurity. The impurity Hamiltonian can couple orange boxes from different chains. 
Our algorithm exploits a special ground state structure: a box with a  large Krylov depth has an exponentially small probability of being nonempty. More strongly, joint occupation number probabilities satisfy
product bounds 
that control simultaneous excitations in multiple boxes, see Lemma~\ref{ge:pull} for a formal statement.  We truncate the Hilbert space using a weighted particle budget $R$. Each particle in a blue box at Krylov depth $\ell$ consumes $\ell+1$ units of this budget, summed over all chains. Occupations of the orange boundary boxes are unrestricted. We retain the Fock configurations satisfying this constraint and diagonalize the original Hamiltonian projected onto their span. For additive ground-energy error $\varepsilon$, choose
$\omega=\varepsilon/(4m)$. It then suffices to take
$R=\lceil C m\log(m/\varepsilon)\rceil$ for a sufficiently
large universal constant $C$.
}
    \label{fig:impurity-shell-krylov}
\end{figure}

\paragraph{Opening an artificial bath gap.}
The first obstacle is that the bath $H_0$ may contain arbitrarily many low-energy modes. To reduce this dimension, we can choose to only keep those configurations below a given bath-energy threshold. However, this still may retain exponentially many configurations since many modes can be occupied independently while their total energy remains below the threshold. Restricting instead to configurations with at most a fixed number of particles reduces the dimension, but may discard relevant configurations since low bath energy does not imply low particle number. 
To keep both the retained space small and control the error in the ground energy, we thus choose to regularize the bath energies instead of truncating the Fock space. We achieve this by fixing a parameter $\omega = \varepsilon/(4m)$ given a desired accuracy $\varepsilon \in (0,1)$ and then replacing $e_j$ in Eq.~\eqref{eq:H0_after_Bogo} by
$$
\widetilde{e}_j := \max\{e_j,\omega\},
$$
and defining the modified bath Hamiltonian denoted by $\widetilde{H}_0$ and that including the impurity interaction $V$, denoted by $\widetilde{H}$, as
\begin{equation}\label{eq:gap_Hamiltonians}
\widetilde H_0 = \sum_j \widetilde{e}_j \beta_j^\dagger \beta_j, \qquad \widetilde H=\widetilde H_0+V .
\end{equation}
Denoting the new one-particle matrix as $\widetilde{h} = \mathrm{diag}(\widetilde{e}_1,\ldots,\widetilde{e}_n)$ so that $\widetilde{H}_0 = d\Gamma(\widetilde{h})$, we immediately observe that $\omega \Id \preceq \widetilde h\preceq I$. The bath is thus now gapped but still contains all $n$ modes.
Note also that $H\preceq \widetilde H$.

Moreover, we prove that the ground-state energy of this modified Hamiltonian is still a good approximation of the true energy $E_1(H)$. Particularly, we show in Lemma~\ref{ge:rounding-lemma} (in Section~\ref{sec:ground-energy-proof}) that
$$
0 \leq E_1(\widetilde H)-E_1(H) \leq m\omega = \varepsilon/4.
$$
Note that the approximation error only depends on the impurity dimension and not on the number of low-energy bath modes. This is due to the fact that among the bath directions whose original energies are
below $\omega$, only the projection of the impurity subspace $\Aimp$ can participate nontrivially in the interaction, and this projection has dimension at most \(m\).  All orthogonal low-energy directions are free zero-energy modes after temporarily setting these energies to zero, and a ground state may be chosen with all such modes empty.  Thus only an \(m\)-dimensional collection of low-energy directions can incur the rounding cost.  

\paragraph{Energy shells and the enlarged boundary.}
The next step is to exploit the fact that the impurity, which acts only through the subspace $\Aimp$ of dimension at most $m$, also only sees a few directions at each energy scale. We partition the spectrum of $\widetilde{h}$ into dyadic intervals
\[
[\omega,2\omega),\ [2\omega,4\omega),\ [4\omega,8\omega),\ldots ,
\]
Specifically, we define energy \emph{shells}
$$
I_s := [\omega_s, 2\omega_s) \quad \text{where} \quad \omega_s := 2^s \omega, \forall s = 0,\ldots, S-1,
$$
with $S = O(\log(1/\omega)) = O\!\left(\log{m}/{\varepsilon}\right)$. Let $\mathcal{H}_s$ be the corresponding spectral subspaces of the single-body Hilbert space and let $\Pi_s$ denote the orthogonal projector onto $\mathcal{H}_s$. Note that we simply bin the spectrum of $\widetilde{h}$ and do not replace the energies within a shell by a common value. A shell may contain many modes but the interaction $V$ sees at most $m$ directions in it. We define the shell boundary $\A_s$, which is the portion of a given shell $s$ accessible to the impurity interaction, as
\begin{equation}\label{def:shell_boundary}
\A_s := \Pi_s \Aimp \subseteq\mathcal H_s
   \end{equation}
where $\Aimp$ is the impurity subspace defined in Eq.~\eqref{eq:one-particle-directions}. As $\dim \Aimp = m$, we have  $\dim \A_s\le m$. The corresponding orthogonal complement $\B_s$ is then
\[
\B_s := \mathcal{H}_s \ominus \A_s = \mathcal{H}_s \cap \A_s^\perp,
\]
which we call the \emph{residual bath} and where $\ominus$ denotes the orthogonal complement within the indicated space. Thus, $\mathcal H_s=\mathcal{A}_s\oplus \mathcal{B}_s$. Modes in \(\mathcal{B}_s\) do not appear directly in the interaction,
although they may still influence the impurity indirectly through the
quadratic bath Hamiltonian. Taking direct sums over all shells gives
\[
\A=\bigoplus_s \A_s, \qquad \B=\bigoplus_s \B_s,
\]
with $ \dim \A \leq mS = O\!\left(m\log {m}/{\varepsilon}\right)$. Thus all interacting degrees of freedom have been collected into a small one-particle boundary \(\A\), while the potentially large remainder \(\B\) is free except for its coupling to this boundary. Relative to \(\A\oplus \B\), the rounded one-particle Hamiltonian has block
form
\[
\widetilde h =
\begin{pmatrix}
    h_{\A} &T\\
    T^\dagger & E
\end{pmatrix},
\]
where \(E=\bigoplus_sE_s\) acts within the residual bath and
\(T=\bigoplus_sT_s\) couples the residual bath to the boundary.
If shell \(s\) contains energies in \([\omega_s,2\omega_s]\), then
Lemma~\ref{ge:block-lemma} shows that
\[
    \omega_sI\preceq E_s\preceq2\omega_sI,
    \qquad
    \|T_s\|\le \frac{\omega_s}{2}.
\]
These bounds use only that the energies within a shell differ by at most a
factor of two. At the many-body level, the same decomposition takes the schematic form
\[
    \widetilde H
    =
    H_A
    +
    b^\dagger E b
    +
    a^\dagger T b
    +
    b^\dagger T^\dagger a.
\]
Here the fermionic modes \(a_\ell\) span the boundary \(\mathcal{A}\), the modes
\(b_j\) span the residual bath \(\mathcal{B}\), and
\[
    n_A:=\dim \mathcal{A}\le mS.
\]
The operator \(H_A\), which contains the full interaction \(V\), acts only
on the boundary modes.  The term \(b^\dagger E b\) describes free evolution
within the residual bath, while
\[
    a^\dagger T b+b^\dagger T^\dagger a
\]
contains the only couplings between the residual bath and the interacting
boundary. This decomposition suggests the central question of the proof:
\emph{which of the many residual-bath directions can actually influence the
interacting ground state, and how strongly can they do so?} To answer this, we use a decomposition of the residual bath closely related to Lanczos-based constructions used in numerical impurity solvers~\cite{busser2013lanczos,shirakawa2014block}.

\paragraph{Krylov depth.}
The residual bath \(\B_s=\mathcal H_s\ominus \A_s\) may still have dimension
\(\Theta(n)\), but its directions need not be equally relevant.  Since the
interaction acts only through the boundary \(\A_s\), a residual-bath mode can
influence the impurity only by communicating with that boundary. The map
\[
T_s^\dagger : \A_s \to \B_s
\]
identifies the bath directions that couple directly to \( \A_s\), namely
$\operatorname{Im}(T_s^\dagger)$.  Once inside the residual bath, further propagation is governed entirely by the free bath Hamiltonian \(E_s\).  It is therefore natural to organize \(\B_s\) according to how many applications of \(E_s\) are needed to generate
a new direction starting from \(\operatorname{Im}(T_s^\dagger)\). For an integer \(\ell\ge0\), define the block Krylov space~\cite{golub2013matrix}
\[
    \mathcal K_{s,\ell}
    :=
    \operatorname{span}
    \bigl\{
        E_s^jT_s^\dagger u:
        u\in \A_s,\ 0\le j<\ell
    \bigr\},
    \qquad
    \mathcal K_{s,0}:=\{0\}.
\]
They are subspaces of $\mathcal B_s$. In particular,
\[
    \mathcal K_{s,1}
    =
    \operatorname{Im}(T_s^\dagger)
\]
contains precisely the bath directions coupled directly to the boundary,
while \(\mathcal K_{s,2}\) additionally contains the new directions obtained
after one application of \(E_s\), and so on.  The genuinely new directions
appearing at step \(\ell+1\) form the orthogonal layer
\[
    {\mathcal W}_{s,\ell}
    :=
    \mathcal K_{s,\ell+1}
    \ominus
    \mathcal K_{s,\ell}.
\]
We call \(\ell\) the \emph{Krylov depth}.  It is an algebraic, rather than
spatial, notion of distance from the impurity: a direction at large depth is
one that can only be reached after repeatedly propagating the
boundary--bath coupling through the free bath Hamiltonian. Because the initial space
\(\operatorname{Im}(T_s^\dagger)\) has dimension at most \(m\), each new
Krylov layer also has dimension at most \(m\):
\[
    \dim {\mathcal W}_{s,\ell}\le m.
\]
Hence across all shells there are at most
$mS$ 
bath modes at any fixed depth.  More importantly, if
\(v\in {\mathcal W}_{s,\ell}\), then orthogonality to all shallower Krylov spaces gives
\[
    T_sE_s^jv=0,
    \qquad
    0\le j<\ell.
\]
The formal Krylov construction and these properties are given in
Section~\ref{sec:blockkrylov}. 
Thus a depth-\(\ell\) mode cannot couple to the boundary through fewer than $\ell$ applications of $E_s$.

This is exactly the structural information we need, but by itself it is
only a one-particle statement.  To obtain an algorithm we must convert
Krylov depth into a statement about the occupation of these modes in the
actual interacting many-body ground state.

\paragraph{From Krylov depth to ground-state occupation.}
Let \(\psi\) be a normalized ground state of \(\widetilde H\). For $v \in \mathcal{B}$ and $\norm{v}=1$, let $b(v) := \beta(v)$ where we write $b$ to emphasize that it belongs to the residual bath.
Since a fermionic mode can have occupation only \(0\) or \(1\), the operator $n_v:=b(v)^\dagger b(v)$ has eigenvalues \(0\) and \(1\). Therefore, for any state \(\psi\),
\[
\langle\psi|n_v|\psi\rangle
=
\langle\psi|b(v)^\dagger b(v)|\psi\rangle
=
\|b(v)\psi\|^2.
\]
Thus \(\|b(v)\psi\|^2\) is exactly the probability that the mode \(v\)
is occupied in the state \(\psi\), referred to as the occupation probability of the fermionic mode \(v\). It remains to understand how likely a given bath mode is to be occupied in the interacting ground state. 
 To this end,  we must keep track of all the ways in which the free bath Hamiltonian
can propagate \(v\) into a direction visible to the boundary.  This is captured by the imaginary-time evolution
\[
    e^{-tE}v
    =
    v-tEv+\frac{t^2}{2}E^2v-\cdots.
\]
Applying $T$ then extracts components that couple back into the boundary, i.e.,  $\|Te^{-tE}v\|_2$
measures the strength with which the bath direction \(v\), after evolving under the free residual-bath Hamiltonian for imaginary time \(t\), couples
to the impurity-facing boundary.
 This naturally leads us to define
\[
    s(v)
    :=
    \int_0^\infty
    \|Te^{-tE}v\|\,dt,
\]
which measures the total effective coupling of the bath direction \(v\) to
the boundary over imaginary time. The usefulness of the Krylov decomposition is that this quantity decays
exponentially with depth.  More precisely, Lemma~\ref{ge:depth} proves that
for every unit vector \(v\in\mathcal W_{s,\ell}\),
\begin{align}
\label{eq:boundons}
    s(v)
    \le
    \frac12\,3^{-\ell},
\end{align}
while \(s(v)=0\) for completely decoupled bath directions. The reason for this decay is already visible from the Krylov construction.
If \(v\in\mathcal W_{s,\ell}\), then
\[
    T_sE_s^jv=0,
    \qquad
    0\le j<\ell.
\]
Consequentially,
\[
T_se^{-tE_s}v
=
\sum_{j=0}^\infty
\frac{(-t)^j}{j!}T_sE_s^jv = \sum_{j \geq \ell} \frac{(-t)^j}{j!}T_sE_s^jv
\]
as the first \(\ell\) terms vanish.  Thus a depth-\(\ell\) mode can couple back
to the boundary only through terms involving at least \(\ell\) powers of
\(E_s\).  Using the fact that the energies within a shell differ by at most
a factor of two gives the exponential bound in Eq.~\eqref{eq:boundons}. The full calculation is
given in Section~\ref{sec:krylovdecay}.  Thus, Krylov depth has the desired operational meaning: bath modes
far from the boundary in Krylov depth are exponentially unlikely to be
occupied in the interacting ground state.

A one-mode bound, however, is not enough to control the full many-body state:
there are exponentially many Fock configurations, and a configuration may
contain several individually unlikely bath modes.  To truncate the Fock space,
we therefore need a bound on the probability that an arbitrary collection of
bath modes is occupied simultaneously.

\paragraph{Multiparticle pull-through.}
To convert the one-particle decay above into control of the many-body
ground state, we use the multiparticle pull-through lemma, see Lemma~\ref{ge:pull}. Intuitively, a pull-through identity moves bath annihilation operators through the imaginary-time evolution, replacing each of them by an integrated source acting on the impurity-facing boundary.  Let \(\psi\) be a normalized global ground state of \(\widetilde H\). 
We show in Lemma~\ref{ge:pull} that for any one-particle directions $v_1,\ldots,v_p$, 
\[
\|b(v_1)\cdots b(v_p)\psi\| \le \prod_{i=1}^p s(v_i).
\]
The proof is given in Section~\ref{sec:pullthrough}. Consequently, if $v_1,\ldots,v_p$ are orthonormal vectors such that $v_j$ belong to a Krylov layer of depth $\ell(j)$,  the corresponding number operators $n_j =b(v_j)^\dagger b(v_j)$  obey
\[
\left\langle\psi\middle|
    \prod_{j=1}^p n_j
\middle|\psi\right\rangle
\le
\prod_{j=1}^p s(v_j)^2\leq      \prod_{j=1}^p
\left(\frac14\,9^{-\ell(j)}\right),
\]
where the second inequality follows from Eq.~\eqref{eq:boundons}. 
Thus not only is each deep mode unlikely to be occupied; simultaneous
occupation of several deep modes is suppressed by the product of their
individual penalties. This product structure is what finally lets us control the exponentially
large set of Fock configurations.  Rather than truncating at a fixed Krylov
depth, we assign a cost to each occupied mode and bound the total cost.

\paragraph{Weighted Fock-space truncation.}
The joint occupation bound suggests a natural many-body truncation. For a layer mode at Krylov depth $\ell$, the occupation probability is suppressed exponentially in $\ell$, and the penalties multiply for simultaneous occupations.  Thus, after taking logarithms, the relevant quantity is an \emph{additive} cost over all occupied bath modes, rather than the maximum Krylov depth appearing in the configuration. 

Let $\{v_i\}_{i \in [n]}$ be an orthonormal one-particle basis over the boundary $\mathcal{A}$, all non-zero Krylov layers $\mathcal{W}_{s,\ell}$ and the decoupled remainder bath subspace. Let $J \subseteq [n]$ be the indices of the Krylov-layer vectors. For $j \in J$, let $\ell(j)$ denote the depth of the layer containing $v_j$ and let $n_j := b(v_j)^\dagger b(v_j)$ denote its occupation. 
Let $U$ be the Bogoliubov transformation mapping the original Fock
basis $\{\ket{x}\}$ to the occupation basis of the constructed modes.
Thus $U\ket{x}$ has occupation $x_j$ in the mode associated with
$v_j$.
We then define the weighted occupation operator
\[
Q := \sum_{j\in J} (\ell(j)+1)n_j,
\]
where the extra $+1$ also penalizes the number of bath excitations, so that many simultaneous occupations at shallow depth are not retained for free.  Boundary occupations are left completely unrestricted.
 
For a Fock configuration $x\in\{0,1\}^n$ in this basis, where $x_j$ is the occupation of $v_j$, let
\[
Q(x) := \sum_{j\in J}(\ell(j)+1)x_j,
\]
and for an integer cutoff $R\ge 0$ define
\[
\mathcal{X}_R
:=
\left\{
    x\in\{0,1\}^n:
    Q(x)\le R,\quad
    x_j=0\ \text{for every decoupled mode }j
\right\}.
\]
Let $P_R$ project onto the span of $\{U\ket{x}:x\in \mathcal{X}_R\}$, that is,
\[
 P_R=U\left(\sum_{x\in \mathcal{X}_R}\ket{x}\bra{x}\right)U^\dagger,
\]
and write $    p_R:=\|(I-P_R)\psi\|^2,$ and 
$    D_R:=|\mathcal{X}_R|=\operatorname{rank}(P_R).$ 
Thus $p_R$ is the ground-state probability mass discarded by the
truncation, while $D_R$ is the dimension of the retained many-body
space. In Section~\ref{sec:weightfoctruncation}, we give the two complementary estimates
\begin{align}
\label{eq:boundonpR}
    p_R\le e^{mS}3^{-R},
    \qquad
    D_R\le e^{(\log 2+1/2)mS}3^R.
\end{align}
The first follows by bounding an exponential moment of $Q$ using the 
joint occupation estimates and then applying Markov's inequality.
The second is the corresponding weighted counting argument for the
number of configurations satisfying $Q(x)\le R$. The important point is that we do \emph{not} keep all configurations
supported on the first $R$ Krylov layers.  Such a cutoff would ignore
the additional suppression of configurations containing many
simultaneous excitations.  Instead, we constrain their total weighted
Krylov cost.  Hence increasing $R$ exponentially suppresses the
discarded ground-state mass while increasing the retained dimension
by only an exponential factor in $R$.  This is the main compression
step of the proof.

\paragraph{From norm truncation to energy truncation.}
The fact that $p_R$ is small does not by itself imply that the energy
changes little, since $\|\widetilde H\|$ may grow with $n$.  The
particular truncation above has an additional useful property: it does
not restrict the boundary modes.  Since the interaction $V$ is
supported entirely on the boundary, $[P_R,V]=0$. 

Let
$   N=\sum_j \beta_j^\dagger \beta_j
$  
be the total quasiparticle-number operator in the basis used above.
Every retained configuration contains at most $mS$ boundary particles
and at most $R$ Krylov-layer particles, because every occupied layer
mode contributes at least one unit to $Q$.  Therefore
\[
    \|NP_R\|\le mS+R.
\]
Since the rounded one-particle energies satisfy
$    0\le \widetilde e_j\le 1,$ 
we have $   0\preceq \widetilde H_0\preceq N,
$ and 
$ \|\widetilde H_0P_R\|\le mS+R.$  Now define the normalized projected state
\[
    \phi_R
    :=
    \frac{P_R\psi}{\sqrt{1-p_R}}.
\]
Using the ground-state equation together with $[P_R,V]=0$ gives
\[
    0
    \le
    \langle\phi_R|\widetilde H|\phi_R\rangle
    -
    E_1(\widetilde H)
    \le
    \frac{\sqrt{p_R}}{1-p_R}(mS+R).
\]
This is proved in Section~\ref{sec:groundenergyapprox}.  Notice in
particular that the bound is independent of $\|V\|$: the interaction
term cancels exactly because $P_R$ commutes with $V$. We now choose the cutoff exactly as in the formal proof, $R=O(mS+\log 1/\varepsilon)$
 and plugging this with the upper bound on $p_R$ in Eq.~\eqref{eq:boundonpR}, we have
\[
    0
    \le
    \langle\phi_R|\widetilde H|\phi_R\rangle
    -
    E_1(\widetilde H)
    \le {\varepsilon}/{4}.
\]
Similarly plugging this into Eq.~\eqref{eq:boundonpR}, we also have $D_R\leq \exp(O(m\log (m/\varepsilon)))$.
 Thus the original $2^n$-dimensional ground-energy problem has been
reduced to an explicitly defined subspace whose dimension is
$(m/\varepsilon)^{O(m)}$, and hence polynomial in $1/\varepsilon$ for
constant impurity size $m$.
\paragraph{The algorithm.}
It remains to show that the retained space and the restriction of the
Hamiltonian to it can be constructed efficiently.  This is carried out
in Section~\ref{ge:algorithm}.

Set $\omega=\varepsilon/(4m)$, $R=O(mS+\log 1/\varepsilon)$. 
Starting from the compact description of the input Hamiltonian, we
construct the rounded bath energies, the dyadic shells and their
boundary spaces, and the full ordered basis of boundary, Krylov-layer,
and decoupled modes. Only boundary modes and Krylov layers of depth
less than $R$ can be occupied in the retained space; all other modes
are required to be empty.
The full basis and the induced action of $U$ on the modes can be
computed in polynomial time in $n$ and the relevant parameters.
We enumerate the configurations in $\mathcal{X}_R=\{x_1,\ldots,x_{D_R}\}$
and form the compression of the original Hamiltonian,
\[
    B_{\mu\nu}
    :=
    \bra{x_\mu}U^\dagger HU\ket{x_\nu} .
\]
The quadratic contribution is obtained directly from the retained
matrix elements of the original one-particle Hamiltonian $h$.
The interaction $V$ acts only on the
boundary and is given by its compact impurity description, so its
matrix elements can also be computed without ever constructing the
full $2^n\times2^n$ Hamiltonian.  Importantly, the entire retained
basis is determined from the input Hamiltonian alone; the unknown
ground state $\psi$ is used only in the analysis.

Finally, we compute a normalized approximate lowest eigenvector $z$
of $B$, with $z^\dagger Bz\le E_1(B)+\varepsilon/4$.
Here $E_1(B)$ is the smallest eigenvalue of $B$.
Define the sparse state in the original occupation basis by
\[
 \ket\phi=\sum_{\mu=1}^{D_R}z_\mu\ket{x_\mu},
 \qquad U\ket\phi=\sum_{\mu=1}^{D_R}z_\mu U\ket{x_\mu},
\]
and return $U$, $\phi$, and
\[
 \widehat E_1=z^\dagger Bz
 =\bra\phi U^\dagger HU\ket\phi.
\]
The projected ground state $\phi_R$ is a variational witness for $B$.
Using $H\preceq\widetilde H$, the truncation bound, and the rounding
bound gives
\begin{align*}
 E_1(H)\leq E_1(B)
 &\leq\bra{\phi_R}H\ket{\phi_R}
 \leq\bra{\phi_R}\widetilde H\ket{\phi_R}\\
 &\leq E_1(\widetilde H)+\varepsilon/4
 \leq E_1(H)+\varepsilon/2.
\end{align*}
The eigensolver bound therefore implies
\[
 E_1(H)\leq\widehat E_1
 \leq E_1(H)+3\varepsilon/4\leq E_1(H)+\varepsilon,
\]
as claimed. The output $\phi$ is supported on $D_R$ original Fock
configurations, while $U\phi$ is the low-energy physical state.

Since $S=O(\log(m/\varepsilon))$, $R=O(m\log(m/\varepsilon))$, and $D_R
    =
    \exp\!\left[
        O\!\left(
            m\log {m}/{\varepsilon}
        \right)
    \right],$
all one-particle linear-algebra operations are polynomial in $n$ and
the relevant parameters, while constructing and diagonalizing the
projected many-body Hamiltonian costs polynomial time in $D_R$ and
$1+\|V\|$.  Section~\ref{ge:algorithm} gives the detailed arithmetic
complexity and yields the total runtime
\[
    \poly(n,1+\|V\|)
    \exp\!\left[
        O\!\left(
            m\log\frac{m}{\varepsilon}
        \right)
    \right].
\]
For constant impurity size $m$, this is polynomial in both $n$ and
$1/\varepsilon$.

\subsubsection{Gaussian spanning set for the TFD state}

Next we sketch the proof of Theorem~\ref{thm:rank}, which constructs
a small Gaussian spanning set approximating the thermofield-double
state. The proof uses the Krylov decomposition and weighted truncation
ideas from the ground-energy proof, but applies them to excitations
relative to the TFD state of a certain free-fermion Hamiltonian.
 We first explain the choice of
annihilation operators that makes this possible.
Below we write $\rho_\beta(H)$ and 
$\ket{\mathsf{TFD}_\beta(H)}$ for the Gibbs state and the thermofield double state at the inverse temperature $\beta$ associated with a Hamiltonian $H$.

\paragraph{Free TFD state as a vacuum.}
Consider first a single fermionic mode $b,b^\dagger$ with a Hamiltonian
$\lambda b^\dagger b$, where $\lambda\ge0$.
The corresponding Gibbs state  at the inverse temperature $\beta$ is 
\[
\rho_\beta(\lambda b^\dagger b) = p_0 |0\rangle \langle 0| + p_1 |1\rangle \langle 1|
\]
where the occupation probabilities are 
\[
p_0 = \frac1{1+e^{-\beta\lambda}}   \quad \mbox{and} \quad p_1= \frac{e^{-\beta\lambda}}{1+e^{-\beta\lambda}}
\]
The corresponding TFD state  is defined as 
\[
 \ket{\mathsf{TFD}_\beta(\lambda b^\dag b)}
 =\sqrt{p_0} \ket{00}-i\sqrt{p_1}\ket{11} \in {\cal F}_2.
\]
Here the second mode is an ancilla (environment) used for purifying the first (physical) mode.
The extra phase shift $-i$ follows 
the fermionic purification convention of Section~\ref{sec:tfd-preliminaries} and does not affect the physical Gibbs state.
By definition,  the Gibbs state is obtained from the TFD state by tracing out the ancillary mode.
Let $\widetilde b,\widetilde b^\dagger$ be the annihilation and creation operators on the ancillary mode.
The operators $b,\widetilde b$ obeys the canonical anticommutation rules. In particular, $b \widetilde b = - \widetilde b b$.

Define a pair of operators
\[
 \mathsf d_-=\sqrt{p_0}\, b+i\sqrt{p_1}\, \widetilde b^\dagger
 \quad\mbox{and}\quad
 \mathsf d_+=\sqrt{p_1}\, b^\dagger+i\sqrt{p_0}\, \widetilde b.
\]
One can check that $ \mathsf d_\pm$  annihilate the TFD state,
\[
 \mathsf d_\pm\ket{\mathsf{TFD}_\beta(\lambda b^\dag b)}=0.
\]
Besides, $\mathsf d_\pm$
satisfy the canonical anticommutation relations, that is,
\[
\{ \mathsf d_\eta, \mathsf d_{\eta'}\}=0  \quad \mbox{and} \quad 
\{ \mathsf d_\eta, \mathsf d^\dagger_{\eta'}\} =\delta_{\eta,\eta'} I.
\]
Thus the TFD state is equivalent to the vacuum state $|00\rangle$ modulo a basis change described
by a Bogoliubov transformation mapping $b,\widetilde b$ to $\mathsf d_+,\mathsf d_-$.
Informally, this basis change removes all  thermal fluctuations.

This construction naturally extends to any free fermion Hamiltonian.
Indeed, any such Hamiltonian can be diagonalized by a Bogoliubov transformation
and, up to additive scalar,  takes a form $H_{\mathrm{free}}=\sum_{j=1}^n \lambda_j b_j^\dagger b_j$,
where $\lambda_j\ge0$ and $b_1,\ldots,b_n$ are annihilation operators for $n$ fermionic modes.
The corresponding Gibbs state $\rho_\beta(H_{\mathrm{free}})$ is a product of $n$ single-mode Gibbs states
considered above and the  state $\ket{\mathsf{TFD}_\beta(H_{\mathrm{free}})}$
is defined by purifying each mode independently, see Section~\ref{sec:tfd-preliminaries} for details.
The state $\ket{\mathsf{TFD}_\beta(H_{\mathrm{free}})}$  is annihilated by $2n$  operators $\mathsf d_{\pm,j}$ such that 
$\mathsf d_{-,j}$ is a linear combination of $b_j$ and the corresponding ancillary mode operator $\widetilde b_j^\dagger$,
while $\mathsf d_{+,j}$ is a linear combination of $b_j^\dagger$ and $\widetilde b_j$.
The operators $\mathsf d_{\pm,j}$ with $j=1,\ldots,n$
obey the canonical anticommutation rules
and thus define a basis in which the state $\ket{\mathsf{TFD}_\beta(H_{\mathrm{free}})}$ becomes the vacuum $|0^{2n}\rangle$.
Informally, this basis change removes all  thermal fluctuations {\em in the absence of interactions}. 
We refer to the  operators  $\mathsf d_{\pm,j}$ as \emph{free TFD annihilators}.

\paragraph{Exponential suppression of excitations.}
We express the interacting TFD
state $\ket{\mathsf{TFD}_\beta(H)}$ in the basis determined by the free TFD annihilators $\mathsf d_{\pm,j}$ associated with the residual bath Hamiltonian
$b^\dag E b$
used  in the ground state algorithm (with a few technical modifications).
The residual bath contains $n_B\le n$ modes, where $n_B$ depends on details of the Krylov decomposition.
These details are not essential for this high-level overview. 
The number operators $\mathsf d_{\pm,j}^\dagger\mathsf d_{\pm,j}$ with $j=1,\ldots,n_B$
count only  excitations caused by the interaction $V$ relative to the free TFD.
The crux of the proof is to control these  excitations and showing that joint occupations of modes at large Krylov depth
are exponentially suppressed. 
To this end we consider annihilators
\[
 \mathsf d_-(v)=\sum_{j=1}^{n_B}
                   \overline{\langle v_j,v\rangle}\mathsf d_{-,j}
 \quad \mbox{and} \quad
 \mathsf d_+(v)=\sum_{j=1}^{n_B}
                   \langle v_j,v\rangle \mathsf d_{+,j},
\]
where $v\in \B$ is an arbitrary vector in the one-particle space describing the residual bath
and $v_1,\ldots,v_{n_B}\in \B$ is the eigenbasis of $E$.
By definition, the operators $\mathsf d_{\pm}(v)$ annihilate the free TFD state $\ket{\mathsf{TFD}_\beta(b^\dag E b)}$.
We prove a thermal analogue of the multiparticle pull-through lemma showing that 
multiparticle excitations associated with the modes $\mathsf d_{\pm}(v)$ and the {\em interacting} TFD state are exponentially suppressed.
More formally,
\[
 \norm{\mathsf d_{\eta_1}(v_1)\cdots
       \mathsf d_{\eta_k}(v_k)\ket{\mathsf{TFD}_\beta(H)}}
 \leq\prod_{i=1}^k s_\beta(v_i)
\]
for arbitrary vectors $v_i\in\B$ and signs $\eta_i\in\{+,-\}$, where
\[
 s_\beta(v)=\int_0^{\beta/2}\norm{Te^{-tE}v}\,dt
\]
is a finite-temperature analogue of the integrated coupling used in the ground state algorithm. 
One technical difference is that the map $T$ coupling the residual bath $\B$ to the boundary $\A$ now
contains an extra thermal factor,
\[
T=h_{\A\B}(1+e^{-\beta E})^{-1/2}.
\]

For $\beta>0$, we
partition the bath spectrum into $S=O(\log(2+\beta))$ energy shells:
the low-energy interval $[0,\min\{1,\beta^{-1}\}]$ and disjoint
dyadic intervals above it.
The Krylov decomposition gives at most $m$ directions per shell
at each depth $\ell$, with
$s_\beta(v)\leq\frac12\,3^{-\ell}$ for every unit depth-$\ell$
vector. Each direction gives two free TFD modes, so there are at
most $2mS$ such modes at each depth. For any subset of distinct canonical modes
with number operators $\mathsf n_j$ and depths $\ell(j)$,
the pull-through lemma implies
\[
 \bra{\mathsf{TFD}_\beta(H)}\prod_j\mathsf n_j
       \ket{\mathsf{TFD}_\beta(H)}
 \leq\prod_j\left(\frac14\,9^{-\ell(j)}\right).
\]

\paragraph{Weighted truncation.}
Finally, we leave boundary occupations unrestricted and assign
cost $\ell+1$ to each occupied free TFD mode of depth $\ell$.
Together with the physical and ancillary boundary modes, the
free TFD modes define a Fock basis $\{U_\beta\ket{x}\}$, where $U_\beta$
is a Bogoliubov transformation. 
Projecting $\ket{\mathsf{TFD}_\beta(H)}$ onto the span of
configurations of total cost at most $R$, with all decoupled modes
empty, discards squared norm at most $e^{4mS}4^{-R}$ and retains
at most $e^{O(mS)}2^R$ configurations. Choosing
$R=4mS+\lceil\log_2(2/\delta)\rceil$ therefore gives discarded
squared norm at most $\delta^2$, and the retained configuration
set satisfies
\[
 |\mathcal X_{\beta,\delta}|
 \leq\delta^{-1}(2+\beta)^{O(m)}.
\]
The transformation $U_\beta$ is computed from $H_0$, the impurity
support, and $\beta$; selecting $\mathcal X_{\beta,\delta}$
additionally uses $\delta$. Both constructions are independent
of $V$, and the approximation holds without any bound on
$\norm{V}$. At $\beta=0$, the TFD is itself a Gaussian state
independent of $H$, so a single spanning vector suffices.

\subsubsection{Constructive approximation of the TFD state}

We now sketch the proof of Theorem~\ref{ct:main}, which computes
$\ket{\mathsf{TFD}_\beta(H)}$ as a linear
combination of states in the Gaussian spanning set of
Theorem~\ref{thm:rank}, assuming $\norm{V}\le1$.
The idea is to simulate Trotterized  imaginary time evolution 
of  $\ket{\mathsf{TFD}_\tau(H)}$ over the interval $\tau \in [0,\beta]$
starting from  $\tau=0$.
The initial state $\ket{\mathsf{TFD}_0(H)}$
is an explicitly known Gaussian EPR-like state. 
We adapt the proof of Theorem~\ref{thm:rank} to obtain
one $\tau$-independent Krylov decomposition.
At each time $\tau$, the corresponding free TFD annihilators
define a $\tau$-dependent Bogoliubov transformation $\mathcal U_\tau$
and a Fock basis $\{ \mathcal U_\tau |x\rangle\}$  such that 
$\ket{\mathsf{TFD}_\tau(H)}$ 
has almost all its mass in the span of low-weight Fock configurations (weighted by the Krylov depth)
in the basis determined by $\mathcal U_\tau$.
We truncate all Fock configurations with the weight above certain cutoff $R$.
The truncated version of $\ket{\mathsf{TFD}_\tau(H)}$ 
can be represented  succinctly by its coefficients in the chosen basis. 
We can also  efficiently compute matrix elements of $H$ needed 
to simulate each Trotter step, including 
the contribution from the time dependence of the basis.

The main challenge  is that imaginary  time evolution can amplify
errors. If we simply evolve and truncate at each Trotter step, the
per-step truncation errors
can accumulate into a large final error growing exponentially with $\beta$.
We avoid this problem by fitting the entire trajectory using
least squares. We fix the known initial state and determine
the coefficients at all subsequent Trotter  steps simultaneously,
minimizing the sum of squared errors in the Trotterized 
evolution. These errors include the components
outside the retained Gaussian span. The exact  trajectory  $\ket{\mathsf{TFD}_\tau(H)}$,
after truncation, provides a good candidate for this fit with a small total residual. 
Therefore, the trajectory obtained by solving the least-squares problem also has a small total residual
and, in contrast to the exact trajectory, can be computed efficiently.

The key technical step is to prove that this small
residual implies a small error in the final state.
We show that the residuals lie in a small, explicitly computable enlargement
of the retained Gaussian span. We use H\"older's inequality for Schatten norms
to control
how much subsequent imaginary-time evolution can amplify
vectors in this enlarged span, while the thermal multiparticle pull-through lemma
makes the truncation errors small enough to compensate
for that amplification. We also need a computable rescaling
of the trajectory: an approximate ground energy, supplied by
Theorem~\ref{thm:ground-energy-algorithm}, and the free-bath
partition function keep the vector norms controlled without
requiring the unknown interacting partition function.
The rescaled trajectory satisfies a linear evolution equation,
so the fit is an ordinary linear least-squares problem.
Together,
these estimates give an algorithm that computes the final low-rank Gaussian approximation
of $\ket{\mathsf{TFD}_\beta(H)}$
 in time polynomial in $n$, $1+\beta$,
and $\delta^{-1}$ for fixed impurity size $m$.

\subsection{Technical overview II: time evolution}

In this subsection, we overview the underlying construction that proves \Cref{thm:main-reduction} and hence \Cref{cor:bqp}.  Let \(\mathcal C\) be a quantum
circuit on \(k\) qubits with \(g\) gates.  Our goal is to encode the 
computation of \(\mathcal C\) into the time evolution of a
time-independent quantum impurity Hamiltonian \(H=H_0+V\) with constant
impurity size. 
More precisely, we construct a time-independent Hamiltonian
$    H=H_0+V,$ 
an explicitly specified initial state \(\ket{\Psi}\), an evolution time
\(T\), and a fixed fermionic mode \(o\), such that
\[
    \bra{\Psi}
    e^{iHT}a_o^\dagger a_o e^{-iHT}
    \ket{\Psi}\approx p_{\mathcal C},
\]
where \(p_{\mathcal C}\) is the output probability of the original
circuit.  Thus, estimating the occupation of a single fermionic mode
after time evolution allows us to recover the output of the quantum
computation.
The main point is that
\(V\) is supported on only a constant number of modes, while both the
total number \(n\) of fermionic modes and \(T\) are polynomial in \(k+g\).

\begin{figure}[H]
\centering
\IfFileExists{figures/processor_architecture.pdf}{\includegraphics[width=0.98\linewidth]{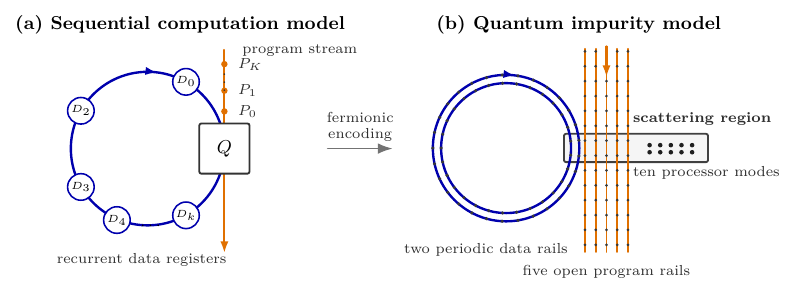}}{\fbox{\parbox{0.9\linewidth}{\centering Processor architecture figure unavailable.}}}
\caption{The two stages of the construction.  We first rewrite the
quantum circuit as a sequence of interactions with the same
constant-size system.  We then realize this computation using the time
evolution of a quantum impurity model.}
\label{fig:processor-architecture}
\end{figure}

The construction has two parts as described by the figure above. 
We first rewrite the circuit in what we call a ``sequential computational model'', i.e.,  its gates are implemented one at a time through repeated interactions with the same constant-size system.  
We then realize these interactions using the time evolution generated by \(H\). 
We summarize the reduction in five steps, illustrated in the following schematic diagram and in \Cref{fig:processor-architecture}.
\begin{equation}
\begin{gathered}
\text{quantum circuit}\\
\bigg\downarrow\;\text{\small(steps 1,2)}\\
\text{sequential computation model\;(\Cref{fig:processor-architecture}(a))}\\
\bigg\downarrow\;\text{\small(steps 3,4,5)}\\
\text{time evolution of impurity model\;(\Cref{fig:processor-architecture}(b))}
\end{gathered}
\end{equation}

The purpose of the Steps $(1,2)$ is to separate the computation itself
from the fermionic translation. 
We define a small processor that handles one register at a time, and then show how to realize any quantum circuit by choosing the sequence of program registers appropriately. After Step~2, this sequential computation
model is observed to be already universal for BQP; Steps $(3-5)$ only show how to realize this universal sequential model as the time evolution of a quantum impurity
Hamiltonian (with constant impurity size). 

 \paragraph{Step 1: sequential computation model.}

We first rewrite the circuit so that the computation is carried out by repeatedly
using the same small processor. The \(k\) qubits of the circuit, denoted
\(D_1,\ldots,D_k\), are presented to the processor one at a time in a cyclic
order. Before each data qubit arrives, a program register specifies which
operation should be performed next.
 The processor has two parts. An auxiliary qubit \(A\) is used to mediate
two-qubit gates, while a five-state register \(C\) stores the current instruction.
We write
$    Q=A\otimes C$ 
for the processor. A visiting register is either a program register \(P\),
carrying one of five possible instructions, or a data register \(D\), carrying
one qubit. On a program visit, the instruction in \(P\) is loaded into \(C\);
on the following data visit, the operation stored in \(C\) is applied to \(D\)
and, when necessary, to \(A\). 
Each program register $P$ is a 5-dimensional qudit storing one of five commands (which will correspond to a gate):
\begin{equation}
c\in\{0,1,2,3,4\}
\quad\longleftrightarrow\quad
G_c\in\{I,H_D,T_D,\operatorname{SWAP}_{AD},\operatorname{CZ}_{AD}\}.
\end{equation}
Importantly, the program registers travel along an open path, whereas the $k$ data registers move cyclically in the order $D_1,D_2,\ldots,D_k,D_1,\ldots$.

Initially, the program registers will be programmed into a product state, and everything else starts at $\ket{0}$.
Then the program and data registers will visit the processor in an alternating order:
\begin{equation}
    P_1,D_1,P_2,D_2,\cdots,
\end{equation}
with the indices for $D$ understood as mod $k$.
The interaction between the processor and the visiting register $P$ or $D$ is designed in a way that:
\begin{itemize}
    \item on a program visit it swaps the command stored in the program register $P$ into $C$;
    \item on the following data visit it applies the gate controlled by $C$ on $A$ and $D$. 
\end{itemize}
As a result, the computational model will execute the sequence of gates on the data qubits as specified by the program registers $P_1, P_2, \cdots$. There are two possible updates, depending on whether the visiting register is a
program register or a data register.  A program visit loads the new command
into \(C\), while a data visit applies the command currently stored in \(C\).
Accordingly, define
\begin{align}
U_*^{P}\bigl(\ket{a,c}_Q\ket{p}_P\bigr)
 &=\ket{a,p}_Q\ket{c}_P,\\
U_*^{D}\bigl(\ket{a,c}_Q\ket{d}_D\bigr)
 &=G_c\bigl(\ket{a}_A\ket{d}_D\bigr)\otimes\ket{c}_C.
\end{align}
For later convenience, we combine the two possible types of visiting
registers into the seven-dimensional space
\begin{equation}
X=P\oplus D,
\end{equation}
so that at each step the system undergoes an identical, fixed unitary $U_*=\mathrm{diag}\{U_*^{P},U_*^{D}\}$ acting on~$Q\otimes X$.

\paragraph{Step 2: encoding circuits into programs.}

Given a circuit $\mathcal C$, we decompose it into a sequence of elementary gates $\{H,T,\operatorname{CZ}\}$.
For every two-qubit gate, for example $\operatorname{CZ}_{D_i,D_j}$, we further transform it by the identity
\begin{equation}
I_A\otimes\operatorname{CZ}_{D_i,D_j} = \operatorname{SWAP}_{A,D_i}\operatorname{CZ}_{A,D_j}
\operatorname{SWAP}_{A,D_i}.
\end{equation}
Consequently, after the transformation, all gates are either one-qubit gates ($H$ or $T$) on $D$ or two-qubit gates ($\operatorname{SWAP}$ or $\operatorname{CZ}$) acting between $A$ and $D$.
With the cyclic nature of the data registers in mind, we pad the sequence with identity gates until desired registers arrive.
With the initialization specified in step 1, our computational model will produce $\mathcal C\ket{0^k}$ on the data registers.
We then use a final $\operatorname{SWAP}$ to swap the designated output qubit into $A$.  
Hence measuring $A$ in state $\ket{0}$ gives exactly $p_{\mathcal C}$ in \Cref{eq:qc}.

\paragraph{Step 3: realizing data and program registers by fermions.}

We use five one-dimensional chains of fermionic modes, labelled by the five command values $p\in\{0,1,2,3,4\}$.
A program register storing command $p$ is represented by a single fermion in a spatially localized wave packet on chain $p$.
The wave packets representing different program registers live on these same five chains, but are spatially separated according to their position in the program stream. The data registers are handled in the same way, with two one-dimensional rings in place of the five chains. See \Cref{fig:processor-architecture} for an illustration.

The spatial motion of these packets are governed by $H_0$, realizing the ordinary nearest-neighbor hopping (in the first-quantization picture):
\begin{equation}
    \sum_x \left(\ket{x+1}\bra{x}+\ket{x}\bra{x+1}\right)
\end{equation}
with appropriate boundary conditions specified above. 
We will choose the wave packet profiles so that, approximately, every packet travels along its chain or ring as an isolated and rigid particle with a common, constant speed.
Namely, the packets neither interact with one another nor broaden appreciably over the course of the evolution.

\paragraph{Step 4: implementing $U_*$ through impurity interaction.}

We encode the processor by one fermion in 10 modes. Different from the data and program wave packets, the processor modes are placed at a specific, fixed location, for which we choose to be the $0 \leftrightarrow 1$ edge.

We define a quartic impurity-type interaction as follows.
Let $a_q$ be the annihilation operator for the $q$-th processor mode ($q\in \{0,1\}_A\times\{0,1,2,3,4\}_C$).
Denote $a_{x,\alpha}$ be the annihilation operator for the mode at position $x$ and state $\alpha\in\{0,1\}_D\cup\{0,1,2,3,4\}_P$.
We define
\begin{equation}
V_{\mathcal I}=
\sum_{q,q',\alpha,\beta}
u_{q'\alpha,q\beta}
a_{1,\alpha}^\dagger a_{q'}^\dagger a_q a_{0,\beta}
+\mathrm{h.c.},
\end{equation}
where 
\begin{equation}
u_{q'\alpha,q\beta}
=\matrixel{q',\alpha}{U_*}{q,\beta}.
\end{equation}
As alluded from the above equation, the impurity interaction is designed such that it applies $U_*$ to the processor and a passing packet whenever the packet crosses the processor.
The interaction $V_{\mathcal I}$ is number conserving and quartic.
It acts only on the ten processor modes and the seven register modes at each endpoint of the $0\leftrightarrow 1$ edge, so its support has size
\begin{equation}
|\mathcal I|=10+(5+2)\times 2=24.
\end{equation}
We call the support of these modes the scattering region.
Away from this region, the Hamiltonian is purely quadratic, realizing the nearest-neighbor hopping.

\paragraph{Step 5: Initialization and time evolution.}

The initial state \(\ket{\Psi}\) encodes the entire computation.  It contains
one fermion for the processor, one for each data register, and one for each
program register, arranged as an explicitly specified Slater determinant.
The internal state of each program fermion specifies which operation is to be
performed, while its initial position determines when that operation reaches
the processor.  Likewise, the data fermions are positioned so that the data
qubits arrive in the required cyclic order.

Thus the gate sequence is not implemented by changing the Hamiltonian in
time.  Instead, it is built into the initial locations and internal states of
the fermions.  Under the quadratic evolution \(H_0\), these fermions move
toward the interaction region and cross it in the order prescribed in
Step~1.
We will show that the evolution of this state under the impurity Hamiltonian is well approximated by an idealized evolution in which the wave packets propagate without spreading and each passage through the processor applies $U_*$. 
This idealized evolution executes the computational model described above.
As a result, we will show
\begin{equation}
\abs{\expval{e^{iHT} a_{(0,0)}^\dagger a_{(0,0)} e^{-iHT}}{\Psi_{\mathcal C}^{\mathrm{in}}}- p_{\mathcal C}}\le \frac{1}{12},
\end{equation}
while \(n,T=O((k+g)^4)\).

\section{Preliminaries}

\subsection{Fermionic systems}
\label{sec:fermionic-systems}
We begin by recalling the definition of a finite-dimensional Fermionic
system.  A system of $n$ Fermionic modes lives in the \emph{particle} space
$\mathbb{C}^n$, with one basis vector for each available \emph{mode}.  A mode can
be either empty or occupied by one fermion.  The full Hilbert space is called the
$2^n$-dimensional Fermionic Fock space
\[
    \mathcal{F}_n
    :=
    \bigoplus_{r=0}^n \bigwedge^r \mathbb{C}^n,
\]
where $r$ above denotes the total number of occupied modes among the $n$ modes. 
Equivalently, $\mathcal{F}_n$ has
the orthonormal occupation-number basis
\[
    \bigl\{
        \lvert x_1,\ldots,x_n\rangle
        : x_j\in\{0,1\}
    \bigr\}.
\]
Here $x_j=1$ means that mode $j$ is occupied. 
 We often refer to the state $\lvert 0\rangle=\lvert 0,\ldots,0\rangle$ as the vacuum state, and let
$a_j^\dagger$ and $a_j$ denote the \emph{creation and annihilation} operators
for mode $j$.  These operators are known to satisfy the anticommutation relations
\[
    \{a_j,a_k^\dagger\}=\delta_{jk}I,
    \qquad
    \{a_j,a_k\}=0,
    \qquad
    \{a_j^\dagger,a_k^\dagger\}=0,
\]
where $\{A,B\}=AB+BA$. The above in particular implies that 
$(a_j^\dagger)^2=a_j^2=0$, i.e., a mode cannot be occupied twice. Furthermore, we use the convention
\[
    \lvert x_1,\ldots,x_n\rangle
    =
    (a_1^\dagger)^{x_1}\cdots
    (a_n^\dagger)^{x_n}\lvert 0\rangle.
\]
The anticommutation relations induce natural signs when writing out these operator relations. For $j\in [n]$, let $    \eta_j(x):=\sum_{k<j}x_k$. The action of the creation and annihilation operators is given by
\begin{align}
    a_j\lvert x\rangle
    &=
    (-1)^{\eta_j(x)}x_j
    \lvert x_1,\ldots,x_{j-1},0,x_{j+1},\ldots,x_n\rangle,  \\   a_j^\dagger\lvert x\rangle
    &=
    (-1)^{\eta_j(x)}(1-x_j)
    \lvert x_1,\ldots,x_{j-1},1,x_{j+1},\ldots,x_n\rangle.
\end{align}
Thus $a_j$ removes the fermion in mode $j$, when one is present, but
acquires a sign determined by the parity of the occupied modes preceding
$j$. As always, we can canonically work in either the particle basis or the Fock space basis as long as we preserve the operations consistently. 

\paragraph{Majorana operators.} Apart from the creation--annihilation formalism, it is often convenient to use Majorana operators, which treat hopping and pairing terms uniformly. Associated with the chosen fermionic modes are the canonical \emph{Majorana basis operators}
\[
    \gamma_{2j-1}=a_j+a_j^\dagger,
    \qquad
    \gamma_{2j}=-i(a_j-a_j^\dagger).
\]
equivalently, 
\[
    a_j=\frac{\gamma_{2j-1}+i\gamma_{2j}}{2},
    \qquad
    a_j^\dagger
    =\frac{\gamma_{2j-1}-i\gamma_{2j}}{2},
\]
 The Majorana basis
operators are Hermitian and satisfy
\[
    \{\gamma_p,\gamma_q\}=2\delta_{pq}I,\quad     \gamma_p^2=I,
    \qquad
    \gamma_p\gamma_q=-\gamma_q\gamma_p
    \quad\text{for }p\neq q.
\]
\subsection{Fermionic Hamiltonians}\label{sec:bogoliubov}
We will work with quadratic fermionic Hamiltonians, namely
Hermitian linear combinations of monomials involving at most two creation
and annihilation operators. A general parity-preserving quadratic
Hamiltonian can be written as
\begin{equation}\label{def:quadratic_H_parity_preserving}
H_0
=
E_0 I
+
\sum_{j,k=1}^n h_{jk}a_j^\dagger a_k
+
\frac12\sum_{j,k=1}^n
\left(
    \Delta_{jk}a_j^\dagger a_k^\dagger
    +
    \overline{\Delta_{jk}}a_k a_j
\right),
\end{equation}
where $h=h^\dagger$ and
$\Delta=-\Delta^{\mathsf T}$. A Hamiltonian is parity preserving if it
commutes with $P=(-1)^N$, where $N=\sum_j a_j^\dagger a_j$ is the
total particle-number operator.  Equivalently, it does not mix states
with even and odd particle number, although it may change the particle
number by an even amount. Here $E_0 I$ is an overall energy
shift.  In the number-preserving term
\[
    \sum_{j,k=1}^n h_{jk}a_j^\dagger a_k,
\]
the diagonal coefficient $h_{jj}$ is the on-site energy of mode $j$.
For $j\neq k$, the operator $a_j^\dagger a_k$ moves a fermion from
mode $k$ to mode $j$; the coefficient $h_{jk}$ is its hopping
amplitude, and $\lvert h_{jk}\rvert$ is the corresponding hopping
strength.  The terms proportional to $\Delta_{jk}$ create or annihilate
pairs of fermions. In Majorana notation, a general parity-preserving quadratic Hamiltonian
takes the compact form
\[
    H_0
    =
    E_0 I
    +
    \frac{i}{4}
    \sum_{p,q=1}^{2n}
    A_{pq}\gamma_p\gamma_q,
\]
where $A\in\mathbb{R}^{2n\times 2n}$ is antisymmetric.  Thus the
occupation, hopping, and pairing terms appearing separately in the
creation--annihilation representation are encoded by a single real
antisymmetric matrix. Now, every real antisymmetric matrix can be brought by an orthogonal change of
basis to the block-diagonal form
\[
    OAO^{\mathsf T}
    =
    \bigoplus_{j=1}^n
    \begin{pmatrix}
        0&e_j\\
        -e_j&0
    \end{pmatrix},
    \qquad e_j\geq0.
\]
Define the corresponding transformed Majorana operators by
\[
    \widetilde{\gamma}_p
    =
    \sum_{q=1}^{2n}O_{pq}\gamma_q.
\]
Because $O$ is orthogonal, these operators satisfy the same Majorana
relations,
\[
    \{\widetilde{\gamma}_p,\widetilde{\gamma}_q\}
    =
    2\delta_{pq}I.
\]
We now combine each consecutive pair of transformed Majoranas into an
ordinary fermionic mode:
\[
    \beta_j
    =
    \frac{\widetilde{\gamma}_{2j-1}
    +i\widetilde{\gamma}_{2j}}{2},
    \qquad
    \beta_j^\dagger
    =
    \frac{\widetilde{\gamma}_{2j-1}
    -i\widetilde{\gamma}_{2j}}{2}.
\]
Since
\[
    i\widetilde{\gamma}_{2j-1}\widetilde{\gamma}_{2j}
    =
    2\beta_j^\dagger\beta_j-I,
\]
the Hamiltonian becomes
\[
    H_0
    =
    E_{\mathrm{free}}I
    +
    \sum_{j=1}^n e_j\beta_j^\dagger\beta_j,
\]
where $E_{\mathrm{free}}=E_0-\frac12\sum_j e_j$.  This orthogonal
change of Majorana coordinates is equivalently a linear transformation
mixing the original creation and annihilation operators.

Based on this definition, there are two orthonormal bases of the same \(2^n\)-dimensional Fock 
space. 
 The original modes \(a_1,\ldots,a_n\) define the \emph{occupation basis}
\[
    \lvert x_1,\ldots,x_n\rangle
    :=
    (a_1^\dagger)^{x_1}\cdots
    (a_n^\dagger)^{x_n}\lvert0^n\rangle.
\]
 and the canonical modes
\(\beta_1,\ldots,\beta_n\) define a different occupation basis
\[
    \lvert x_1,\ldots,x_n\rangle_\beta
    :=
    (\beta_1^\dagger)^{x_1}\cdots
    (\beta_n^\dagger)^{x_n}\lvert0^n\rangle_\beta.
\]
Since any quadratic fermionic Hamiltonian can be efficiently diagonalized by a Bogoliubov transformation, we henceforth work in a canonical basis in which the free bath Hamiltonian takes the form
\begin{equation}
\label{eq:bath-hamiltonian}
    H_0
    =
    E_{\mathrm{free}} I
    +
    \sum_{j=1}^n e_j \beta_j^\dagger \beta_j,
    \qquad e_j \ge 0.
\end{equation}
The scalar term \(E_{\mathrm{free}}I\) shifts all energies by the same amount, so we subtract it throughout and restore it at the end. After rescaling the overall energy unit, we may therefore assume
\begin{equation}
    H_0
    =
    \sum_{j=1}^n e_j \beta_j^\dagger \beta_j,
    \qquad
    0 \le e_j \le 1.
\end{equation}

\subsection{Impurity models}
\label{sec:impurity-model}

A fermionic quantum impurity model describes a small interacting
subsystem, called the \emph{impurity}, coupled to a large collection of
otherwise free fermionic modes.  The model  Hamiltonian is
\[
    H=H_0+V,
\]
where \(H_0\) is a quadratic fermionic Hamiltonian on all \(n\) modes
called the {\em bath Hamiltonian}
and \(V\) is an {\em interaction} supported on a small subset of modes
${\mathcal I}\subseteq [n]$. Thus $V$ belong to the algebra generated by the operators $a_j,a_j^\dagger$ with $j\in \mathcal I$.
We assume that $V$ is parity preserving.
The impurity size is defined as $m=2|\mathcal I|$.
As shown in Section~\ref{sec:bogoliubov}, after a
Bogoliubov transformation, together with a constant energy shift and a
rescaling if necessary, the bath Hamiltonian takes the form
\begin{equation}
\label{eq:writingoutH}
    H_0
    =
    \sum_{j=1}^n e_j\beta_j^\dagger\beta_j,
    \qquad
    0\leq e_j\leq 1.
\end{equation}
The normalization \(e_j\leq 1\) does not imply
\(\lVert H_0\rVert\leq 1\).  Indeed, the eigenvalues of \(H_0\) are $   \sum_{j\in X}e_j$ for $X\subseteq [n]$ 
and hence
$    \lVert H_0\rVert=\sum_{j=1}^n e_j$ 
may be as large as \(n\).
We say that the bath Hamiltonian has \emph{bath gap}
$\omega>0$ if
\[
    \omega\leq e_j\leq1
    \qquad\text{for every }j\in[n].
\]
In our original problem
we do not make any assumptions about the bath gap.

\paragraph{Impurity subspace.}
For $z\in\C^n$, define the annihilation operator
\[
 \beta(z):=\sum_{j=1}^n\overline{z_j}\,\beta_j.
\]
Each original annihilation operator can be written as
\[
 a_j=\beta(w_j)+\beta(w_j')^\dagger
\]
for some vectors $w_j,w_j'\in\C^n$.
Since $V$ is an even Hermitian polynomial in $a_j,a_j^\dagger$
with $j\in\mathcal I$, the subspace
\[
 \Span\{w_j,w_j':j\in\mathcal I\}
\]
supports $V$ and has dimension at most $2|\mathcal I|=m$.
Assuming $m\leq n$, choose an $m$-dimensional subspace
$\Aimp\subseteq\C^n$ containing this span.
We call $\Aimp$ the \emph{impurity subspace}.
By construction, $V$ is an even Hermitian polynomial in
$\beta(w)$ and $\beta(w)^\dagger$ with $w\in\Aimp$.

Fix an orthonormal basis $u_1,\ldots,u_m$ of $\Aimp$.
The operators $\beta(u_1),\ldots,\beta(u_m)$ obey the canonical
anticommutation relations and describe $m$ fermionic modes
supporting $V$. These modes need not coincide with the original
modes indexed by $\mathcal I$. No bound on the interaction
strength $\norm V$ is imposed here.

\paragraph{Equivalent bath-subspace formulation.}
The preceding definition can equivalently be expressed by identifying
the one-particle directions on which the interaction acts trivially.
If \(z\in\Aimp^\perp\), then
\(\{\beta(z),\beta(w)^\dagger\}=\langle z,w\rangle I=0\)
and \(\{\beta(z),\beta(w)\}=0\) for every \(w\in\Aimp\), so
$\beta(z)$ anticommutes with every creation and annihilation
operator supported on $\Aimp$. It therefore commutes with
every even polynomial in these operators, and hence
\begin{equation}
\label{eq:impurity-commutant}
    [V,\beta(z)]=0
    \qquad
    \text{for every }z\in\Aimp^\perp.
\end{equation}
Thus the \((n-m)\)-dimensional subspace \(\Aimp^\perp\) consists of
fermionic directions on which the interaction acts trivially, while
the interaction is confined to the low-dimensional subspace \(\Aimp\).

Conversely, suppose that \(V\) is even and Hermitian and that
\([V,\beta(z)]=0\) for every vector \(z\) in some subspace of
\(\mathbb C^n\) of codimension \(m\).  Hermiticity also gives
\([V,\beta(z)^\dagger]=0\) for these \(z\).  It follows that \(V\) is
supported on the \(m\)-dimensional orthogonal complement of that
subspace, which can therefore be taken as \(\Aimp\).  
Thus support on an $m$-dimensional one-particle subspace
is equivalent to commutation with all bath annihilation
operators in its orthogonal complement.

\paragraph{Representation size.}
Although $H$ acts on the $2^n$-dimensional Fock space $\mathcal F_n$,
it has a classical description using
\[
 O(n^2+nm+4^m)
\]
real parameters: $O(n^2)$ parameters for the quadratic Hamiltonian,
$O(nm)$ parameters for an orthonormal basis of $\Aimp$, and
$O(4^m)$ parameters for the matrix of $V$ on the $m$ fermionic
modes associated with this basis.
Since $V$ is even and Hermitian, this matrix consists of two
Hermitian blocks, one for each parity sector, and can be
specified by $2^{2m-1}$ real parameters.
We assume $m\geq1$, since for $m=0$ the Hamiltonian is quadratic
and its ground energy can be computed directly by a Bogoliubov
transformation.
 
\paragraph{Enumeration of weighted configurations}

\begin{lemma}[Enumeration of weighted configurations]\label{lem:enum}
Let $M,R$ be nonnegative integers, and let $w_1,\ldots,w_M$ be
nonnegative integer weights. The set
\[
 \mathcal Y_R=\left\{x\in\{0,1\}^M:
                   \sum_{j=1}^M w_jx_j\le R\right\}
\]
can be listed with $O\bigl((M+1)|\mathcal Y_R|\bigr)$ elementary
operations.
\end{lemma}
\begin{proof}
Traverse the binary tree of occupation prefixes in the fixed mode
order, pruning a prefix as soon as its accumulated weight exceeds
$R$. Every surviving prefix extends to an element of $\mathcal Y_R$
by setting the remaining occupations to zero. Thus every surviving
node lies on the path to a listed configuration, and each such path
contains $M+1$ nodes. There are at most $(M+1)|\mathcal Y_R|$ surviving
nodes, each requiring at most two child checks. Writing the listed
configurations has the same asymptotic cost.
\end{proof}

\subsection{Fermionic purification and thermofield double}
\label{sec:tfd-preliminaries}

We first introduce a canonical purification of fermionic operators and mixed states, which will be necessary to define the fermionic thermofield double. To identify \(\mathcal F_n\otimes\mathcal F_n\) with the doubled Fock space \(\mathcal F_{2n}\), we introduce auxiliary annihilation operators
$    \widetilde\beta_1,\ldots,\widetilde\beta_n$ 
which satisfy the canonical anticommutation relations and anticommute
with all physical creation and annihilation operators.  Let
\[
    \ket{\Omega_{2n}}
    :=
    \ket{0,\ldots,0}\in\mathcal F_{2n}
\]
denote the common vacuum of the physical and auxiliary modes, so that
\[
    \beta_j\ket{0^{2n}}
    =
    \widetilde\beta_j\ket{\Omega_{2n}}
    =
    0
    \qquad
    (j\in[n]).
\]
\paragraph{Fermionic purification} Define the unnormalized maximally entangled fermionic vector
\begin{equation}\label{tf:entangled-vector}
    \ket{I_f}
    :=
    \prod_{j=1}^n
    \left(
        I-i\beta_j^\dagger\widetilde\beta_j^\dagger
    \right)
    \ket{\Omega_{2n}},
\end{equation}
where the product is taken in increasing order of \(j\).  The factors
are even fermionic operators and therefore commute.
 For an operator \(Y\) on \(\mathcal F_n\), define its fermionic
vectorization by
\[
    \ket{Y}_f:=Y_L\ket{I_f},
\]
where \(Y_L=Y\otimes I_{\mathrm{anc}}\) acts as \(Y\) on the original
physical modes and as the identity on the auxiliary copy. The subscript
\(L\) denotes left multiplication under vectorization: for every
operator \(X\) on \(\mathcal F_n\),
\[
    Y_L\ket{X}_f=\ket{YX}_f.
\]
This vectorization preserves the Hilbert--Schmidt inner product:
\begin{equation}\label{tf:isometry}
    \!\langle Y|Z\rangle_f
    =
    \operatorname{Tr}(Y^\dagger Z),
    \qquad
    \norm{Y}_{\mathrm{HS}}
    =(\operatorname{Tr}(Y^\dagger Y))^{1/2}.
\end{equation}
The canonical purification for any fermionic mixed state $\rho = X^\dagger X$ is then given by $\ket{X}_f$.

A fermionic
operator $Y$ with definite parity $p(Y)\in\{0,1\}$ satisfies
$PYP=(-1)^{p(Y)}Y$ and is said to be \emph{homogeneous}. Then, it follows from \eqref{tf:entangled-vector} that
\begin{equation}\label{tf:right-action}
 \widetilde \beta_j^\dagger\ket{Y}_f
   =i(-1)^{p(Y)}\ket{Y \beta_j}_f,\quad \text{and} \quad
 \widetilde \beta_j\ket{Y}_f
   =i(-1)^{p(Y)}\ket{Y \beta_j^\dagger}_f
\end{equation}
for all homogeneous $Y$ and any $j$. 
It can be verified that the above identities hold for $Y=\Id$, and moving an auxilliary 
operator past $Y_L$ gives the indicated parity sign. 

\paragraph{Gibbs state and thermofield double}

The Gibbs state for a Hamiltonian
$H$  at inverse temperature \(\beta\geq0\) is
\[
    \rho_\beta(H)
    :=
    \frac{e^{-\beta H}}{Z_\beta(H)},
    \text{ where }
    Z_\beta(H):=\operatorname{Tr}(e^{-\beta H}).
\]
Following the purification convention introduced above, we define the {\em thermofield-double state} for a Hamiltonian $H$ at inverse temperature $\beta$ as the canonical purification of the Gibbs state $\rho_\beta(H)$:
\begin{equation}\label{tf:tfd}
    \ket{\mathsf{TFD}_\beta(H)}
    =\frac1{\sqrt{Z_\beta(H)}}
    \ket{e^{-\beta H/2}}_f
    \in {\mathcal F}_{2n}
\end{equation}
The contribution of a canonical mode \(j\) to the free Hamiltonian
is
$    H_{0,j}=e_j\beta_j^\dagger\beta_j.$ 
The empty state has energy \(0\), while the occupied state has energy
\(e_j\).  Their Gibbs probabilities are therefore
\[
    p_{0,j}
    =
    \frac{1}{1+e^{-\beta e_j}},
    \qquad
    p_{1,j}
    =
    \frac{e^{-\beta e_j}}{1+e^{-\beta e_j}}.
\]
Define $ u_j:=\sqrt{p_{0,j}},  w_j:=\sqrt{p_{1,j}}$. 
Thus $  u_j^2+w_j^2=1,$ and $ {w_j}/{u_j}=e^{-\beta e_j/2}.$ 
In the fermionic phase convention above, the thermofield double of
this mode and its auxiliary copy is
\[
    u_j\ket{00}_j-iw_j\ket{11}_j,
\]
where the first bit records the occupation of the physical mode
\(\beta_j\), and the second records the occupation of its auxiliary
copy \(\widetilde\beta_j\).  Tracing out the auxiliary mode gives
\[
\begin{aligned}
    &\operatorname{Tr}_{\mathrm{aux}}
    \left[
       \left(
          u_j\ket{00}_j-iw_j\ket{11}_j
       \right)
       \left(
          u_j\bra{00}_j+iw_j\bra{11}_j
       \right)
    \right] 
    =
    u_j^2\ket0\!\bra0+w_j^2\ket1\!\bra1
    =
    \frac{
       \ket0\!\bra0+
       e^{-\beta e_j}\ket1\!\bra1
    }{
       1+e^{-\beta e_j}
    },
\end{aligned}
\]
which is precisely the Gibbs state of \(H_{0,j}\).

\paragraph{Thermal annihilation operators.}
We shall now construct a canonical set of annihilation operators
for the state $\ket{\mathsf{TFD}_\beta(H_{\mathrm{ free}})}$, where
$H_{{\mathrm free}}$ is a quadratic fermionic Hamiltonian of the form
\[
H_{{\mathrm free}}
    =
    \sum_{j=1}^n \lambda_j b_j^\dagger b_j,
\]
where $b_j=\beta(v_j)=
\sum_{k=1}^n \overline{(v_j)_k} \beta_k$ for some orthonormal set $v_1,\ldots,v_n\in \C^n$ and  $\lambda_j \in [0,1]$ are single-particle excitation energies. 
Let $\widetilde b_j=\sum_{k=1}^n(v_j)_k\widetilde\beta_k$ be the correspondingly rotated ancillary modes. This conjugate choice gives $\sum_j b_j^\dagger\widetilde b_j^\dagger=\sum_k\beta_k^\dagger\widetilde\beta_k^\dagger$, so the entangled vector $\ket{I_f}$ and the right-action identities \eqref{tf:right-action} keep their form with $\beta_j,\widetilde\beta_j$ replaced by $b_j,\widetilde b_j$.
The same calculation as above shows that 
the TFD state factorizes as
\[
\ket{\mathsf{TFD}_\beta(H_{\mathrm{ free}})}
    =
    \prod_{j=1}^n
    \left(
       u_jI-iw_j
       b_j^\dagger\widetilde b_j^\dagger
    \right)
    \ket{\Omega_{2n}}=  \prod_{j=1}^n( u_j\ket{00}_j-iw_j\ket{11}_j),
\]
where Fock configurations are defined in the basis determined by $b_j, \widetilde b_j$,
\[
    u_j
    =
    \left(\frac{1}{1+e^{-\beta \lambda_j}}\right)^{1/2}
    \quad \mbox{and} \quad
    w_j
    =
    \left(\frac{e^{-\beta \lambda_j}}{1+e^{-\beta \lambda_j}}\right)^{1/2}.
\]
The TFD state is generally not the ordinary vacuum, since
\[
    b_j\ket{\mathsf{TFD}_\beta(H_{\mathrm{ free}})}\neq0
\]
whenever \(w_j\neq0\).  We therefore define the \emph{transformed doubled
operators}
\[
    \mathsf d_{-,j}
    :=
    u_j b_j + i w_j \widetilde b_j^\dagger
    \quad \mbox{and} \quad
    \mathsf d_{+,j}:=
   w_j b_j^\dagger + iu_j\widetilde b_j.
\]
For each \(j\), a direct calculation gives
\[
\begin{aligned}
    \mathsf d_{-,j}
    \left(
       u_j\ket{00}_j-iw_j\ket{11}_j
    \right)
    &=
    u_j\left(-iw_j\ket{01}_j\right)
    +
    iw_j\left(u_j\ket{01}_j\right)=0.
\end{aligned}
\]
Similarly,
\[
    \mathsf d_{+,j}
    \left(
       u_j\ket{00}_j-iw_j\ket{11}_j
    \right)
    =0.
\]
It follows that
\[
    \mathsf d_{\pm,j}
    \ket{\mathsf{TFD}_\beta(H_{\mathrm{ free}})}
    =
    0
    \qquad
    (j\in[n]).
\]
Because \(u_j^2+w_j^2=1\), these operators satisfy the canonical
anticommutation relations
\[
    \{\mathsf d_{\eta,j},\mathsf d_{\zeta,k}\}=0,
    \qquad
    \{\mathsf d_{\eta,j},\mathsf d_{\zeta,k}^\dagger\}
    =
    \delta_{\eta\zeta}\delta_{jk}I,
    \qquad
    \eta,\zeta\in\{-,+\}.
\]
Hence \(\mathsf d_{-,j}\) and \(\mathsf d_{+,j}\) are valid fermionic annihilation
operators. In this sense,
$\ket{\mathsf{TFD}_\beta(H_{\mathrm{ free}})}$  is the common vacuum of the
transformed doubled modes \(\mathsf d_{-,j},\mathsf d_{+,j}\).  This does not mean
that its physical modes are empty.  Rather, it means that the state
contains no excitations with respect to these transformed modes.

The transformation from the physical and auxiliary modes
$\{    b_j,\widetilde b_j\}$ 
to the doubled thermal modes
$\{    \mathsf d_{-,j},\mathsf d_{+,j}\}$ 
is a Bogoliubov transformation: it mixes creation and annihilation
operators while preserving the canonical anticommutation relations.
Since the free thermofield double is the common vacuum of these
Bogoliubov modes, it is a pure fermionic Gaussian state. For the interacting Hamiltonian
\[
    H=H_0+V,
\]
the thermofield double need not be Gaussian, and the transformed modes
need not be empty.  Their occupations measure how the interaction
changes the thermal state relative to the free thermal state.  These
are the occupations controlled in the proof of 
Theorem~\ref{thm:rank}.

\section{Classical algorithm for ground energy estimation}
\label{sec:ground-energy-proof}

In this section we prove Theorem~\ref{thm:ground-energy-algorithm}.

\subsection{Exponential suppression of multi-particle bath excitations}
\label{sec:pullthrough}

Consider a bipartite fermionic system $AB$, where $A$ and $B$ consist of $n_A$ and $n_B$ fermionic modes respectively.
We refer to these subsystems as  the {\em boundary} ($A$) and the {\em residual bath} ($B$).
We write $a_\ell$ and $b_j$ for the fermionic annihilation operators describing $A$ and $B$, where 
$\ell \in [n_A]$ and $j\in [n_B]$. These operators obey the standard canonical anticommutation rules.
In particular,  $a_\ell b_j=-b_j a_\ell$ for all $j$ and $\ell$. Suppose $E\in \C^{n_B\times n_B}$ is a  positive definite hermitian matrix and $T\in \C^{n_A\times n_B}$ is an arbitrary matrix.
Consider a Hamiltonian 
\begin{equation}\label{ge:boundary-H}
 \widetilde H
 =H_A+\sum_{i,j=1}^{n_B}E_{ij}b_i^\dagger b_j
+\sum_{\ell=1}^{n_A}\sum_{j=1}^{n_B}
   \left(T_{\ell j}a_\ell^\dagger b_j+
        \overline{T_{\ell j}}b_j^\dagger a_\ell\right),
\end{equation}
where $H_A$ is any even hermitian Hamiltonian supported  only on $A$.
Given a vector $v\in \C^{n_B}$, let $b(v)=\sum_{j=1}^{n_B} \overline{v}_j b_j$,
where $\overline{v}_j$ is the complex conjugate of $v_j$.
The following lemma provides an upper bound on occupation 
of the residual bath modes. Below, empty operator and scalar products should be understood as $I$ and $1$ respectively. 
\begin{lemma}[\bf Multiparticle pull-through]
\label{ge:pull}
Let $\psi\in \calF_{n_A+n_B} $ be a normalized ground state of $\widetilde H$.
For any integer $p\geq 0$ and any vectors $v_1,\ldots,v_p\in \C^{n_B}$
one has 
\begin{equation}\label{ge:pull-bound}
 \norm{b(v_1)\cdots b(v_p)\psi}
 \leq\prod_{i=1}^p s(v_i),
\end{equation}
where
\[
s(v)=\int_0^\infty\norm{Te^{-tE}v}\,dt.
\]
If the vectors $v_1,\ldots,v_p$ are orthonormal, the occupation number operators
$n_i:=b(v_i)^\dagger b(v_i)$ pairwise commute and obey
\begin{equation}\label{ge:joint-general}
 \bra{\psi}\prod_{i=1}^p n_i\ket{\psi}
 \leq\prod_{i=1}^p s(v_i)^2.
\end{equation}
\end{lemma}

\begin{proof}
Given a vector $u\in \C^{n_A}$, let $a(u)=\sum_{\ell=1}^{n_A} \overline{u}_\ell a_\ell$.
The canonical anticommutation rules imply that 
$\|a(u)\|=\|u\|$ and $\|b(v)\|=\|v\|$ for all $u\in \C^{n_A}$ and $v\in \C^{n_B}$.
Furthermore,  for any vectors $v,w\in \C^{n_B}$ one has
\begin{equation}\label{ge:commutator}
 [\widetilde H,b(v)]=-b(Ev)-a(Tv)\quad \mbox{and} \quad
 \{a(Tv),b(w)\}=0.
\end{equation}

The proof of \eqref{ge:pull-bound} follows by induction on $p$. The base case $p=0$ is simply the claim that $\|\psi\| \le 1$ since no operator is applied to the LHS of \eqref{ge:pull-bound}.
For $p\geq1$ and $t\geq0$, define the ordered products
\[
 B_p(t)=\prod_{j=1}^p b(e^{-tE}v_j),\qquad
 B_{p,i}(t)=\prod_{\substack{j\in[p]\\j\ne i}}b(e^{-tE}v_j)
 \quad(i\in[p]),
\]
where both products are taken in increasing order of $j$, and an empty product in $ B_{p,i}(t)$ is understood to be the identity $I$. Note that $B_p(0) = \prod_{j=1}^p b(v_j)$.
Differentiation yields
\begin{equation}\label{ge:derivative}
 \frac{d}{dt}\bigl(e^{-t(\widetilde H-E_1(\widetilde H)\Id)}B_p(t)\psi\bigr)
 =e^{-t(\widetilde H-E_1(\widetilde H)\Id)}\sum_{i=1}^p(-1)^{i-1}
                a(Te^{-tE}v_i)B_{p,i}(t)\psi.
\end{equation}
Indeed, the derivative of the exponential contributes
$-e^{-t(\widetilde H-E_1(\widetilde H)\Id)}[\widetilde H,B_p(t)]\psi$.
The terms containing $b(Ee^{-tE}v_i)$
cancel the derivatives of the annihilators. Moving each remaining
boundary annihilator $a(Te^{-tE}v_i)$ past its $i-1$ preceding bath
annihilators using the second identity of \eqref{ge:commutator}
gives the displayed sign. 

We claim that the vector being differentiated in the left hand side of \eqref{ge:derivative}
 tends to zero as $t\to\infty$. Indeed, we have $\widetilde H-E_1(\widetilde H)\Id\succeq0$
 and thus $\|e^{-t(\widetilde H-E_1(\widetilde H)\Id)}\|\le 1$ for all $t\ge 0$.
By assumption, $E$ is positive definite. Let $\omega>0$ be the smallest eigenvalue of $E$.
Then
\[
 \norm{B_p(t)}\leq e^{-p\omega t}\prod_{i=1}^p\norm{v_i}
\]
which tends to $0$ as $t\to\infty$.
Integrating \eqref{ge:derivative} over $[0,\infty)$ therefore gives
\begin{equation}\label{ge:derivative-integrated}
B_p(0)\psi= -\sum_{i=1}^p(-1)^{i-1} \int_0^\infty e^{-t(\widetilde H -  E_1(\widetilde H)\Id)} a(Te^{-tE}v_i)B_{p,i}(t)\psi\,dt.  
\end{equation}

Now, apply the induction hypothesis to each $(p-1)$-particle product and
again use $\norm{a(Te^{-tE}v_i)}=\norm{Te^{-tE}v_i}$.
For every $i\in[p]$ one has
\[
 s(e^{-tE}v_i)=\int_t^\infty\norm{Te^{-uE}v_i}\,du
 \quad \mbox{and} \quad
 \frac{d}{dt}s(e^{-tE}v_i)=-\norm{Te^{-tE}v_i}.
\]
Consequently,
\begin{align}
 \norm{B_p(0)\psi}
 &\leq\sum_{i=1}^p\int_0^\infty
       \norm{Te^{-tE}v_i}
       \prod_{\substack{j\in[p]\\j\ne i}}s(e^{-tE}v_j)\,dt\notag\\
 &=-\int_0^\infty\frac{d}{dt}
              \left(\prod_{i=1}^p s(e^{-tE}v_i)\right)dt
 =\prod_{i=1}^p s(v_i).
 \label{ge:induction}
\end{align}
Here $s(e^{-tE}v_i)\to0$ as $t\to\infty$.
This proves
\eqref{ge:pull-bound}. For orthonormal vectors $v_1,\ldots,v_p$,
the canonical anticommutation rules imply that the operators
$n_i=b(v_i)^\dagger b(v_i)$ pairwise commute.
Then  \eqref{ge:joint-general} follows directly from \eqref{ge:pull-bound}.
\end{proof}
Lemma~\ref{ge:pull} motivates imposing an additional structure on the
residual bath space $\C^{n_B}$ that facilitates application of the upper bound \eqref{ge:pull-bound}.
In the next sections we decompose $\C^{n_B}$ into a direct sum of subspaces labelled by 
a Krylov depth and an energy scale, such that the quantity $s(v)$ that appears in the upper bound \eqref{ge:pull-bound}
decays exponentially with the Krylov depth for each energy scale.

\subsection{Shelled Krylov decompositions}
\label{sec:krylov-decompositions}

Here we define an abstract Krylov decomposition of impurity models.
It serves as a general template for  more specialized Krylov decompositions used in the proof of  Theorem~\ref{thm:ground-energy-algorithm}
and the two TFD theorems.
The abstract Krylov decomposition depends on the data $(\A,\B,E,T,\beta)$, where
$\A$ and $\B$ are finite-dimensional complex Hilbert spaces,
$E\, : \,  \B\to \B$ is a hermitian operator such that all eigenvalues of $E$ lie in $[0,1]$,
$T\, : \, \B \to \A$ is an arbitrary linear map, and $\beta\in[0,\infty]$ is the inverse temperature.
We call $\A$ the {\em boundary} and we call $\B$ the {\em residual bath}.
Let  $n_A=\dim\A$ and $n_B=\dim\B$.   
Below all vector norms are Euclidean, and all
operator norms are the induced norms.

\subsubsection{Krylov depth and integrated coupling}

Let us  define two measures of how strongly a vector $v\in \B$ is coupled to the boundary $\A$.

\begin{definition}[Krylov depth]
A vector $v\in\B$ has Krylov depth $\ell$, where $\ell\geq0$ is an integer,
if
\begin{equation}\label{kd:depth}
 TE^jv=0\quad \mbox{for all} \quad 0\leq j<\ell.
\end{equation}
\end{definition}

\begin{definition}[Integrated coupling]
A vector  $v\in\B$ has integrated coupling 
\begin{equation}\label{kd:integrated-coupling}
 s_\beta(v)=\int_0^{\beta/2}\norm{Te^{-tE}v}\,dt.
\end{equation}
For $\beta=\infty$, the integral is over $[0,\infty)$ and may be infinite. We use a shorthand $s(v)\equiv s_{\infty}(v)$.
\end{definition}
Note that the zero-temperature integrated coupling $s(v)$ appears in the upper bound \eqref{ge:pull-bound}  that controls occupation numbers of the residual bath modes
in the ground state of impurity models. 
In the thermal applications, \(s_\beta(v)\) controls deviations of the residual bath from its free thermal state.
The following lemma gives an upper bound on the integrated coupling in terms of the Krylov depth.
\begin{lemma}[Krylov depth vs integrated coupling]
\label{kd:interval-lemma}
Assume all eigenvalues of $E$ lie in the interval $[e^\downarrow,e^\uparrow]\subseteq [0,1]$..
Suppose first that $e^\downarrow>0$.
Then for any unit vector $v\in \B$ of Krylov depth $\ell$ one has
\begin{equation}\label{kd:interval-bound}
 s_\beta(v)\leq\frac{\norm{T}}{e^\downarrow}
 \left(\frac{e^\uparrow-e^\downarrow}
 {e^\uparrow+e^\downarrow}\right)^\ell.
\end{equation}
Suppose next that $e^\downarrow=0$, 
 $e^\uparrow>0$, and $\beta<\infty$. Then
\begin{equation}\label{kd:finite-time-bound}
 s_\beta(v)\leq
 \frac{\norm{T}(e^\uparrow/2)^\ell(\beta/2)^{\ell+1}}
 {(\ell+1)!}.
\end{equation}
\end{lemma}

\begin{proof}
Subtracting the midpoint of the spectral interval gives
\[
 \norm{E-\frac{e^\downarrow+e^\uparrow}{2}I}
 \leq\frac{e^\uparrow-e^\downarrow}{2}.
\]
The RHS in \eqref{kd:depth} gives 0 even with
$\bigl(E-\tfrac12(e^\downarrow+e^\uparrow)I\bigr)^j$ in place of
$E^j$. Expanding the exponential about the midpoint therefore yields
\[
 \norm{Te^{-tE}v}
 \leq\norm{T}e^{-t(e^\downarrow+e^\uparrow)/2}
 \sum_{j=\ell}^\infty
 \frac{\bigl(t(e^\uparrow-e^\downarrow)/2\bigr)^j}{j!}.
\]
Integrating the nonnegative terms over $[0,\infty)$ gives
\begin{align*}
 s_\beta(v)\leq s_\infty(v)
 &\leq\frac{2\norm{T}}{e^\downarrow+e^\uparrow}
 \sum_{j=\ell}^\infty
 \left(\frac{e^\uparrow-e^\downarrow}
 {e^\uparrow+e^\downarrow}\right)^j=\frac{\norm{T}}{e^\downarrow}
 \left(\frac{e^\uparrow-e^\downarrow}
 {e^\uparrow+e^\downarrow}\right)^\ell.
\end{align*}

To prove \eqref{kd:finite-time-bound}, note that the midpoint now is $e^\uparrow/2$.
For $t\geq0$,
\[
 \norm{Te^{-tE}v}
 \leq\norm{T}e^{-te^\uparrow/2}
 \sum_{j=\ell}^\infty\frac{(te^\uparrow/2)^j}{j!}
 \leq\norm{T}\frac{(te^\uparrow/2)^\ell}{\ell!}.
\]
The last inequality follows from
$\sum_{j=\ell}^\infty x^j/j!\leq x^\ell e^x/\ell!$ for $x\geq0$.
Integrating over $[0,\beta/2]$ proves the claim.
\end{proof}
In particular, if
\begin{equation}\label{kd:positive-shell-hypotheses}
 e^\downarrow I\leq E\leq2e^\downarrow I,
 \qquad e^\downarrow>0,
 \qquad\norm{T}\leq e^\downarrow/2,
\end{equation}
then Lemma~\ref{kd:interval-lemma} gives
\begin{equation}\label{kd:positive-depth-bound}
 s_\beta(v)\leq\tfrac12\,3^{-\ell}
 \qquad(\beta\in[0,\infty])
\end{equation}
for any vector $v$ of Krylov depth $\ell$.
This motivates splitting the spectrum into intervals with a bounded
ratio of upper and lower endpoints. 

\subsubsection{Energy shells }

We shall mostly work with a fine-grained version of Krylov decompositions where the residual bath $\B$
is represented as a direct sum of subspaces called energy shells.

\begin{definition}\label{kd:shelled-krylov-definition}
A shelled Krylov decomposition is a tuple 
 $(\A,\B,E,T,\beta)$ as above equipped with
 an orthogonal decomposition
\begin{equation}\label{kd:shells}
 \B=\bigoplus_{s=0}^{S-1}\B_s
\end{equation}
such that each subspace $\B_s$ is $E$-invariant. We refer to $\B_s$ as energy shells.
Accordingly, 
\begin{equation}\label{kd:shells_for_E}
E=\bigoplus_{s=0}^{S-1}E_s, \qquad E_s=E|_{\B_s}.
\end{equation}
Let $T_s\, : \, \B_s\to \A$ be the restriction of $T$ onto $\B_s$ such that $T=\sum_sT_sP_s$ and $Te^{-tE}v=T_se^{-tE_s}v$ for $v\in\B_s$.
Let $P_s\, :\, \B\to \B$ be the orthogonal projector onto $\B_s$.
Then $T_s^\dagger=P_sT^\dagger$.
\end{definition}
Note that all maps $T_s$ have the same target $\A$ and their images need not be
orthogonal. 
Concretely, we will later obtain such decompositions by partitioning $[0,1]$ into
intervals $I_s$ and choosing $P_s$ as the projector onto the span of eigenvectors of $E$
with eigenvalues lying in $I_s$. 

For each shell, define the Krylov spaces and their orthogonal layers by
\begin{align}
 \K_{s,\ell}
 &=\Span\{E_s^jT_s^\dagger u:
       u\in\A,\ 0\leq j<\ell\},
       \qquad \K_{s,0}=\{0\},\label{kd:krylov}\\
 \W_{s,\ell}
 &=\K_{s,\ell+1}\ominus\K_{s,\ell}
       \qquad(\ell\geq0).\label{kd:layers}
\end{align}
Here $U\ominus V=U\cap V^\perp$ when $V\subseteq U$.
Thus $\K_{s,1}=\ran T_s^\dagger$ consists of fermionic directions that couple
directly to the boundary. Each subsequent layer contains the new
directions obtained by one further application of $E_s$. The lemma below collects some useful facts about the geometry of Krylov spaces.

\begin{lemma}[Krylov geometry]\label{kd:geometry}
The following hold for any shelled Krylov decomposition as defined above.
\begin{itemize}
\item For every shell $s$ and integer $\ell\geq0$,
\begin{equation}\label{kd:layer-size}
 \dim\W_{s,\ell}\leq\rank T_s\leq n_A.
\end{equation}
\item  A vector $v\in\B_s$ has Krylov depth $\ell$ if and only if
$v\perp\K_{s,\ell}$. In particular, every vector in $\W_{s,\ell}$
has Krylov depth $\ell$.
\item The Krylov spaces stabilize by $\ell=n_B$. 
\item The largest $E_s$-invariant subspace of $\ker T_s$ subspace is 
\begin{equation}\label{kd:decoupled}
 \Z_s=\B_s\ominus\K_{s,n_B},
\end{equation}
and for any vector $v\in\Z_s$ and $\beta\in[0,\infty]$ we have
\begin{equation}\label{kd:zero-coupling}
 s_\beta(v)=0.
\end{equation}
\item Lastly, we have
\begin{align}
 E_s\W_{s,\ell}
 &\subseteq\W_{s,\ell-1}\oplus\W_{s,\ell}\oplus\W_{s,\ell+1},
 \label{kd:nearest-layers}\\
 \B_s&=\left(\bigoplus_{\ell=0}^{n_B-1}\W_{s,\ell}\right)\oplus\Z_s,
 \label{kd:decomposition}
\end{align}
where $\W_{s,-1}=\{0\}$ and layers with $\ell\geq n_B$ are zero.
\end{itemize}
\end{lemma}

\begin{proof}
Passing from $\K_{s,\ell}$ to $\K_{s,\ell+1}$ adds only
$E_s^\ell\ran T_s^\dagger$, whose dimension is at most $\rank T_s$.
Also,
\[
 v\perp\K_{s,\ell}
 \quad\Longleftrightarrow\quad
 T_sE_s^jv=0\quad(0\leq j<\ell).
\]
For $v\in\B_s$ these are exactly the conditions $TE^jv=0$.

The Cayley--Hamilton theorem implies stabilization by
$\dim\B_s\leq n_B$, so $\K_{s,n_B}$ is invariant under $E_s$.
Hermiticity makes $\Z_s$ invariant as well. Orthogonality to
$\ran T_s^\dagger$ gives $T_s\Z_s=0$.
Conversely, if an $E_s$-invariant subspace lies in $\ker T_s$, each
of its vectors is orthogonal to every $E_s^j\ran T_s^\dagger$, so it
lies in $\Z_s$. Since $e^{-tE_s}$ preserves $\Z_s$ and $T_s$ vanishes
there, \eqref{kd:zero-coupling} follows.

Finally, $E_s\W_{s,\ell}\subseteq\K_{s,\ell+2}$. For $\ell\geq1$,
Hermiticity and $E_s\K_{s,\ell-1}\subseteq\K_{s,\ell}$ imply
$E_s\W_{s,\ell}\perp\K_{s,\ell-1}$. This proves
\eqref{kd:nearest-layers}; the case $\ell=0$ is immediate.
Equation~\eqref{kd:decomposition} follows by taking successive
orthogonal differences of the nested Krylov spaces.
\end{proof}

Choose an orthonormal basis in each nonzero layer and in each $\Z_s$.
Their union is an orthonormal basis of $\B$. At any fixed depth, the
number of layer vectors is at most
\begin{equation}\label{kd:total-layer-size}
 \sum_{s=0}^{S-1}\dim\W_{s,\ell}
 \leq\sum_{s=0}^{S-1}\rank T_s.
\end{equation}
Thus a bound $\rank T_s\leq m$ gives at most $mS$ vectors per depth.
Lemma~\ref{kd:interval-lemma} gives the following.
\begin{corollary}\label{cor:kd-shell-depth}
Consider a shell $s$ and suppose all eigenvalues of $E_s$ lie in the interval
$[e_s^\downarrow,e_s^\uparrow]\subseteq (0,1]$.
For any unit vector $v\in\W_{s,\ell}$ and any  $\beta\in[0,\infty]$ one has
\begin{equation}\label{kd:shell-depth-bound}
 s_\beta(v)\leq\frac{\norm{T_s}}{e_s^\downarrow}
 \left(\frac{e_s^\uparrow-e_s^\downarrow}
 {e_s^\uparrow+e_s^\downarrow}\right)^\ell.
\end{equation}
In particular, $s_\beta(v)\leq\tfrac12 3^{-\ell}$ if
$e_s^\uparrow\leq2e_s^\downarrow$ and
$\norm{T_s}\leq e_s^\downarrow/2$.
\end{corollary}

On a shell with $0\leq E_s\leq e_s^\uparrow I$, the finite-time
estimate instead gives
\begin{equation}\label{kd:low-shell-bound}
 s_\beta(v)\leq
 \frac{(\beta e_s^\uparrow/4)^{\ell+1}}{(\ell+1)!}
 \leq\tfrac12\,3^{-\ell}
\end{equation}
for any unit vector $v\in\W_{s,\ell}$
provided $\beta<\infty$, $\beta e_s^\uparrow\leq1$, and
$\norm{T_s}\leq e_s^\uparrow/2$.
Indeed, the first inequality follows from \eqref{kd:finite-time-bound},
and the second uses
$4^{-(\ell+1)}/(\ell+1)!\leq\tfrac12 3^{-\ell}$.
If $e_s^\uparrow=0$, the claim follows from $T_s=0$.

\subsection{Opening a bath gap}

    We now discuss how to create an artificial  bath gap  by replacing
each bath excitation energy \(e_j\) with
$\max\{e_j,\omega\}$, where $0<\omega\le 1$ is the desired bath gap.
A rounding argument then shows that this modification changes the
ground energy by at most \(m\omega\). \footnote{Strictly speaking, only positive one-particle energies below $\omega$ need to be rounded up. The exact zero-energy subspace splits into an at-most-$m$-dimensional part visible to the interaction, which can be retained explicitly, and a complementary part that decouples and can be removed without changing the ground energy. We round up every one-particle energy below $\omega$ for convenience.}
Let $\beta_1,\ldots,\beta_n$ be the annihilation operators diagonalizing $H_0$ such that 
\[
H_0=\sum_{j=1}^n e_j \beta_j^\dag \beta_j.
\]
For $j\in [n]$, define
\begin{equation}\label{ge:rounded-H}
 \widetilde e_j=\max\{e_j,\omega\},\quad
 \widetilde H_0=\sum_{j=1}^n\widetilde e_j \beta_j^\dagger \beta_j,
 \qquad \widetilde H=\widetilde H_0+V.
\end{equation}
For subsequent one-particle calculations, write
\[
 h=\operatorname{diag}(e_1,\ldots,e_n),\qquad
 \widetilde h=\operatorname{diag}(\widetilde e_1,\ldots,\widetilde e_n).
\]
Thus $\omega\Id\preceq\widetilde h\preceq\Id$.
\begin{lemma}\label{ge:rounding-lemma}
The rounded Hamiltonian $\widetilde{H}$ of Eq.~\eqref{ge:rounded-H} satisfies
$$
H\preceq\widetilde H,\qquad 0\leq E_1(\widetilde H)-E_1(H)\leq m\omega.
$$
\end{lemma}
\begin{proof}
Let us define $J_\omega := \{j\in[n]:e_j<\omega\}$ to be the set of indices for which the corresponding energy is lower than the threshold i.e., $e_j < \omega$. Let $P_{\omega}$ be the projector onto the corresponding coordinate subspace. We now define
$$
H^- := \sum_{j\in[n] \setminus J_{\omega}} e_j \beta_j^\dagger \beta_j+V, \qquad N_{\omega}=\sum_{j\in J_{\omega}} \beta_j^\dagger \beta_j.
$$
Since $\dim P_{\omega} \Aimp \leq m$, most of the low-energy space is untouched by either term of $H^-$: the modes spanning $\ran P_{\omega} \ominus P_{\omega} \Aimp$ (where $U\ominus W:=U\cap W^\perp$ for $W\subseteq U$) carry zero quadratic energy under $H^-$ and their
one-particle vectors lie in $\Aimp^\perp$, so they don't feel the interaction either. We may therefore pick a ground state $\theta$ of $H^-$ in which all of these unused modes are empty — this is always possible, since their annihilators commute with $H^-$ and so send a ground state either
to zero or to another ground state.
For this $\theta$ we get $\bra{\theta}N_{\omega} \ket{\theta}\leq m$, because only the (at most $m$-dimensional) modes in $P_{\omega}\Aimp$ remain available to be occupied. Now observe that
\[
    H-H^-=\sum_{j\in J_\omega} e_j\,\beta_j^\dagger\beta_j\succeq0,
    \qquad
    \widetilde H-H=\sum_{j\in J_\omega}(\omega-e_j)\beta_j^\dagger\beta_j\succeq0,
\]
so \(H^-\preceq H\preceq\widetilde H\).  Moreover,
\[
    \widetilde H-H^-=\omega\sum_{j\in J_\omega}\beta_j^\dagger\beta_j
    =\omega N_\omega.
\]
so we have that
\[
 E_1(\widetilde H)
 \leq\bra{\theta}\widetilde H\ket{\theta}
 \leq E_1(H^-)+m\omega\leq E_1(H)+m\omega.
\]
The remaining inequality $E_1(H)\le E_1(\widetilde H)$ is immediate from
$H\preceq\widetilde H$.
\end{proof}

\subsection{Zero-temperature  Krylov decomposition}
\label{ge:shelled-decomposition}

Let us specialize abstract Krylov decompositions of Section~\ref{sec:krylov-decompositions}
to the rounded impurity model Hamiltonian. We shall use $\beta=\infty$ version of the Krylov decompositions
since we are interested in ground state properties.

\subsubsection{Energy shells and the impurity boundary}
\label{sec:shell1}\label{ge:shells}
We begin by decomposing the spectrum of $\widetilde h$ into energy shells of increasing width, reminiscent of the logarithmic discretization scheme used in renormalization group methods~\cite{wilson1975renormalization,bulla2008numerical}. Let
\[
 S=1+\left\lfloor\log_2(1/\omega)\right\rfloor
\]
be the total number of shells. For $s=0,\ldots,S-1$, define
\[
 I_s=[2^s\omega,\,2^{s+1}\omega),
 \qquad
 \Pi_s=\sum_{\substack{j\in[n]\\ \widetilde e_j\in I_s}}|j\rangle\langle j|,
 \qquad
 \mathcal H_s=\ran\Pi_s \subseteq \C^n.
\]
Empty shells are allowed. The spaces $\mathcal H_s$ are mutually
orthogonal, have direct sum $\C^n$, and are preserved by $\widetilde h$.
The spectrum of $\widetilde h|_{\mathcal H_s}$ lies in
$[2^s\omega,\min\{2^{s+1}\omega,1\}]$. 

The boundary and residual-bath subspaces in shell $s$ are
\begin{equation}\label{ge:shell-AB}
 \A_s=\Pi_s\Aimp,
 \qquad
 \B_s=\mathcal H_s\ominus\A_s.
\end{equation}
Note that $\A_s$ is generally not a subspace of $\Aimp$,
since $\Pi_s$ need not preserve $\Aimp$. Define
\begin{equation}\label{ge:AB}
 \A=\bigoplus_{s=0}^{S-1}\A_s,
 \qquad
 \B=\bigoplus_{s=0}^{S-1}\B_s,
 \qquad
 n_A=\dim\A,\quad n_B=\dim\B=n-n_A.
\end{equation}
Then $\C^n=\A\oplus\B$ orthogonally, $\dim\A_s\leq m$,
$n_A\leq mS$, and $\Aimp\subseteq\A$.
In particular, $\B\subseteq\Aimp^\perp$ and $V$ is supported entirely on $\A$.
Throughout the proof, \emph{boundary} means $\A$ and
\emph{residual bath} means $\B$.

\subsubsection{Boundary--bath decomposition}
\label{sec:shell2}

Relative to the decomposition $\C^n=\A\oplus\B$, write $ \widetilde h$ as a block matrix
\[
 \widetilde h=
 \begin{pmatrix}
  \widetilde h_{\A\A}&\widetilde h_{\A\B}\\
  \widetilde h_{\A\B}^{\dagger}&\widetilde h_{\B\B}
 \end{pmatrix}.
\]
We instantiate the abstract Krylov decomposition template of Section~\ref{sec:krylov-decompositions}  with
\begin{equation}\label{ge:template-data}
 E=\widetilde h_{\B\B},\qquad
 T=\widetilde h_{\A\B},\qquad
 \beta=\infty,
\end{equation}
and with the spaces $\A$, $\B$, and $\B_s$ defined above.
Each $\B_s$ is invariant under $E$, and $\omega\Id_{\B}\preceq E\preceq\Id_{\B}$.
The template's shell projectors are $P_s=\Pi_s|_{\B}$;
its shell maps are $E_s=E|_{\B_s}$ and $T_s=T|_{\B_s}$.
Here $T_s$ has target $\A$, as in the template, but we have the stronger condition that
$\ran T_s\subseteq\A_s$.

Choose orthonormal bases $u_1,\ldots,u_{n_A}$ of $\A$
and $v_1,\ldots,v_{n_B}$ of $\B$, and set
$a_\ell=\beta(u_\ell)$ and $b_j=\beta(v_j)$.
In these bases, the rounded Hamiltonian has the form
\eqref{ge:boundary-H}, with
\[
 H_A=V+\sum_{\ell,\ell'=1}^{n_A}
 (\widetilde h_{\A\A})_{\ell\ell'}
 a_\ell^\dagger a_{\ell'}.
\]
Since $H_A$ is even and supported on the boundary, and
$E\succeq\omega\Id_{\B}$ with $\omega>0$,
Lemma~\ref{ge:pull} applies.

\begin{lemma}[Shell block bounds]\label{ge:block-lemma}
For every $s\in\{0,\ldots,S-1\}$ one has
\begin{equation}\label{ge:block-bounds}
 2^s\omega\Id_{\B_s}\preceq E_s
 \preceq\min\{2^{s+1}\omega,1\}\Id_{\B_s},
 \qquad
 \norm{T_s}\leq\frac{2^s\omega}{2}.
\end{equation}
\end{lemma}
\begin{proof}
The bounds on $E_s$ follow by compressing
$\widetilde h|_{\mathcal H_s}$ to $\B_s$.
For the coupling bound, subtract the midpoint of the shell:
\[
 \left\|\widetilde h|_{\mathcal H_s}
       -\frac{3\cdot2^s\omega}{2}\Id_{\mathcal H_s}\right\|
 \leq\frac{2^s\omega}{2}.
\]
The off-diagonal block relative to
$\mathcal H_s=\A_s\oplus\B_s$ is still $T_s$, with its image
regarded as a subspace of $\A_s$. Its norm is therefore at most
$2^s\omega/2$.
\end{proof}

\subsubsection{Krylov layers and integrated coupling}
\label{sec:blockkrylov}\label{sec:krylovdecay}
Let $\K_{s,\ell}$, $\W_{s,\ell}$, and $\Z_s$ be the spaces
defined in \eqref{kd:krylov}, \eqref{kd:layers}, and
\eqref{kd:decoupled} for this instance of the template. 
Lemma~\ref{kd:geometry} gives
\begin{equation}\label{ge:layer-size}
 \dim\W_{s,\ell}\leq\rank T_s\leq\dim\A_s\leq m
\end{equation}
where the last inequality is strengthened in this case because $\ran T_s\subseteq\A_s$. It also provides the orthogonal decomposition
\[
 \C^n=\A\oplus
 \bigoplus_{s=0}^{S-1}
 \left(\bigoplus_{\ell=0}^{n_B-1}\W_{s,\ell}\oplus\Z_s\right).
\]
Choose an orthonormal basis in each nonzero layer and each $\Z_s$.
There are at most $mS$ layer vectors at any fixed depth.
Throughout this section, write $s(v)=s_\infty(v)$ for the integrated
coupling defined in \eqref{kd:integrated-coupling}.
It is finite for every $v\in\B$, since $E\succeq\omega\Id_{\B}$ and $\omega>0$.

\begin{lemma}[Integrated coupling by depth]\label{ge:depth}
For every $s\in\{0,\ldots,S-1\}$, $0\leq\ell<n_B$, and unit vector
$v\in\W_{s,\ell}$, one has
\begin{equation}\label{ge:depth-bound}
 s(v)\leq\tfrac12\,3^{-\ell}.
\end{equation}
For every $v\in\Z_s$, one has $s(v)=0$.
\end{lemma}
\begin{proof}
Every vector in $\W_{s,\ell}$ has Krylov depth $\ell$ by
Lemma~\ref{kd:geometry}. Applying \Cref{cor:kd-shell-depth}
to the shell maps $E_s,T_s$ and using Lemma~\ref{ge:block-lemma}, proves \eqref{ge:depth-bound}.
The assertion for $\Z_s$ is part of Lemma~\ref{kd:geometry}.
\end{proof}

\subsection{From Krylov depth to ground-state occupations}
\subsubsection{Joint occupation bounds}
We now combine the boundary, Krylov-layer, and decoupled bath bases
into one ordered orthonormal basis $\xi_1,\ldots,\xi_n$ of $\C^n$.
The orthogonal decomposition established above ensures that the
following lists together form a complete orthonormal basis.
\begin{enumerate}
\item \emph{Boundary modes.} The first $n_A$ vectors are
$\xi_\ell=u_\ell$ for $1\leq\ell\leq n_A$, where
$u_1,\ldots,u_{n_A}$ is the orthonormal basis of the boundary space $\A$
used in Section~\ref{sec:shell2}.

\item \emph{Krylov-layer modes.} Next list the chosen orthonormal
bases of the nonzero spaces $\W_{s,\ell}$. Order the shells
as $s=0,1,\ldots,S-1$. Within each shell, order the layers as
$\ell=0,1,\ldots,n_B-1$. This is lexicographic order on $(s,\ell)$,
with the shell index compared first. Fix any order for the basis
vectors within each layer.

\item \emph{Decoupled bath modes.} Finally, for
$s=0,1,\ldots,S-1$, list an orthonormal basis of
$\Z_s=\B_s\ominus\K_{s,n_B}$, fixing any order within that basis.
These vectors lie in the part of the residual bath that cannot couple
to the boundary under free bath propagation:
$T_se^{-tE_s}v=0$ for $v\in\Z_s$ and every $t\geq0$.
This is the meaning of \emph{decoupled} here.
\end{enumerate}
The occupation bounds do not depend on this choice of order; we fix
it to specify the ordered products that define Fock basis states below.

Each vector $\xi_j$ has a position $j$ in this combined basis.
Let $\mathcal J\subseteq[n]$ be the set of positions occupied by
the Krylov-layer basis vectors in item~2, and let
$\mathcal D\subseteq[n]$ be the set of positions occupied by the
decoupled bath basis vectors in item~3. Thus
\[
 [n]=\{1,\ldots,n_A\}\cup \mathcal J \cup \mathcal D
\]
where the unions are disjoint.
For $j\in\mathcal J$, the vector $\xi_j$ belongs to a unique chosen
layer $\W_{s,\ell}$; define $\ell(j)=\ell$ to be its
Krylov depth. The position $j$ and the depth $\ell(j)$ are different
indices: many basis vectors can have the same depth.

For each $j\in[n]$, let
\[
 f_j=\beta(\xi_j),\qquad n_j=f_j^\dagger f_j.
\]
Here $f_j$ annihilates the fermionic mode associated with the
one-particle vector $\xi_j$, and $n_j$ is the projection onto states
in which that mode is occupied. There are at most $mS$ Krylov-layer
modes at any given depth, by \eqref{ge:layer-size}.
From now on, $\psi$ is any normalized global ground state of
$\widetilde H$.

\begin{corollary}[Joint occupation probabilities]\label{ge:joint}
For every subset $I\subseteq\mathcal J$,
\begin{equation}\label{ge:joint-depth}
 \bra{\psi}\prod_{j\in I}n_j\ket{\psi}
 \leq\prod_{j\in I}\left(\tfrac14 9^{-\ell(j)}\right).
\end{equation}
Every decoupled mode is empty: $n_j\psi=0$ for $j\in\mathcal D$.
\end{corollary}
\begin{proof}
Apply Lemma~\ref{ge:pull} to the orthonormal layer vectors and use
Lemma~\ref{ge:depth}. A decoupled vector has $s(v)=0$, so its
annihilator kills $\psi$.
\end{proof}

\subsection{Weighted Fock-space truncation}
\label{sec:weightfoctruncation}
Let \(\Omega\) denote the common vacuum of
\(\beta_1,\ldots,\beta_n\), and hence also of
\(f_1,\ldots,f_n\). 
We reserve $\ket{x}$ for the occupation basis of the original modes
$a_1,\ldots,a_n$ from Section~\ref{sec:fermionic-systems}.
Define the unitary $U$ on $\calF_n$ by
\begin{equation}\label{ge:bogoliubov-U}
 U\ket{x}=\prod_{j=1}^n(f_j^\dagger)^{x_j}\ket{\Omega}
 \qquad(x\in\{0,1\}^n).
\end{equation}
The states $U\ket{x}$ form the orthonormal Krylov Fock basis.
The ordered products defining the two occupation bases obey the
same sign rules, so
$Ua_jU^\dagger=f_j$, which is a linear combination of the $a_k$ and
$a_k^\dagger$. Hence $U$ is a Bogoliubov transformation. It is
determined by the canonical transformation of
Section~\ref{sec:fermionic-systems} and the one-particle basis
$\xi_1,\ldots,\xi_n$. Its linear action on the modes specifies $U$
up to an irrelevant overall phase using $O(n^2)$ coefficients.

Define
a weighted occupation number operator for the Krylov-layer modes,
\begin{equation}\label{ge:Q}
 Q=\sum_{j\in\mathcal J}(\ell(j)+1)n_j.
\end{equation}
For $x\in\{0,1\}^n$, let $Q(x)$ be the eigenvalue of $Q$ on
the Krylov Fock state $U\ket{x}$:
\begin{equation}
Q(x)=\bra{x}U^\dagger QU\ket{x}
    =\sum_{j\in\mathcal J}(\ell(j)+1)x_j.
\end{equation}
For an integer  cutoff $R\geq0$, define a set of Fock configurations 
\begin{equation}
 \mathcal X_R=\{x\in\{0,1\}^n:Q(x)\leq R \quad \mbox{and} \quad
                   x_j=0\text{ for every }j\in\mathcal D\}.
\end{equation}        
Let $P_R$ project onto the span of the retained Krylov Fock states,
\[
 P_R=U\left(\sum_{x\in\mathcal X_R}\ket{x}\bra{x}\right)U^\dagger.
\]
Let $\psi$ be any fixed normalized ground state of $\widetilde H$
and $p_R$ be its weight outside this retained subspace,
\[
p_R=\norm{(I-P_R)\psi}^2
   =1-\sum_{x\in\mathcal X_R}|\bra{x}U^\dagger\ket\psi|^2.
\]
Our goal is to choose $R$ large enough that $p_R$ is negligible and, at the same time, the dimension $|\mathcal X_R|$ is 
sufficiently small.

\begin{proposition}[\bf Tail and dimension bounds]\label{ge:tail}
For every integer $R\geq0$,
\begin{equation}\label{ge:tail-dimension}
 p_R\leq e^{mS} 3^{-R}\quad \mbox{and} \quad
 |\mathcal X_R|\leq e^{(\log2+1/2)mS}3^R.
\end{equation}
\end{proposition}
\begin{proof}
For any numbers $z_j\geq1$, $j\in\mathcal J$, expand the product
of commuting number projections and apply \eqref{ge:joint-depth}:
\begin{align}
 \bra{\psi}\prod_{j\in\mathcal J}z_j^{n_j}\ket{\psi}
 &=\sum_{I\subseteq\mathcal J}
     \left(\prod_{j\in I}(z_j-1)\right)
       \bra{\psi}\prod_{j\in I}n_j\ket{\psi}\leq\prod_{j\in\mathcal J}
        \left[1+(z_j-1)\tfrac14 9^{-\ell(j)}\right].
 \label{ge:generating}
\end{align}
All coefficients used in the inequality are nonnegative.
Set $z_j=3^{\ell(j)+1}$. Since there are at most $mS$ modes per Krylov depth,
\[
 \log\bra{\psi}3^Q\ket{\psi}
 \leq\frac {mS}4\sum_{\ell=0}^\infty(3^{\ell+1}-1)9^{-\ell}
 =\frac{27mS}{32}<mS.
\]
Decoupled modes are empty, so $p_R$ is the probability of measuring
$Q>R$. Markov's inequality proves the first bound in
\eqref{ge:tail-dimension}. For every retained configuration $x\in {\mathcal X}_R$ one has  $1\leq3^{R-Q(x)}$.
Sum this inequality over all $x\in {\mathcal X}_R$ and then allow arbitrary layer occupations:
\begin{align*}
 |\mathcal X_R| &\leq2^{n_A}3^R\prod_{j\in\mathcal J}
                  (1+3^{-(\ell(j)+1)})\leq2^{n_A}3^R\exp\left(mS\sum_{\ell=0}^\infty3^{-(\ell+1)}\right)
 \leq e^{(\log2+1/2)mS}3^R.
\end{align*}
\end{proof}
Let $n^*(R)$ be the number of modes $j\in [n]$ such that $x_j=1$
for some $x\in {\mathcal X}_R$. In other words,
$n^*(R)$ is the number of modes one needs to express the set of retained Fock configurations
${\mathcal X}_R$, ignoring $|0\rangle$ modes.
By definition of ${\mathcal X}_R$, only 
 layers of Krylov  depth $0,\ldots,R-1$ can be occupied.  Thus
 \begin{equation}\label{ge:active-count}
 n^*(R)=n_A+\sum_{s=0}^{S-1}\;
            \sum_{\ell=0}^{\min\{R,n_B\}-1}\dim\W_{s,\ell}
 \leq n_A+mSR\leq mS(R+1),
\end{equation}
with an empty inner sum when $R=0$. The retained space contains only $|{\mathcal X}_R|$
Fock configurations, rather than all $2^{n^*(R)}$ configurations
on these modes.

\subsection{Ground-energy approximation}
\label{sec:groundenergyapprox}
Since the basis change is unitary on the one-particle space,
\[
 N=\sum_{j=1}^n \beta_j^\dagger \beta_j=\sum_{j=1}^n f_j^\dagger f_j.
\]
The number $N$ counts bath quasiparticles in the diagonal convention
of the main-results section; it need not coincide with a particle
number defined before the initial Bogoliubov transformation.

\begin{lemma}\label{ge:energy-lemma}
If $R\geq0$ is an integer and $p_R<1$, define
$\phi_R=P_R\psi/\sqrt{1-p_R}$. Then
\begin{equation}\label{ge:energy-bound}
 0\leq\bra{\phi_R}\widetilde H\ket{\phi_R}-E_1(\widetilde H)
 \leq\frac{\sqrt{p_R}}{1-p_R}(mS+R).
\end{equation}
\end{lemma}
\begin{proof}
The interaction is even and supported on the boundary, whose
occupations are unrestricted. Hence $[P_R,V]=0$.
Every retained configuration has at most $n_A+R\leq mS+R$ particles,
so $\norm{NP_R}\leq mS+R$.
The commuting operators $\widetilde H_0$ and $N$ satisfy
$0\preceq\widetilde H_0\preceq N$ and
$\widetilde H_0^2\preceq N^2$, since each $\widetilde e_j$ lies
in $[0,1]$. Therefore
\begin{equation}\label{ge:H0P}
 \norm{\widetilde H_0P_R}\leq mS+R.
\end{equation}
Put $\chi_R=(\Id-P_R)\psi$. The ground-state equation gives
\begin{align}
 &(1-p_R)\bigl(\bra{\phi_R}\widetilde H\ket{\phi_R}
                             -E_1(\widetilde H)\bigr)=-\bra{\chi_R}\widetilde H P_R\ket{\psi}
              =-\bra{\chi_R}\widetilde H_0P_R\ket{\psi}.
 \label{ge:energy-identity}
\end{align}
The interaction term vanishes by commutation with $P_R$.
The variational principle makes the left-hand side nonnegative.
Taking absolute values on the right and using
$\norm{\chi_R}=\sqrt{p_R}$ and \eqref{ge:H0P} proves the result.
\end{proof}

\begin{proposition}\label{ge:R-proposition}
For $\varepsilon\in (0,1/2],$  let $R=4\left(mS+\left\lceil\log(16/\varepsilon)\right\rceil\right).
$
Then $p_R<1/2$~and
\begin{equation}\label{ge:R-bounds}
 \bra{\phi_R}\widetilde H\ket{\phi_R}-E_1(\widetilde H)
 \leq\varepsilon/4,
 \qquad |\mathcal X_R|\leq81e^{6mS}(16/\varepsilon)^7.
\end{equation}
\end{proposition}
\begin{proof}
Let $t=\lceil\log(16/\varepsilon)\rceil$ and $u=mS+t$.
Then $e^{-t}\leq\varepsilon/16$. Since $\log3>1$ and
$R=4u$, Proposition~\ref{ge:tail} gives
\[
 p_R\leq e^{mS-R}\leq e^{-3mS-4t}<1/2.
\]
Also $mS+R\leq6u$. Lemma~\ref{ge:energy-lemma}
bounds the energy loss by
\begin{align*}
 2(mS+R)\sqrt{p_R}
 &\leq12u e^{-3mS/2-2t}\leq\frac\varepsilon4\left(3u e^{-u}e^{-mS/2}\right)
 \leq\frac\varepsilon4\,3e^{-3/2}<\frac\varepsilon4.
\end{align*}
We used $u e^{-u}\leq e^{-1}$ for $u>0$ and $mS\geq1$.
For the dimension, substitute $R\leq4mS+4\log(16/\varepsilon)+4$ into
\eqref{ge:tail-dimension}:
\[
 |\mathcal X_R| \leq81e^{(\log2+1/2+4\log3)mS}(16/\varepsilon)^{4\log 3}
 \leq81e^{6mS}(16/\varepsilon)^7,
\]
proving the proposition statement.
\end{proof}

\subsection{Algorithm and analysis}
\label{ge:algorithm}

\paragraph{Algorithm.}
Start from the input Hamiltonian $H=H_0+V$, represented as
described in Section~\ref{sec:impurity-model}.

\begin{enumerate}
\item Set $\omega=\varepsilon/(4m)$ and construct the rounded bath
energies \eqref{ge:rounded-H}.

\item Construct the shelled Krylov decomposition from
Section~\ref{ge:shelled-decomposition}, with 
$S=1+\left\lfloor\log_2(1/\omega)\right\rfloor$ energy shells and
the data in
\eqref{ge:template-data}.
Set $R=4\left(mS+\left\lceil\log(16/\varepsilon)\right\rceil\right).$

\item Compute an orthonormal basis of $\A$. In each shell, compute
the Krylov layers by
starting from $\ran T_s^\dagger$, successively applying $E_s$,
and orthogonalizing until the Krylov space becomes invariant.
Compute an orthonormal basis of the decoupled complement $\Z_s$.
Together these vectors give the full ordered one-particle basis
$\xi_1,\ldots,\xi_n$ defined above and the Bogoliubov transformation
$U$ of~\eqref{ge:bogoliubov-U}. Retain all boundary modes and
Krylov-layer modes of depth less than $R$; require all other modes
to be empty.

\item Using Lemma~\ref{lem:enum},
enumerate all configurations in $\mathcal X_R$, with
arbitrary boundary occupations. Write their binary labels as
$x_1,\ldots,x_D$, where 
\[
D=|\mathcal X_R|.
\]

\item Form the matrix of the original Hamiltonian $H$ projected onto the span
of the
retained Krylov Fock states:
\begin{equation}\label{ge:matrix}
 B_{\mu\nu}=\bra{x_\mu}U^\dagger HU\ket{x_\nu}
 \qquad(\mu,\nu\in[D]).
\end{equation}
Here $B$ is a many-body matrix; $\B$ is the residual one-particle
space.
\item Compute a normalized vector $z\in\C^D$ such that
\begin{equation}\label{ge:eigensolver}
 z^\dagger Bz\leq E_1(B)+\varepsilon/4,
\end{equation}
where $E_1(B)$ is the smallest eigenvalue of $B$ and $\varepsilon/4$ accounts for eigensolver error.
Define the normalized sparse state
\begin{equation}\label{ge:trial-state}
 \ket\phi=\sum_{\mu=1}^D z_\mu\ket{x_\mu}.
\end{equation}
Return $U$, $\phi$, and the energy
\begin{equation}\label{ge:reporting}
 \widehat E_1:=z^\dagger Bz=\bra\phi U^\dagger HU\ket\phi.
\end{equation}
Thus $U\ket{\phi}$ is the physical trial state, whereas $\phi$ is specified
by its $D$ coefficients in the original Fock basis.
\end{enumerate}

Computing the full one-particle basis makes $U$ explicit and costs
only polynomial time in $n$. No ground state is used to construct
$U$ or enumerate $\mathcal X_R$.

\paragraph{Correctness.}
Rounding and the choice of $\omega$ give
\[
 H\preceq\widetilde H,\qquad
 E_1(\widetilde H)\leq E_1(H)+\varepsilon/4.
\]
Proposition~\ref{ge:R-proposition} supplies the normalized vector
$\phi_R\in\ran P_R$ with rounded energy at most
$E_1(\widetilde H)+\varepsilon/4$. Since $B$ is a projection of $H$ onto
$\ran P_R$, the variational principle gives
\begin{align*}
 E_1(H)\leq E_1(B)
 &\leq\bra{\phi_R}H\ket{\phi_R}\\
 &\leq\bra{\phi_R}\widetilde H\ket{\phi_R}\\
 &\leq E_1(\widetilde H)+\varepsilon/4
 \leq E_1(H)+\varepsilon/2.
\end{align*}
Combining this with~\eqref{ge:eigensolver}--\eqref{ge:reporting} yields
\[
 E_1(H)\leq\widehat E_1
 =\bra\phi U^\dagger HU\ket\phi
 \leq E_1(H)+3\varepsilon/4\leq E_1(H)+\varepsilon.
\]
The state $\phi$ is normalized because $z$ is normalized and the
original Fock basis is orthonormal. It has at most $D=|\mathcal X_R|$
nonzero coefficients. Together with~\eqref{ge:bogoliubov-U}, this
proves both the energy and state-output claims in
Theorem~\ref{thm:ground-energy-algorithm}.

\paragraph{Constructing the projected matrix from the impurity input}
\label{ge:matrix-construction}
For this matrix construction, choose an orthonormal basis
$w_1,\ldots,w_m$ of $\Aimp$. For $X,Y\subseteq[m]$, the
normal-ordered monomials
\begin{equation}\label{ge:impurity-monomials}
 \left(\prod_{\ell\in X}\beta(w_\ell)^\dagger\right)
 \left(\prod_{\ell'\in Y}^{\leftarrow}\beta(w_{\ell'})\right)
\end{equation}
form a basis of the operator space on the $m$ modes spanning $\Aimp$. The
first product has increasing indices, and the arrow specifies
decreasing indices in the second product. There are $4^m$ monomials,
each of length at most $2m$. The input matrix can be converted to
this basis by finite-dimensional linear algebra using $\exp[O(m)]$
arithmetic operations. Since $V$ is even, only monomials with
$|X|+|Y|$ even have nonzero coefficients. Each $w_\ell$ belongs to $\A$, so
\[
 \beta(w_\ell)=\sum_{\gamma=1}^{n_A}
                 \overline{\ip{u_\gamma}{w_\ell}}f_\gamma
\]
for $\ell\in [m]$. Expanding all monomials \eqref{ge:impurity-monomials} therefore gives
at most
\begin{equation}\label{ge:impurity-term-count}
 \sum_{X\subseteq[m]}\sum_{Y\subseteq[m]}n_A^{|X|+|Y|}
 =(1+n_A)^{2m}
\end{equation}
ordered products of boundary ladder operators. Since
$U^\dagger f_jU=a_j$, the matrix elements of these products between
$U\ket{x_\mu}$ and $U\ket{x_\nu}$ can be evaluated by applying the
corresponding original-mode ladder operators to $\ket{x_\nu}$ from
right to left, including the CAR signs.
For the quadratic contribution, use the original one-particle matrix
$h$ in the constructed basis:
\[
 H_0=\sum_{i,j=1}^n\ip{\xi_i}{h\xi_j}f_i^\dagger f_j.
\]
Only terms with both indices on retained modes can have nonzero
matrix elements between retained configurations, so at most
$(n^*(R))^2$ quadratic terms contribute. This constructs every entry
of~\eqref{ge:matrix} using $H$, while the basis itself is constructed
using $\widetilde H$.

\subsubsection{Runtime}\label{ge:cost}

The one-particle matrix calculations, construction of the full
ordered basis, and computation of the linear action of $U$ cost
$\poly(n,m,R)$ arithmetic operations. By Lemma~\ref{lem:enum}, the retained configurations can be enumerated
in $O((n^*(R)+1)|\mathcal X_R|)$ operations, assigning weight zero to boundary modes
and weight $\ell+1$ to bath modes in layer $\ell$.
A coarse sufficient bound for the projected matrix construction is
\[
 \poly(n^*(R),m)
       \bigl((n^*(R))^2+(1+n_A)^{2m}\bigr)|\mathcal X_R|^2+\exp[O(m)].
\]
Since $n_A\leq mS$, this obeys the claimed cost bound. Computing a
normalized vector satisfying~\eqref{ge:eigensolver} has arithmetic cost
polynomial in $D\equiv |\mathcal X_R|$, $\max{(1,\|B\|)}$, and $1/\varepsilon$.
Evaluating~\eqref{ge:reporting} costs $O(D^2)$ additional operations.
Since $0\preceq H_0\preceq N$ and every retained configuration has
at most $mS+R$ quasiparticles,
$\norm{P_RH_0P_R}\leq mS+R$, and hence
$\norm B\leq mS+R+\norm V$.
Together these bounds give runtime
\[
 \poly(n,1+\|V\|)\exp[O(m\log(m/\varepsilon))]
\]
and complete the proof of
Theorem~\ref{thm:ground-energy-algorithm}.

\section{Gaussian decomposition of thermofield-double state}
\label{tf:section}

In this section we turn to the thermal Gibbs state of quantum impurity models and show that it admits an efficient classical description. In particular, we show that the purification of the Gibbs state, i.e., the thermofield-double state, as defined in Section~\ref{sec:tfd-preliminaries} can be approximated by a compact superposition of fermionic Gaussian states and thereby prove Theorem~\ref{thm:rank}. 

The infinite-temperature or $\beta=0$ case is trivial as the state $\ket{\TFD_0(H)}=2^{-n/2}\ket{I_f}$
is independent of $H$. Define
\begin{equation}\label{tf:beta-zero}
 U_0=\prod_{j=1}^n
 \exp\!\left[-\frac{i\pi}{4}
 \left(\beta_j^\dagger\widetilde\beta_j^\dagger+
       \widetilde\beta_j\beta_j\right)\right],
 \qquad \mathcal X_{0,\delta}=\{0^{2n}\}.
\end{equation}
and observe that quadratic generators act on disjoint mode pairs, implying
$U_0\ket{\Omega_{2n}}=2^{-n/2}\ket{I_f}=\ket{\TFD_0(H)}$.
Thus Theorem~\ref{thm:rank} holds at $\beta=0$ with a singleton
configuration set and $O(n)$ one-particle operations.

For the rest of this section, we assume that $\beta>0$.
We write $\delta$ for the error tolerance in Theorem~\ref{thm:rank}. 
We use the Gaussian purification of a free thermal bath as our reference state.
A Bogoliubov transformation makes this purification the vacuum of a
set of doubled fermionic modes. Their occupations therefore measure
deviations from the free thermal state. We bound joint occupations of
these modes by a thermal version of multiparticle pull-through.
Its proof integrates boundary sources over the finite imaginary-time
interval $[0,\beta/2]$. H\"older's inequality controls the
resulting products without any dependence on $\norm{V}$.
In our proof, we instantiate the Krylov template of Section~\ref{sec:krylov-decompositions}
and obtain exponential decay
with Krylov depth. The resulting weighted truncation gives
a set of Fock configurations of size proportional
to $1/\delta$ with discarded probability at most $\delta^2$.

\subsection{Finite-temperature Krylov decomposition}
\label{tf:shells}

\subsubsection{Energy shells and the impurity boundary}
We work with the unrounded one-particle Hamiltonian
$h=\operatorname{diag}(e_1,\ldots,e_n)$.
The shell construction is analogous to Section~\ref{sec:shell1},
but here we retain a low-energy shell containing all zero modes.
Set 
\[
a_0=\min\{1,\beta^{-1}\} \quad \mbox{and} \quad
S=1+\lceil\log_2(1/a_0)\rceil.
\]
Define
\begin{equation}\label{tf:intervals}
 I_0=[0,a_0],\qquad
 I_s=(2^{s-1}a_0,\min\{2^sa_0,1\}]
 \quad(1\leq s<S).
\end{equation}
These intervals are disjoint and cover $[0,1]$. Empty shells are allowed.
Let
\begin{equation}\label{tf:shell-projectors}
 \Pi_s=\sum_{\substack{j\in[n]\\e_j\in I_s}}|j\rangle\langle j|,
 \qquad \mathcal H_s=\ran\Pi_s.
\end{equation}
The boundary and residual bath in shell $s$ are
\begin{equation}\label{tf-boundary-shells}
 \A_s=\Pi_s\Aimp,\qquad
 \B_s=\mathcal H_s\ominus\A_s.
\end{equation}
Recall that \(\Aimp\subseteq\C^n\) denotes the one-particle impurity subspace, i.e., the span of the one-particle directions on which \(V\) is supported. 
Define
\begin{equation}\label{tf-boundary}
 \A=\bigoplus_{s=0}^{S-1}\A_s,\qquad
 \B=\bigoplus_{s=0}^{S-1}\B_s,\qquad
 n_A=\dim\A,\quad n_B=\dim\B=n-n_A.
\end{equation}
As in the ground-energy proof, $\dim\A_s\leq m$ and
$\Aimp\subseteq\A$. In particular,
\begin{equation}\label{tf:shell-count}
 S=1+\left\lceil\log_2\max\{1,\beta\}\right\rceil,
 \qquad n_A\leq mS,\qquad \B\subseteq\Aimp^\perp.
\end{equation}
Thus $V$ is supported entirely on $\A$.

\subsubsection{Boundary--bath decomposition}
Relative to $\C^n=\A\oplus\B$, write the one-particle Hamiltonian $h$ as a block matrix
\[
 h=\begin{pmatrix}
 h_{\A\A}&h_{\A\B}\\
 h_{\A\B}^{\dagger}&h_{\B\B}
 \end{pmatrix}.
\]
We instantiate the abstract Krylov decomposition template of Section~\ref{sec:krylov-decompositions}  with
 the spaces above, the given
inverse temperature $\beta$, and
\begin{equation}\label{tf:blocks}
 E=h_{\B\B},\qquad
 T=h_{\A\B}(\Id+e^{-\beta E})^{-1/2}.
\end{equation}
Each $\B_s$ is invariant under $E$, and $0\preceq E\preceq\Id_{\B}$.
The template's shell projectors are $P_s=\Pi_s|_{\B}$,
and its shell maps are $E_s=E|_{\B_s}$ and $T_s=T|_{\B_s}$.
Since the factor $(\Id+e^{-\beta E})^{-1/2}$ preserves each $\B_s$,
$\ran T_s\subseteq\A_s$ and $\rank T_s\leq m$.

\begin{lemma}[Thermal shell block bounds]\label{tf:thermal-block-lemma}
For $1\leq s<S$,
\begin{equation}\label{tf:block-bounds}
 2^{s-1}a_0\Id_{\B_s}\preceq E_s
 \preceq\min\{2^sa_0,1\}\Id_{\B_s},
 \qquad \norm{T_s}\leq\frac{2^{s-1}a_0}{2}.
\end{equation}
For the low-energy shell,
\begin{equation}\label{tf:low-block-bounds}
 0\preceq E_0\preceq a_0\Id_{\B_0},
 \qquad \norm{T_0}\leq a_0/2.
\end{equation}
\end{lemma}
\begin{proof}
Compression gives the bounds on $E_s$. Subtracting the interval
midpoint, as in Lemma~\ref{ge:block-lemma}, gives
$\norm{h_{\A\B}|_{\B_s}}\leq2^{s-1}a_0/2$ for $s\geq1$
and $\norm{h_{\A\B}|_{\B_0}}\leq a_0/2$.
The same bounds hold for $T_s$, since
$(\Id+e^{-\beta E_s})^{-1/2}$ is a contraction.
\end{proof}

Choose orthonormal bases $u_1,\ldots,u_{n_A}$ of $\A$
and $v_1,\ldots,v_{n_B}$ of $\B$, and set
$a_\ell=\beta(u_\ell)$ and $b_j=\beta(v_j)$.
Write $a(u)=\beta(u)$ for $u\in\A$ and $b(v)=\beta(v)$ for $v\in\B$.
Matrix entries below are taken in these bases.
The original impurity model Hamiltonian becomes
\begin{align}
 H&=H_A+H_B+
 \sum_{\ell=1}^{n_A}\sum_{j=1}^{n_B}
 \left((h_{\A\B})_{\ell j}a_\ell^\dagger b_j+
 \overline{(h_{\A\B})_{\ell j}}b_j^\dagger a_\ell\right),\notag\\
 H_B&=\sum_{i,j=1}^{n_B}E_{ij}b_i^\dagger b_j, \quad  H_A=V+\sum_{\ell,\ell'=1}^{n_A}
 (h_{\A\A})_{\ell\ell'}a_\ell^\dagger a_{\ell'}.
 \label{tf:boundary-H}
\end{align}
For all $v,w\in\B$, the canonical anticommutation relations imply
\begin{align}
 [H,b(v)]&=-b(Ev)-a(h_{\A\B}v),
 \qquad\{a(h_{\A\B}v),b(w)\}=0.
 \label{tf:commutators}
\end{align}
The boundary--bath anticommutators also vanish if either or both
operators are replaced by their adjoints.

\subsection{Annihilators adapted to the free thermal bath}
\label{tf:annihilators}

Below we use the definitions of Section~\ref{sec:tfd-preliminaries} with $H_{\mathrm{free}}$ replaced by $H_B$ and $n$ replaced by $n_B$, that is, we regard $H_B$ as a Hamiltonian on the $n_B$ residual-bath modes $b_1,\ldots,b_{n_B}$ and its thermofield double as a state of these modes and their ancillas.

Let us fix the orthonormal basis $v_1,\ldots,v_{n_B}\in \B$ chosen above to be an eigenbasis of $E$, so that $Ev_j=\lambda_jv_j$ with $\lambda_j\in[0,1]$.
The corresponding physical
and ancilla annihilators are
\[
 b_j=\sum_{k=1}^n \overline{(v_j)_k} \beta_k
 \quad \mbox{and}\quad
 \widetilde b_j=\sum_{k=1}^n(v_j)_k\widetilde\beta_k
\]
respectively. For $v\in\B$, define {\em free TFD annihilators}
\[
 \mathsf d_-(v)=\sum_{j=1}^{n_B}
                   \overline{\langle v_j,v\rangle}\mathsf d_{-,j},
 \qquad
 \mathsf d_+(v)=\sum_{j=1}^{n_B}
                   \langle v_j,v\rangle \mathsf d_{+,j},
\]
where
\begin{equation}\label{tf:bogoliubov}
 \mathsf d_{-,j}
 =\frac{b_j+i e^{-\beta\lambda_j/2}\widetilde b_j^\dagger}
        {\sqrt{1+e^{-\beta\lambda_j}}},
 \qquad
 \mathsf d_{+,j}
 =\frac{e^{-\beta\lambda_j/2}b_j^\dagger+i\widetilde b_j}
        {\sqrt{1+e^{-\beta\lambda_j}}}.
\end{equation}
The minus family is antilinear in $v$, and the plus family is
linear in $v$. The operators $ \mathsf d_\pm(v)$ with $v\in \B$
annihilate the thermofield-double state of $H_B$, see Section~\ref{sec:tfd-preliminaries}.

The canonical anticommutation relations of the physical and ancilla
modes give
\[
 \{\mathsf d_{\eta,i},\mathsf d_{\zeta,j}\}=0,\qquad
 \{\mathsf d_{\eta,i},\mathsf d_{\zeta,j}^\dagger\}
 =\delta_{\eta\zeta}\delta_{ij}\Id
\]
for $\eta,\zeta\in\{-,+\}$ and $i,j\in[n_B]$.
Consequently, any orthonormal basis of $\B$ gives $2n_B$ mutually
canonical free TFD annihilators.
Together with the physical and ancilla
boundary modes, these form a complete set of $2n$ annihilators.

We next express the action of free TFD annihilators  on vectorized operators.
For a homogeneous operator $Y$ of parity $p(Y)\in \{0,1\}$, the right-action
identities \eqref{tf:right-action} give
\begin{align}
 \mathsf d_-(v)\ket{Y}_f
 &=\ket{b\bigl((\Id+e^{-\beta E})^{-1/2}v\bigr)Y}_f
-(-1)^{p(Y)}
 \ket{Yb\bigl(e^{-\beta E/2}(\Id+e^{-\beta E})^{-1/2}v\bigr)}_f,
 \notag\\
 \mathsf d_+(v)\ket{Y}_f
 &=\ket{b\bigl(e^{-\beta E/2}(\Id+e^{-\beta E})^{-1/2}v\bigr)^\dagger Y}_f
 -(-1)^{p(Y)}
 \ket{Yb\bigl((\Id+e^{-\beta E})^{-1/2}v\bigr)^\dagger}_f.
 \label{tf:annihilators-def}
\end{align}
These identities extend linearly to arbitrary $Y$.
Using the isometry \eqref{tf:isometry}, we also write
$\mathsf d_\pm(v)Y$ for the operator whose vectorization is
$\mathsf d_\pm(v)\ket{Y}_f$. 
When $h_{\A\B}=0$, any free TFD annihilator 
annihilates $\ket{\mathsf{TFD}_\beta(H)}$.

\subsection{Thermal multiparticle pull-through}
\label{tf:pull-through}

The next lemma bounds joint occupations of the free TFD modes
in terms of their integrated couplings. It is a finite-temperature
analogue of the ground-state pull-through bound in Lemma~\ref{ge:pull}.

\begin{lemma}[Thermal multiparticle pull-through]\label{tf:pull}
Let $s_\beta(v)$ be the integrated coupling defined in
\eqref{kd:integrated-coupling}, with $E,T$ from \eqref{tf:blocks}.
For every integer $k\geq0$, vectors $v_1,\ldots,v_k\in\B$, and
signs $\eta_1,\ldots,\eta_k\in\{-,+\}$, one has
\begin{equation}\label{tf:product-bound}
 \norm{\mathsf d_{\eta_1}(v_1)\cdots
       \mathsf d_{\eta_k}(v_k)\ket{\mathsf{TFD}_\beta(H)}}
 \leq\prod_{i=1}^k s_\beta(v_i).
\end{equation}
If the vectors $v_i$ and $v_j$ are orthogonal whenever $i\ne j$ and $\eta_i=\eta_j$ then the corresponding
number operators 
$\mathsf n_i=\mathsf d_{\eta_i}(v_i)^\dagger
                 \mathsf d_{\eta_i}(v_i)$ with $i=1,\ldots,k$ pairwise commute and obey
\begin{equation}\label{tf:joint-bound}
 \bra{\mathsf{TFD}_\beta(H)}\prod_{i=1}^k\mathsf n_i
                  \ket{\mathsf{TFD}_\beta(H)}
 \leq\prod_{i=1}^k s_\beta(v_i)^2.
\end{equation}
\end{lemma}

The above lemma applies to any orthogonal decomposition
$\C^n=\A\oplus\B$ with $V$ supported on $\A$, using the corresponding
choice of $E,T$ in \eqref{tf:blocks}; it does not require the specific
shell partition of Section~\ref{tf:shells}.

We first use H\"older's inequality for Schatten norms to bound
products of imaginary-time evolution operators with boundary
operators inserted between them.
We use the Schatten norms
$\norm{Y}_p=(\Tr|Y|^p)^{1/p}$ for $1\le p<\infty$, where
$|Y|=(Y^\dagger Y)^{1/2}$, and $\norm{Y}_\infty=\norm{Y}$.
In particular,
$\norm{Y}_2=\norm{Y}_{\mathrm{HS}}$.
The notation $\norm{\cdot}_{\mathrm{HS}}$ denotes the
Hilbert--Schmidt  norm (which coincides with the Frobenius norm).

\begin{lemma}[Thermal product norm]\label{tf:holder}
For $k\geq0$, let $t_0,\ldots,t_k\geq0$ satisfy
$\sum_{j=0}^k t_j=\beta/2$,
and let $C_1,\ldots,C_k$ be arbitrary operators on $\calF_n$.
Then
\begin{equation}\label{tf:holder-bound}
 \norm{e^{-t_0H}C_1e^{-t_1H}\cdots C_ke^{-t_kH}}_{\mathrm{HS}}
 \leq\sqrt{Z_\beta(H)}\prod_{j=1}^k\norm{C_j}.
\end{equation}
More generally, for every $p\in[1,\infty]$,
\begin{equation}\label{tf:holder-bound-p}
 \norm{e^{-t_0H}C_1e^{-t_1H}\cdots C_ke^{-t_kH}}_p
 \leq\norm{e^{-\beta H/2}}_p\prod_{j=1}^k\norm{C_j}.
\end{equation}
\end{lemma}

\begin{proof}
We prove \eqref{tf:holder-bound-p} first.
Suppose $\beta>0$ and $1\leq p<\infty$.
H\"older's inequality for Schatten norms~\cite{bhatia97book} states that, for
arbitrary operators $O_1,\ldots,O_L$ on $\calF_n$ and exponents
$p,q_1,\ldots,q_L\in[1,\infty]$ satisfying
\[
 \frac{1}{q_1}+\cdots+\frac{1}{q_L}=\frac{1}{p},
\]
one has
\[
 \norm{O_1\cdots O_L}_p
 \leq \prod_{j=1}^L\norm{O_j}_{q_j}.
\]
Here $q_j$ specifies the Schatten norm used for the factor
$O_j$, with the conventions $1/\infty=0$ and
$\norm{O_j}_\infty=\norm{O_j}$.

For each factor $e^{-t_jH}$ with $t_j>0$, choose exponent
$p\beta/(2t_j)$. For each inserted operator $C_j$, choose
exponent $\infty$, so its contribution is its operator norm.
If $t_j=0$, the factor $e^{-t_jH}$  is the identity; assign it exponent
$\infty$ as well. These choices are admissible because
\[
 \sum_{\substack{0\leq j\leq k\\t_j>0}}
 \frac{2t_j}{p\beta}
 =\frac{1}{p}.
\]
Thus
\begin{equation}\label{tf:holder-step}
 \norm{e^{-t_0H}C_1e^{-t_1H}\cdots C_ke^{-t_kH}}_p
 \leq
 \prod_{\substack{0\leq j\leq k\\t_j>0}}
 \norm{e^{-t_jH}}_{p\beta/(2t_j)}
 \prod_{j=1}^k\norm{C_j}.
\end{equation}
Since $H$ is Hermitian,
$e^{-t_jH}$ is positive definite, and thus
\[
 \norm{e^{-t_jH}}_{p\beta/(2t_j)}
 =
 \left(\Tr e^{-p\beta H/2}\right)^{2t_j/(p\beta)}.
\]
Multiplying these expressions and using
$\sum_j t_j=\beta/2$ gives
\[
 \prod_{\substack{0\leq j\leq k\\t_j>0}}
 \norm{e^{-t_jH}}_{p\beta/(2t_j)}
 =
 \left(\Tr e^{-p\beta H/2}\right)^{1/p}
 =
 \norm{e^{-\beta H/2}}_p.
\]
Substitution into \eqref{tf:holder-step} proves
\eqref{tf:holder-bound-p} for finite $p$.

For $p=\infty$, use submultiplicativity of the operator norm.
All  factors $e^{-t_jH}$ attain their operator norm on a ground state
of $H$, so
\[
 \prod_{j=0}^k\norm{e^{-t_jH}}
 =
 e^{-(\beta/2)E_1(H)}
 =
 \norm{e^{-\beta H/2}}.
\]
This proves the same bound for $p=\infty$.

If $\beta=0$, all $t_j$ vanish. Repeatedly applying
$\norm{XY}_p\leq\norm{X}_p\norm{Y}$, starting with $X=\Id$,
gives
\[
 \norm{C_1\cdots C_k}_p
 \leq\norm{\Id}_p\prod_{j=1}^k\norm{C_j},
\]
which is precisely \eqref{tf:holder-bound-p} in this case. Finally, taking $p=2$ gives \eqref{tf:holder-bound}, since
\[
 \norm{e^{-\beta H/2}}_2
 =\sqrt{\Tr e^{-\beta H}}
 =\sqrt{Z_\beta(H)}.
\]
\end{proof}

\begin{proof}[Proof of Lemma~\ref{tf:pull}]
We prove the lemma by induction. In the base case $k=0$, no operators are applied to the LHS and \eqref{tf:product-bound} holds because the state is normalized. Now, fix $k\geq1$ and the vectors and
signs as in the statement. 
For $i\in[k]$ and $0\leq t\leq\beta/2$, set
\begin{equation}\label{tf:odd-paths}
 O_i(t)=
 \begin{cases}
 b(e^{-tE}(\Id+e^{-\beta E})^{-1/2}v_i),&\eta_i=-,\\
 b(e^{-(\beta/2-t)E}(\Id+e^{-\beta E})^{-1/2}v_i)^\dagger,&\eta_i=+.
 \end{cases}
\end{equation}
These odd bath operators satisfy $\dot O_i(t)=[H_B,O_i(t)]$.
Their boundary sources are
\begin{equation}\label{tf:sources}
 C_i(t)=[H,O_i(t)]-\dot O_i(t)=
 \begin{cases}
 -a(Te^{-tE}v_i),&\eta_i=-,\\
 a(Te^{-(\beta/2-t)E}v_i)^\dagger,&\eta_i=+.
 \end{cases}
\end{equation}
Boundary--bath anticommutation and a change of integration variable give
\begin{equation}\label{tf:source-properties}
 \{O_i(t),C_j(u)\}=0,\qquad
 \int_0^{\beta/2}\norm{C_i(t)}\,dt=s_\beta(v_i)
 \quad(i,j\in[k],\ 0\leq t,u\leq\beta/2).
\end{equation}

For $0\leq s\leq t\leq\beta/2$ and homogeneous $Y$, define
\begin{equation}\label{tf:endpoint-map}
 \mathcal D_i^{s,t}(Y)=O_i(s)Y-(-1)^{p(Y)}YO_i(t).
\end{equation}
At the full interval, $\mathcal D_i^{0,\beta/2}=\mathsf d_{\eta_i}(v_i)$.
Differentiating
$e^{-(u-s)H}O_i(u)e^{-(t-u)H}$ and integrating over $u\in[s,t]$
gives
\begin{equation}\label{tf:duhamel}
 \mathcal D_i^{s,t}(e^{-(t-s)H})
 =\int_s^t e^{-(u-s)H}C_i(u)e^{-(t-u)H}\,du.
\end{equation}
For homogeneous $Y,Z$ and $s\leq u\leq t$, direct expansion gives
\begin{equation}\label{tf:product-rule}
 \mathcal D_i^{s,t}(YZ)
 =\mathcal D_i^{s,u}(Y)Z
      +(-1)^{p(Y)}Y\mathcal D_i^{u,t}(Z).
\end{equation}
When the map acts on a source already inserted at time $u$, it gives
$\mathcal D_i^{u,u}(C_j(u))=\{O_i(u),C_j(u)\}=0$.
Thus each application inserts exactly one new boundary source into
one of the exponential intervals.

Let $\mathfrak S_k$ be the permutations of $[k]$ and set $t_0=0$.
Induction using \eqref{tf:duhamel}--\eqref{tf:product-rule} yields
\begin{align}
 &\mathsf d_{\eta_1}(v_1)\cdots
       \mathsf d_{\eta_k}(v_k)e^{-\beta H/2}\notag\\
 &=\sum_{\pi\in\mathfrak S_k}\operatorname{sgn}(\pi)
   \int_{0<t_1<\cdots<t_k<\beta/2}
   \left[\prod_{j=1}^k
       e^{-(t_j-t_{j-1})H}C_{\pi(j)}(t_j)\right]
       e^{-(\beta/2-t_k)H}\,dt_1\cdots dt_k.
 \label{tf:ordered-expansion}
\end{align}
The product has increasing $j$ from left to right. Inserting a new
source past each preceding odd source gives one minus sign, producing
the permutation sign in this formula.

Take Hilbert--Schmidt norms in \eqref{tf:ordered-expansion}, apply
Lemma~\ref{tf:holder}, divide by $\sqrt{Z_\beta(H)}$, and use the vectorization isometry. This gives
\begin{align}
 &\norm{\mathsf d_{\eta_1}(v_1)\cdots
       \mathsf d_{\eta_k}(v_k)
       \ket{\mathsf{TFD}_\beta(H)}}\notag\\         
 &\quad\leq\sum_{\pi\in\mathfrak S_k}
       \int_{0<t_1<\cdots<t_k<\beta/2}
                 \prod_{j=1}^k\norm{C_{\pi(j)}(t_j)}\,dt_1\cdots dt_k\notag\\
 &\quad=\prod_{i=1}^k\int_0^{\beta/2}\norm{C_i(t)}\,dt
        =\prod_{i=1}^k s_\beta(v_i).
 \label{tf:cube-bound}
\end{align}
The equality holds because the $k!$ time orderings partition
$[0,\beta/2]^k$, up to sets of measure zero. This proves
\eqref{tf:product-bound}. 

Under the stated orthogonality condition, the anticommutation
relations identify the squared norm in \eqref{tf:product-bound}
with the expectation in \eqref{tf:joint-bound}, proving the latter.
\end{proof}

\subsection{Krylov layers and decay with depth}
\label{tf:krylov}

Let $\K_{s,\ell}$, $\W_{s,\ell}$, and $\Z_s$ be the spaces
defined in \eqref{kd:krylov}, \eqref{kd:layers}, and
\eqref{kd:decoupled} for the data in \eqref{tf:blocks}.
These spaces depend on $\beta$ through $T$.
Lemma~\ref{kd:geometry} gives
\begin{equation}\label{tf:layer-size}
 \dim\W_{s,\ell}\leq\rank T_s\leq\dim\A_s\leq m
\end{equation}
and the orthogonal decomposition
\[
 \B=\bigoplus_{s=0}^{S-1}
 \left(\bigoplus_{\ell=0}^{n_B-1}\W_{s,\ell}\oplus\Z_s\right).
\]
Every vector in $\W_{s,\ell}$ has Krylov depth $\ell$.

\begin{lemma}[Thermal decay with Krylov depth]\label{tf:depth}
For every $s\in\{0,\ldots,S-1\}$, $0\leq\ell<n_B$, and unit vector
$v\in\W_{s,\ell}$,
\begin{equation}\label{tf:depth-bound}
 s_\beta(v)\leq\tfrac12\,3^{-\ell}.
\end{equation}
On the low-energy shell, the stronger estimate
\begin{equation}\label{tf:low-depth}
 s_\beta(v)\leq
 \frac{(\beta a_0/4)^{\ell+1}}{(\ell+1)!}
 \leq\frac{4^{-(\ell+1)}}{(\ell+1)!}
\end{equation}
holds. For every $v\in\Z_s$, one has $s_\beta(v)=0$.
\end{lemma}
\begin{proof}
For $s\geq1$, apply \Cref{cor:kd-shell-depth} to
$E_s,T_s$, using the bounds in Lemma~\ref{tf:thermal-block-lemma}.
For $s=0$, apply \eqref{kd:low-shell-bound} with upper energy bound
$a_0$. Its hypotheses follow from \eqref{tf:low-block-bounds} and
$\beta a_0\leq1$. The assertion for $\Z_s$ follows from
Lemma~\ref{kd:geometry}.
\end{proof}

For a unit vector $v\in\Z_s$, Lemma~\ref{tf:pull} implies that
both $\mathsf d_-(v)$ and $\mathsf d_+(v)$ annihilate the TFD exactly.
Thus these are empty modes in the doubled thermal basis, even though
the corresponding physical free modes may have nonzero thermal occupations.

\subsection{Weighted truncation and the Gaussian superposition}
\label{tf:truncation}

Choose an orthonormal basis in each nonzero layer
$\W_{s,\ell}$ and each
$\Z_s$, with $s\in\{0,\ldots,S-1\}$ and $0\leq\ell<n_B$.
Make these choices, including the phases and ordering of the vectors,
by fixed deterministic one-particle linear algebra using only
$H_0,\Aimp,\beta$. For each such vector $v$, include both doubled
annihilators $\mathsf d_-(v)$ and $\mathsf d_+(v)$.
Complete these with the $2n_A$ physical and ancilla boundary annihilators,
and fix an ordering of the resulting $2n$ modes, denoted
$\mathsf f_1,\ldots,\mathsf f_{2n}$, with boundary modes first.
Their common normalized vacuum is denoted $\Omega_\beta$.
It is the boundary vacuum tensored, in the fermionic ordering,
with the free residual-bath TFD. Fix its phase by
$\langle\Omega_{2n}|\Omega_\beta\rangle>0$; this overlap is nonzero
at every finite $\beta$.

Define $U_\beta$ by its action on every standard doubled Fock vector:
\begin{equation}\label{tf:fock-basis}
 U_\beta\ket{x}=\ket{x}_\beta
 :=\prod_{j=1}^{2n}(\mathsf f_j^\dagger)^{x_j}
                  \ket{\Omega_\beta}
 \quad(x\in\{0,1\}^{2n}).
\end{equation}
The product has increasing $j$ from left to right. The canonical anticommutation rules imply that
these vectors form an orthonormal basis, so $U_\beta$ is unitary.
Conjugating each standard doubled annihilator by $U_\beta$ gives the
corresponding $\mathsf f_j$, a linear combination of the physical and
ancilla creation and annihilation operators. Hence $U_\beta$ is a
Bogoliubov transformation and every $U_\beta\ket{x}$ is a pure
fermionic Gaussian state. This construction uses neither $\delta$
nor $V$.

Let $\mathcal J_\beta\subseteq[2n]$ label the doubled layer modes
and $\mathcal D_\beta\subseteq[2n]$ the doubled decoupled modes.
For $j\in\mathcal J_\beta$, let $\ell(j)$ be its Krylov depth and
put $\mathsf n_j=\mathsf f_j^\dagger\mathsf f_j$.
There are at most $2mS$ doubled layer modes at any fixed depth
and $2n_A\leq2mS$ boundary modes.
For every subset $I\subseteq\mathcal J_\beta$, Lemmas~\ref{tf:pull}
and \ref{tf:depth} give
\begin{equation}\label{tf:depth-product}
 \bra{\mathsf{TFD}_\beta(H)}\prod_{j\in I}\mathsf n_j
                       \ket{\mathsf{TFD}_\beta(H)}
 \leq\prod_{j\in I}\left(\tfrac14\,9^{-\ell(j)}\right).
\end{equation}

Define the weighted particle number and its value on a configuration by
\begin{equation}\label{tf:weight}
 Q_\beta=\sum_{j\in\mathcal J_\beta}(\ell(j)+1)\mathsf n_j,
 \qquad Q_\beta(x)=\sum_{j\in\mathcal J_\beta}(\ell(j)+1)x_j.
\end{equation}
For an integer $R\geq0$, let
\begin{align}
 \mathcal Y_{\beta,R}
 &=\{x\in\{0,1\}^{2n}: Q_\beta(x)\leq R,
             \ x_j=0\text{ for every }j\in\mathcal D_\beta\},\notag\\
 P_{\beta,R}&=\sum_{x\in\mathcal Y_{\beta,R}}
                      \ket{x}_\beta\,{}_{\beta}\!\bra{x},\notag\\
 p_{\beta,R}(H)&=\norm{(\Id-P_{\beta,R})
                              \ket{\mathsf{TFD}_\beta(H)}}^2,
 \qquad D_{\beta,R}=|\mathcal Y_{\beta,R}|.
 \label{tf:cutoff-def}
\end{align}
Boundary occupations are unrestricted.

\begin{lemma}[Thermal tail and dimension bounds]\label{tf:tail}
For every integer $R\geq0$,
\begin{equation}\label{tf:tail-dimension}
 p_{\beta,R}(H)\leq e^{4mS}4^{-R},\qquad
 D_{\beta,R}\leq e^{(2\log2+2)mS}2^R.
\end{equation}
\end{lemma}
\begin{proof}
Expanding the product of commuting number projections and applying
\eqref{tf:depth-product}, for $z_j\geq1$ with
$j\in\mathcal J_\beta$ we obtain
\begin{align}
 &\bra{\mathsf{TFD}_\beta(H)}
       \prod_{j\in\mathcal J_\beta}z_j^{\mathsf n_j}
                   \ket{\mathsf{TFD}_\beta(H)}\notag\\
 &\quad\leq\prod_{j\in\mathcal J_\beta}
       \left[1+(z_j-1)\tfrac14\,9^{-\ell(j)}\right].
 \label{tf:generating-function}
\end{align}
Set $z_j=4^{\ell(j)+1}$, use $\log(1+x)\leq x$ for $x\geq0$,
and sum over the at most $2mS$ modes at each depth:
\begin{equation}\label{tf:four-moment}
 \log\bra{\mathsf{TFD}_\beta(H)}4^{Q_\beta}
                          \ket{\mathsf{TFD}_\beta(H)}
 \leq\frac{2mS}{4}\sum_{\ell=0}^\infty
                   (4^{\ell+1}-1)9^{-\ell}
 =\frac{243mS}{80}\leq4mS.
\end{equation}
All decoupled doubled modes are empty, so $p_{\beta,R}(H)$ is the
probability that $Q_\beta>R$. Markov's inequality proves the first
bound in \eqref{tf:tail-dimension}.

For every retained configuration, $1\leq2^{R-Q_\beta(x)}$.
Summing and then allowing arbitrary layer occupations gives
\begin{align*}
 D_{\beta,R}
 &\leq 2^{2n_A}2^R\prod_{j\in\mathcal J_\beta}
                   (1+2^{-(\ell(j)+1)})\\
 &\leq 2^{2n_A}2^R\exp\left(2mS\sum_{\ell=0}^\infty
     2^{-(\ell+1)}\right)
 \leq e^{(2\log2+2)mS}2^R.
\end{align*}
This proves the second bound. The infinite series only bound sums
over finitely many doubled modes.
\end{proof}

\subsection{Proof of Theorem~\ref{thm:rank}}

The case $\beta=0$ is covered by \eqref{tf:beta-zero}.
For $\beta>0$ and $0<\delta<1$, define
\begin{equation}\label{tf:cutoff-choice}
 R_{\beta,\delta}=4mS+\left\lceil\log_2(2/\delta)\right\rceil,
 \qquad
 \mathcal X_{\beta,\delta}=\mathcal Y_{\beta,R_{\beta,\delta}}.
\end{equation}
Thus $\mathcal Y_{\beta,R}$ is the auxiliary family indexed by an
integer cutoff, and $\mathcal X_{\beta,\delta}$ is the set in the
theorem. The unitary $U_\beta$ is defined in \eqref{tf:fock-basis}.
Let $\Pi_{\beta,\delta}$ project onto its retained vectors:
\begin{equation}\label{tf:theorem-projector}
 \Pi_{\beta,\delta}
 =\sum_{x\in\mathcal X_{\beta,\delta}}
       U_\beta\ket{x}\bra{x}U_\beta^\dagger
 =P_{\beta,R_{\beta,\delta}}.
\end{equation}
Write $R=R_{\beta,\delta}$ in the remaining estimates.
Lemma~\ref{tf:tail} gives, for every allowed interaction $V$ and
$H=H_0+V$,
\begin{equation}\label{tf:chosen-tail}
 p_{\beta,R}(H)
 \leq\frac{\delta^2}{4}\exp[(4-8\log2)mS]
 \leq\frac{\delta^2}{4}<1.
\end{equation}
Consequently,
\begin{equation}\label{tf:concentration}
 \sum_{x\in\mathcal X_{\beta,\delta}}
   |\langle\mathsf{TFD}_\beta(H)|U_\beta|x\rangle|^2
 =1-p_{\beta,R}(H)\geq1-\delta^2,
\end{equation}
which is precisely \eqref{eq:TFD_approx}. For the cardinality,
$2^R\leq(4/\delta)2^{4mS}$ and Lemma~\ref{tf:tail} imply
\begin{equation}\label{tf:rank-bound}
 |\mathcal X_{\beta,\delta}|=D_{\beta,R}
 \leq\frac4\delta e^{(6\log2+2)mS}
 \leq\frac4\delta e^{7mS}
 \leq\frac1\delta(2+\beta)^{O(m)}.
\end{equation}
The last inequality uses \eqref{tf:shell-count} and $m\geq1$.

We now verify the two separate construction costs. The partition,
boundary spaces, template maps $E_s,T_s$, and all
Krylov layers depend only on $H_0,\Aimp,\beta$. They are obtained
from matrices of dimension at most $n$. The Krylov spaces stabilize
after at most $n_B$ steps, so the coefficients of the complete
Bogoliubov transformation $U_\beta$ are computed with polynomially
many one-particle operations, independently of $\delta$ and $V$.
The phase convention in \eqref{tf:fock-basis} is also fixed by these
data. In the arithmetic model, scalar exponentials are evaluated
to the accuracy required by the output tolerance. This gives the stated
$\poly(n,1+\beta)$ cost for $U_\beta$.

Only layer depths less than $R$ can be occupied in retained
configurations. Including the unrestricted boundary, at most
$2mS(R+1)$ doubled modes can be occupied. Applying
Lemma~\ref{lem:enum} to these modes, with all other occupations fixed
to zero, gives enumeration cost polynomial in $n,mS,R,D_{\beta,R}$. Since
$mS=O(m\log(2+\beta))$ and $R=O(mS+\log(1/\delta))$, the cost is
$\poly(n,1+\beta,\delta^{-1})$ for fixed $m$.
This proves the required independence
and the theorem.

\section{Constructive approximation of thermofield-double state}
\label{ct:construction}

Here we prove Theorem~\ref{ct:main}. 
This section depends on the ground-energy algorithm of
Theorem~\ref{thm:ground-energy-algorithm} which we use as a subroutine here
and the thermal multiparticle
pull-through, Lemma~\ref{tf:pull}. We extend the latter to Schatten
norms and use an artificial bath gap to obtain one fixed Krylov decomposition
for an  imaginary-time trajectory $\tau\in [0,\beta]$.
 A linear least-squares fit
  then determines a discretized version of this trajectory within the span of retained Fock configurations.
Throughout this section we
 assume $\norm{V}\leq1$.

\subsection{Krylov decomposition for the least-squares fit}
\label{ct:fixed-basis-section}

Let $e_1,\ldots,e_n \in [0,1]$ be the bath excitation energies.
First, suppose $\omega\leq e_j\leq1$ for all $j$,
where $0<\omega\leq1$ is the bath gap.
The gapped bath assumption is relaxed in  Section~\ref{ct:ungapped-section}.

We shall use the zero-temperature Krylov decomposition of Section~\ref{ge:shelled-decomposition}
which we briefly summarize below. 
The bath spectrum is partitioned into $S$ dyadic energy shells  with
\[
S=1+\lfloor\log_2(1/\omega)\rfloor
\]
and shell indices
$s=0,\ldots,S-1$, see Section~\ref{ge:shells}.  Empty shells are allowed.
Let $\Pi_s$ be the corresponding spectral projector of $h$ and
$\mathcal H_s=\ran\Pi_s$. The boundary and the residual  bath spaces are defined as
\[
 \A_s=\Pi_s\Aimp,\qquad
 \B_s=\mathcal H_s\ominus\A_s,\qquad
 \A=\bigoplus_{s=0}^{S-1}\A_s,\qquad
 \B=\bigoplus_{s=0}^{S-1}\B_s.
\]
Thus  $n_A=\dim\A\leq mS$ and $n_B=\dim\B=n-n_A$.
Since $\Aimp\subseteq\A$, $V$ is supported on $\A$.
Relative to $\A\oplus\B$, write
\[
 h=\begin{pmatrix}
 h_{\A\A}&h_{\A\B}\\
 h_{\A\B}^\dagger&h_{\B\B}
 \end{pmatrix}.
\]
The linear maps $E\, : \, \B\to \B$ and $T\, : \, \B\to \A$ are defined as
\begin{equation}\label{ct:blocks}
 E=h_{\B\B},\qquad T=h_{\A\B},\qquad
 P_s=\Pi_s|_{\B},\qquad E_s=E|_{\B_s},\quad T_s=T|_{\B_s}.
\end{equation}
Consider a vector $v\in \B$. We shall measure how strongly $v$ is coupled to the boundary using 
a modified 
integrated coupling defined as
\begin{equation}\label{ct:integrated-coupling}
\hat{s}_\tau(v)=\int_0^{\tau/2}\norm{Te^{-uE}(\Id+e^{-\tau E})^{-1/2}v}\,du,
\end{equation}
where $\tau \in [0,\beta]$ is the imaginary time used in the least-squares fit interpolation.
Note that $\hat{s}_\tau(v)$ is similar to the integrated coupling used in the finite-temperature Krylov decomposition of Section~\ref{tf:shells}
with $\beta=\tau$,
except that the Fermi-Dirac filter $(\Id+e^{-\tau E})^{-1/2}$ is moved from the linear map $T$ to 
the integrated coupling, see \eqref{tf:blocks}. This is more convenient for our purposes since 
we do not need to rebuild the Krylov basis as $\tau$ varies.

By definition, 
each $\B_s$ is $E$-invariant, $\ran T_s\subseteq\A_s$, and
$\rank T_s\leq m$. Lemma~\ref{ge:block-lemma} gives
\[
 2^s\omega\Id\leq E_s\leq\min\{2^{s+1}\omega,1\}\Id,
 \qquad \norm{T_s}\leq 2^{s-1}\omega.
\]
Let $\K_{s,r}$, $\W_{s,r}$, and $\Z_s$ be the spaces defined in
\eqref{kd:krylov}, \eqref{kd:layers}, and \eqref{kd:decoupled}
for the considered Krylov decomposition.
 Lemma~\ref{kd:geometry} gives
$\dim\W_{s,r}\leq m$, and every vector in $\W_{s,r}$ has
Krylov depth\footnote{We reserve the notation $\ell$ for Krylov depth of certain groups of shells defined below.} $r$. The space $\Z_s$ is $E_s$-invariant and
$T_s\Z_s=0$. Fix an orthonormal basis in every nonzero layer and
in every $\Z_s$. Together with a fixed boundary basis, these
give an orthonormal  basis of $\C^n$ used  throughout the construction.
Let $v_1,\ldots,v_{n_B}$ be the corresponding 
orthonormal basis of $\B$.
This basis is in general not an eigenbasis of $E$; the free TFD annihilators $\mathsf d_{\eta,\tau}(v)$ used below are nevertheless well defined for every $v\in\B$ by \eqref{tf:annihilators-def}.

Write the residual-bath Hamiltonian as
\[
 H_B=\sum_{j,k=1}^{n_B}E_{jk}\beta(v_j)^\dagger\beta(v_k).
\]
Let $Z_\tau(H_B)$ be the partition function of $H_B$ considered as a Hamiltonian 
acting on the Fock space  $\calF_{n_B}$ ignoring the boundary subsystem.
If $\lambda_1,\ldots,\lambda_{n_B} \in [0,1]$ are the eigenvalues of $E=h_{\B\B}$ then 
\[
 Z_\tau(H_B)
         =\prod_{j=1}^{n_B}(1+e^{-\tau\lambda_j}).
\]
One can compute $Z_\tau(H_B)$ in time $\poly(n)$ by diagonalizing $E$.

\begin{lemma}[Exponential decay]\label{ct:lem-fixed}
For every unit vector $v\in\W_{s,r}$ and every $\tau\geq0$,
\begin{equation}\label{ct:fixed-decay}
 \hat{s}_\tau(v)
 \leq\min\{1/2,C_0\,2^{-r}\},
 \qquad C_0=\frac{2}{\sqrt{1-e^{-1/2}}}.
\end{equation}
\end{lemma}
\begin{proof}
Use the lower shell endpoint $e_s^\downarrow=2^s\omega$.
For $u,\tau\geq0$, the function
\[
 z\longmapsto e^{-uz}(1+e^{-\tau z})^{-1/2}
\]
is analytic on the disk
$|z-3e_s^\downarrow/2|\leq e_s^\downarrow$ and satisfies
\[
 \bigl|e^{-uz}(1+e^{-\tau z})^{-1/2}\bigr|\leq
 \frac{e^{-u e_s^\downarrow/2}}{\sqrt{1-e^{-1/2}}}.
\]
Indeed, if $\tau e_s^\downarrow\leq1$, then
$|\tau\operatorname{Im}z|\leq1$ and
$\operatorname{Re}(e^{-\tau z})\geq0$. If
$\tau e_s^\downarrow\geq1$, then $|e^{-\tau z}|\leq e^{-1/2}$.
In either case $\operatorname{Re}(1+e^{-\tau z})\geq1-e^{-1/2}>0$,
so the principal inverse square root is analytic and obeys the
stated bound.

The depth condition \eqref{kd:depth} makes the first $r$ terms
of the Taylor expansion centered at $3e_s^\downarrow/2$ vanish 
 after applying $T_s$ to
$e^{-uE_s}(\Id+e^{-\tau E_s})^{-1/2}v$. Cauchy's coefficient bound and
$\norm{E_s-3e_s^\downarrow\Id/2}\leq e_s^\downarrow/2$ give
\[
 \norm{T_s e^{-uE_s}(\Id+e^{-\tau E_s})^{-1/2}v}
 \leq\frac{e_s^\downarrow e^{-u e_s^\downarrow/2}}
                {\sqrt{1-e^{-1/2}}}\,2^{-r}.
\]
Integration over $u\in[0,\tau/2]$ gives $C_0\,2^{-r}$.
The separate bound $1/2$ follows from
$\norm{T_s}\leq e_s^\downarrow/2$,
$E_s\geq e_s^\downarrow\Id$, and
$\norm{(\Id+e^{-\tau E_s})^{-1/2}}\leq1$.
\end{proof}

The exponential decay offered by Lemma~\ref{ct:lem-fixed} might
not be strong enough on its own to compensate for the amplification
of errors in the least-squares approximation. However, grouping
$O(1)$ successive Krylov layers into a single super-layer gives
sufficiently fast decay in the super-layer index.
More formally, for each energy shell $s$ and integer $\ell\ge0$,
define the Krylov super-layer
\[
 \W^{\rm super}_{s,\ell}
 =\bigoplus_{r=16\ell}^{16\ell+15}\W_{s,r}.
\]
Use the union of the already chosen orthonormal bases of these
layers as a basis of $\W^{\rm super}_{s,\ell}$.
Assign weight $\ell+1$ to each mode in super-layer $\W^{\rm super}_{s,\ell}$.
For each shell, there are at most $16m$ bath modes per super-layer.

For each $v\in\B$, imaginary time $\tau\in[0,\beta]$, and sign
$\eta\in\{-,+\}$, let $\mathsf d_{\eta,\tau}(v)$ be the free TFD
annihilator defined in Subsection~\ref{tf:annihilators}, with
inverse temperature $\tau$.
These operators annihilate $\ket{\mathsf{TFD}_\tau(H_B)}$,
the thermofield-double state of the bath Hamiltonian $H_B$.
Let
\[
 \mathcal J
 =\{(v_i,\eta):i=1,\ldots,n_B,\ \eta\in\{-,+\}\},
\]
where $v_1,\ldots,v_{n_B}\in \B$ is the chosen orthonormal basis of $\B$.
The $2n_B$ operators $\mathsf d_{\eta,\tau}(v)$ with
$(v,\eta)\in\mathcal J$ obey the canonical anticommutation relations.

Let $u_1,\ldots,u_{n_A}\in \A$ be some fixed orthonormal basis of $\A$.
For each vector $u\in \{u_1,\ldots,u_{n_A}\}$
the corresponding physical and ancilla boundary annihilators are
\begin{equation}
\label{eq:boundary_annihilators_confusing}
 \sum_{k=1}^n \overline{u}_k \beta_k
 \quad \mbox{and} \quad
 \sum_{k=1}^n u_k\widetilde\beta_k.
\end{equation}
Fix an ordering of $\mathcal J$ and of the $2n_A$ boundary
annihilators, independent of $\tau$.
Let $\mathcal U_\tau:\mathcal F_{2n}\to\mathcal F_{2n}$ be a
unitary Bogoliubov transformation such that conjugation
$a_j\mapsto\mathcal U_\tau a_j\mathcal U_\tau^\dagger$
maps the standard annihilators $a_1,\ldots,a_{2n_B}$ to the
ordered free TFD annihilators
$\mathsf d_{\eta,\tau}(v)$ with
$(v,\eta)\in\mathcal J$
and maps the remaining standard
annihilators to the ordered boundary annihilators \eqref{eq:boundary_annihilators_confusing}. 
These conjugation relations determine $\mathcal U_\tau$ up to
an overall phase. Fix this phase by requiring that
$\mathcal U_\tau$ maps the standard vacuum to
$\ket{\mathsf{TFD}_\tau(H_B)}$ with both copies of the boundary
in their vacuum. With this convention, $\mathcal U_\tau$
depends smoothly on $\tau$.

Fix the mode ordering in the definition of $\mathcal U_\tau$ so
that modes $2j-1$ and $2j$ correspond to $(v_j,-)$ and $(v_j,+)$,
respectively, for $j=1,\ldots,n_B$. The remaining $2n_A$ modes
are the  boundary modes. This determines the order of bits in Fock configurations.

For a Fock configuration $x\in\{0,1\}^{2n}$, define its weight as
\[
 Q(x)
 =
 \sum_{s=0}^{S-1}\; \sum_{\ell\ge0}
 (\ell+1)
 \sum_{\substack{1\le j\le n_B\\
                 v_j\in\W^{\mathrm{super}}_{s,\ell}}}
 \bigl(x_{2j-1}+x_{2j}\bigr).
\]
Thus each occupied bath bit associated with super-layer $\ell$
contributes weight $\ell+1$, while boundary bits contribute nothing.
Given an integer cutoff $R\ge 1$, define a set of retained configurations
\[
 \mathcal X_R
 =
 \left\{
 x\in\{0,1\}^{2n}:
 \begin{array}{l}
 Q(x)\le R,\\
 x_{2j-1}=x_{2j}=0
 \text{ whenever }v_j\in\displaystyle\bigoplus_s\Z_s
 \end{array}
 \right\}.
\]
The second condition requires all decoupled bath modes to be
empty. Boundary bits are unrestricted.

Define
\[
 P_R=\sum_{x\in\mathcal X_R}\ket{x}\bra{x}.
\]
These are the doubled-space analogues of the set $\mathcal X_R$, the projector $P_R$ and the weight $Q(x)$ of Section~\ref{sec:weightfoctruncation}, with Krylov depth replaced by the super-layer index.
The projector onto the span of
$\{\mathcal U_\tau\ket{x}:x\in\mathcal X_R\}$ is
\[
 \mathcal P_R(\tau)=\mathcal U_\tau P_R\mathcal U_\tau^\dagger.
\]

\subsection{Pull-through bounds in Schatten norms and normalization}\label{ct:schatten-section}

The definitions of $\hat{s}_\tau(v)$, $\mathsf d_{\eta,\tau}(v)$,
$\mathcal U_\tau$, and $\mathcal P_R(\tau)$ extend to every
$\tau\ge0$. We use this extension in the partition function bound
below; the least-squares fit still uses $\tau\in[0,\beta]$.

We use the vectorization \eqref{tf:isometry}, with
\[
 \ket{\mathsf{TFD}_\tau(H)}
 =\frac{\ket{e^{-\tau H/2}}_f}{\sqrt{Z_\tau(H)}}.
\]
Schatten norms are taken on the original Fock space $\calF_n$.
The notation $\norm{\cdot}_{\mathrm{HS}}$ denotes the
Hilbert--Schmidt  norm (which coincides with the Frobenius norm), while doubled vectors have the Euclidean norm.
As in Section~\ref{tf:annihilators}, a doubled projector or
annihilator acting on an operator denotes its induced action
through vectorization.

For future reference, let us record a bound 
\begin{equation}\label{ct:eq-all-p}
 \norm{\mathsf d_{\eta_1,\tau}(v_{i_1})\cdots
       \mathsf d_{\eta_k,\tau}(v_{i_k})e^{-\tau H/2}}_p
 \le \norm{e^{-\tau H/2}}_p
       \prod_{j=1}^k\hat{s}_\tau(v_{i_j})
\end{equation}
which holds for every $p\in[1,\infty]$  and any pairs  $(v_{i_1},\eta_1),\ldots,(v_{i_k},\eta_k)\in\mathcal J$.
The above bound is a rephrasing of \eqref{tf:holder-bound-p} of 
Lemma~\ref{tf:pull}.
We next bound the number of retained Fock configurations and
the truncation error in Schatten norms.

\begin{lemma}[Cutoff estimates]\label{ct:lem-cutoff}
With $16$ Krylov layers per super-layer, there is an absolute
constant $c>0$ such that, simultaneously for all $\tau\ge0$,
integers $R\ge1$, and $p\in[1,\infty]$,
\begin{align}
 D_R:=\rank P_R&\le e^{cmS}2^R,\label{ct:eq-dimension}\\
 \norm{(I-\mathcal P_R(\tau))e^{-\tau H/2}}_{\mathrm{HS}}
      &\le e^{cmS}2^{-8R}\norm{e^{-\tau H/2}}_{\mathrm{HS}},\label{ct:eq-hs-tail}\\
 \norm{(I-\mathcal P_R(\tau))e^{-\tau H/2}}_p
      &\le e^{cmS}2^{-R}\norm{e^{-\tau H/2}}_p.\label{ct:eq-p-tail}
\end{align}
Every configuration  $x\in\mathcal X_R$ has at most $R$ occupied bath bits.
\end{lemma}
\begin{proof}
For each super-layer index $\ell$, there are at most $32mS$
free TFD modes, counting both signs and all shells.
Lemma~\ref{ct:lem-fixed} gives
\[
 \hat{s}_\tau(v)\le C_0\,2^{-16\ell}
\]
for every chosen basis vector $v$ in super-layer $\ell$.
We use these two facts throughout the proof.

\medskip
\noindent\emph{Counting retained configurations.}
For every $x\in\mathcal X_R$,
$1\le 2^{R-Q(x)}$.
Summing this inequality over retained configurations and then
allowing arbitrary occupations of the super-layer modes gives
\[
 D_R=|\mathcal X_R|
 \le
 2^{2n_A}2^R
 \prod_{\ell=0}^{\infty}
 \left(1+2^{-(\ell+1)}\right)^{32mS}
 \le e^{cmS}2^R.
\]
Here $2^{2n_A}$ counts the unrestricted boundary bits.
The last inequality follows from $n_A\le mS$,
$\log(1+u)\le u$ for $u\ge0$, and
$\sum_{\ell\ge0}2^{-(\ell+1)}=1$.
This proves \eqref{ct:eq-dimension}.

\medskip
\noindent\emph{Hilbert--Schmidt truncation error.}
Consider measuring $\ket{\mathsf{TFD}_\tau(H)}$ in the
orthonormal basis $\{\mathcal U_\tau\ket{x}\}$.
The probability of outcome $x$ is
\[
 \left|\bra{x}\mathcal U_\tau^\dagger
                  \ket{\mathsf{TFD}_\tau(H)}\right|^2.
\]
All probabilities and expectations below refer to this
measurement.

Let $I$ be any subset of $\mathcal J$.
Each element of $I$ is a pair $(v,\eta)$ specifying one free
TFD mode. The canonical anticommutation relations and
\eqref{ct:eq-all-p} with $p=2$ give
\begin{align}
 &\Pr[\text{every mode in }I\text{ is occupied}]
 \notag\\
 &\qquad=
 \frac{1}{Z_\tau(H)}
 \left\|
   \left[\prod_{(v,\eta)\in I}\mathsf d_{\eta,\tau}(v)\right]
   e^{-\tau H/2}
 \right\|_{\mathrm{HS}}^2
 \le
 \prod_{(v,\eta)\in I}\hat{s}_\tau(v)^2.
 \label{ct:eq-joint-occupation}
\end{align}
The annihilators may be taken in any fixed order.
If $v\in\bigoplus_s\Z_s$, then $\hat{s}_\tau(v)=0$,
so both modes associated with $v$ are exactly empty.
Consequently, an outcome lies outside $\mathcal X_R$
precisely when its weight exceeds $R$, except on an event
of probability zero.

For $z>1$, each bath bit $x_i$ belonging to super-layer $\ell$
satisfies
\[
 z^{(\ell+1)x_i}
 =1+(z^{\ell+1}-1)x_i.
\]
Multiply this identity over all bath bits with positive weight.
The left-hand side becomes $z^{Q(x)}$.
Expanding the right-hand side produces products of occupied
bits, whose expectations are the joint occupation
probabilities bounded in \eqref{ct:eq-joint-occupation}.
All expansion coefficients are nonnegative.
Using $\hat{s}_\tau(v)^2\le C_0^2\,2^{-32\ell}$ and the bound
of $32mS$ modes at each super-layer index therefore gives
\[
 \mathbb E\!\left[z^{Q(x)}\right]
 \le
 \prod_{\ell=0}^{\infty}
 \left[1+(z^{\ell+1}-1)C_0^2\,2^{-32\ell}\right]^{32mS}.
\]
Set $z=2^{16}$. The logarithm of the right-hand side is at most
\[
 32mS\,C_0^2
 \sum_{\ell=0}^{\infty}
 (2^{16(\ell+1)}-1)2^{-32\ell}
 \le
 32mS\,C_0^2\,2^{16}
 \sum_{\ell=0}^{\infty}2^{-16\ell}
 \le cmS.
\]
Thus Markov's inequality gives
\[
 \Pr[Q(x)>R]
 \le 2^{-16R}
       \mathbb E\!\left[2^{16Q(x)}\right]
 \le e^{cmS}2^{-16R}.
\]
By the isometry \eqref{tf:isometry},
\[
 \frac{\norm{(I-\mathcal P_R(\tau))e^{-\tau H/2}}_{\mathrm{HS}}^2}
      {Z_\tau(H)}
 =
 \Pr[x\notin\mathcal X_R]
 =
 \Pr[Q(x)>R].
\]
Taking square roots proves \eqref{ct:eq-hs-tail}, after
enlarging the absolute constant $c$ if needed.

\medskip
\noindent\emph{Truncation error in Schatten norms.}
We now use \eqref{ct:eq-all-p} for general $p$.
Expand the function
$\mathbf 1_{\{Q(x)>R\}}$
as a multilinear polynomial in the bath bits.
Using the fixed correspondence between bits and mode pairs,
write $c_I$ for the coefficient of the monomial formed by
the bits indexed by $I\subseteq\mathcal J$.

By inclusion--exclusion, $c_I$ is an alternating sum of
$2^{|I|}$ values of this indicator, evaluated on configurations
whose occupied bath modes form a subset of $I$.
Each value is either zero or one, so
\[
 |c_I|\le 2^{|I|}.
\]
If the total weight of the bits indexed by $I$ is at most $R$,
all these indicator values vanish, and hence $c_I=0$.
Also, coefficients involving a decoupled mode vanish because
the indicator does not depend on its bit.

Replace each bath bit by its corresponding number projection
$\mathsf d_{\eta,\tau}(v)^\dagger\mathsf d_{\eta,\tau}(v)$.
Since the decoupled modes are exactly empty, this gives
\begin{align*}
 (I-\mathcal P_R(\tau))e^{-\tau H/2}=
 \sum_{I\subseteq\mathcal J}c_I
 \left[
   \prod_{(v,\eta)\in I}
   \mathsf d_{\eta,\tau}(v)^\dagger
   \mathsf d_{\eta,\tau}(v)
 \right]e^{-\tau H/2}.
\end{align*}

Each free TFD creation operator, acting through vectorization,
has induced Schatten-$p$ norm at most $2$.
Indeed, the formulas in Section~\ref{tf:annihilators} express
its action as a sum of two terms involving left or right
multiplication by operators of norm at most $1$.
Conjugation by fermion parity also preserves Schatten norms.

The canonical anticommutation relations allow us to move all
creators in each monomial to the left of all annihilators.
Applying the creator bound and then \eqref{ct:eq-all-p} gives
\[
 \left\|
 \left[
   \prod_{(v,\eta)\in I}
   \mathsf d_{\eta,\tau}(v)^\dagger
   \mathsf d_{\eta,\tau}(v)
 \right]e^{-\tau H/2}
 \right\|_p
 \le
 2^{|I|}\norm{e^{-\tau H/2}}_p
 \prod_{(v,\eta)\in I}\hat{s}_\tau(v).
\]
Together with $|c_I|\le2^{|I|}$, this implies
\[
 \frac{\norm{(I-\mathcal P_R(\tau))e^{-\tau H/2}}_p}
      {\norm{e^{-\tau H/2}}_p}
 \le
 \sum_{\substack{I\subseteq\mathcal J\\c_I\ne0}}
 \prod_{(v,\eta)\in I}4\hat{s}_\tau(v).
\]

Every contributing subset has total weight greater than $R$.
We may therefore multiply its bound by $2^{-R}$ and by
$2^{\ell+1}$ for each of its modes in super-layer $\ell$:
the combined multiplier is at least one.
After this step, a selected mode in super-layer $\ell$
contributes at most
\[
 4\,2^{\ell+1}\hat{s}_\tau(v)
 \le 8C_0\,2^{-15\ell}.
\]
Allowing all subsets of the super-layer modes now makes the
sum factorize, giving
\[
 \frac{\norm{(I-\mathcal P_R(\tau))e^{-\tau H/2}}_p}
      {\norm{e^{-\tau H/2}}_p}
 \le
 2^{-R}
 \prod_{\ell=0}^{\infty}
 \left(1+8C_0\,2^{-15\ell}\right)^{32mS}
 \le e^{cmS}2^{-R}.
\]
The last inequality follows from
$\log(1+u)\le u$ and the convergence of
$\sum_{\ell\ge0}2^{-15\ell}$.
This proves \eqref{ct:eq-p-tail}.
\end{proof}

For $x\in\{0,1\}^{2n}$, define the operator $G_x(\tau)$ on
$\calF_n$ by
\[
 \ket{G_x(\tau)}_f=\mathcal U_\tau\ket{x},
 \qquad
 k(x)=\sum_{j=1}^{n_B}(x_{2j-1}+x_{2j}).
\]
Thus $k(x)$ counts occupied bath bits without their weights.
With all bath bits zero, the corresponding operator is, up to
a phase, a boundary matrix unit tensored with
$e^{-\tau H_B/2}/\sqrt{Z_\tau(H_B)}$.
Boundary creators have induced Schatten norm at most $1$, and
free TFD bath creators have induced norm at most $2$. Consequently,
\begin{equation}\label{ct:eq-basis-p}
 \norm{G_x(\tau)}_p\le
 2^{k(x)}\frac{Z_{p\tau/2}(H_B)^{1/p}}{\sqrt{Z_\tau(H_B)}}.
\end{equation}
At $p=\infty$ the numerator is interpreted as $1$.

\begin{lemma}[Partition function normalization]\label{ct:lem-partition}
For any $m$ and $S$ 
there exists a constant $C_{m,S}=e^{O(mS)}$ 
such that for every $\tau\ge0$,
\begin{equation}\label{ct:eq-partition}
 C_{m,S}^{-2}\le
       \frac{e^{\tau E_1(H)}Z_\tau(H)}{Z_\tau(H_B)}
       \le C_{m,S}^2.
\end{equation}
The constant $C_{m,S}$ is independent of $n,\tau$, and $\norm V$,
and may be replaced by a known upper bound of the same form.
\end{lemma}
\begin{proof}
Choose $R_0=O(mS)$ so that \eqref{ct:eq-p-tail} is at most $1/2$.
Expand $\mathcal P_{R_0}(\tau)e^{-\tau H/2}$ in the
Hilbert--Schmidt orthonormal basis
$\{G_x(\tau):x\in\mathcal X_{R_0}\}$.
Cauchy--Schwarz on the coefficients and \eqref{ct:eq-basis-p} give
\begin{equation}\label{ct:eq-norm-comparison}
 \norm{e^{-\tau H/2}}_p\le
 C_{m,S}\sqrt{Z_\tau(H)}\frac{Z_{p\tau/2}(H_B)^{1/p}}{\sqrt{Z_\tau(H_B)}},
 \qquad C_{m,S}=e^{O(mS)}.
\end{equation}
The case $p=\infty$ proves the lower bound in
\eqref{ct:eq-partition}.
The case $p=1$ gives
\[
 \frac{Z_{\tau/2}(H)}{Z_{\tau/2}(H_B)}
 \leq C_{m,S}\sqrt{\frac{Z_\tau(H)}{Z_\tau(H_B)}}.
\]
Iterating with successively doubled times yields, for every integer
$a\geq1$,
\[
 \log\frac{Z_\tau(H)}{Z_\tau(H_B)}
 \leq(2-2^{1-a})\log C_{m,S}
       +2^{-a}\log\frac{Z_{2^a\tau}(H)}{Z_{2^a\tau}(H_B)}.
\]
Since $H_B$ has ground energy zero,
\[
 \lim_{a\to\infty}2^{-a}
       \log\frac{Z_{2^a\tau}(H)}{Z_{2^a\tau}(H_B)}=-\tau E_1(H).
\]
This proves the upper bound. The same inequality holds at $\tau=0$ by
continuity. Uniformity of the fixed-basis estimates for arbitrarily
large times is used in this step.
\end{proof}

Use Theorem~\ref{thm:ground-energy-algorithm} to compute a number $E_*$ satisfying
\[
 E_1(H)\le E_*\le E_1(H)+\frac{1}{2(1+\beta)}.
\]
Define the unnormalized but polynomially bounded doubled vector
\begin{equation}\label{ct:eq-f}
\ket{f(\tau)}_f=\frac{e^{\tau E_*/2}\ket{e^{-\tau H/2}}_f}{\sqrt{Z_\tau(H_B)}},\qquad 0\le \tau\le\beta.
\end{equation}
Lemma~\ref{ct:lem-partition} implies
\begin{equation}\label{ct:eq-f-norm}
 C_{m,S}^{-1}\le\norm{\ket{f(\tau)}_f}\le e^{1/4}C_{m,S}.
\end{equation}
This normalization is fully computable without knowing $Z_\tau(H)$.
It is essential: a bound on unnormalized evolution alone would retain
an extensive bath partition-function factor.

\subsection{The residual stability lemma}\label{ct:stability-section}

Let $\widehat{\mathcal X}_R$ consist of all bit strings obtained
from a string in $\mathcal X_R$ by flipping at most two bits.
Define the diagonal projector
\[
 \widehat P_R=\sum_{x\in\widehat{\mathcal X}_R}\ket{x}\bra{x}.
\]
Then
\begin{equation}\label{ct:eq-enlargement}
 \widehat D_R:=\rank\widehat P_R\le O((n+1)^2D_R),
 \qquad k(x)\le R+2\quad(x\in\widehat{\mathcal X}_R).
\end{equation}
Here $k(x)$ is the number of occupied bath bits defined above.
The enlarged configurations need not obey the original weighted cutoff.
Only their number and their ordinary bath particle count are used.

For $0\le\sigma\le\tau\le\beta$, define the propagator
$\Phi(\tau,\sigma)$ on doubled vectors by
\[
 \Phi(\tau,\sigma)\ket{Y}_f
 =\sqrt{\frac{Z_\sigma(H_B)}{Z_\tau(H_B)}}\,
   \ket{e^{-(\tau-\sigma)(H-E_*I)/2}Y}_f
\]
for any operator $Y:\calF_n\to\calF_n$. It satisfies
$\ket{f(\tau)}_f=\Phi(\tau,\sigma)\ket{f(\sigma)}_f$.

\begin{lemma}[Exact propagation on the residual span]\label{ct:lem-propagation}
For every operator $Y:\calF_n\to\calF_n$ with
$\ket{Y}_f\in\mathcal U_\sigma\ran\widehat P_R$ and
$0\le\sigma\le\tau\le\beta$,
\begin{equation}\label{ct:eq-stability}
 \norm{\Phi(\tau,\sigma)\ket{Y}_f}
 \le A_R\norm{\ket{Y}_f},
 \qquad A_R\le e^{O(mS)}(n+1)2^{3R/2}.
\end{equation}
\end{lemma}
\begin{proof}
For $0<\sigma<\tau$, apply H\"older's inequality with exponents
$2\tau/(\tau-\sigma)$ and $2\tau/\sigma$ to a single basis
operator $G_x(\sigma)$, where $x\in\widehat{\mathcal X}_R$.
Using the isometry \eqref{tf:isometry},
Eq.~\eqref{ct:eq-basis-p} gives
\begin{align*}
 \norm{\Phi(\tau,\sigma)\ket{G_x(\sigma)}_f}
 &\le
 2^{k(x)}
 \left(\frac{e^{\tau E_*}Z_\tau(H)}{Z_\tau(H_B)}\right)^{(\tau-\sigma)/(2\tau)}\\
 &\le 2^{R+2}e^{1/4}C_{m,S}.
\end{align*}
At $\sigma=0<\tau$ use exponent $\infty$ for $G_x(0)$;
the same bound holds. At $\sigma=\tau$ use the identity map.
Expand $\ket{Y}_f$ in the orthonormal basis
$\{\mathcal U_\sigma\ket{x}:x\in\widehat{\mathcal X}_R\}$.
Cauchy--Schwarz proves the result using \eqref{ct:eq-dimension}
and \eqref{ct:eq-enlargement}.
\end{proof}

This bound holds for every vector in
$\mathcal U_\sigma\ran\widehat P_R$ and will be used to control
the propagation of the least-squares residuals.

\subsection{A computable linear least-squares problem}\label{ct:least-squares-section}

Set $z(\tau)=\mathcal U_\tau^\dagger\ket{f(\tau)}_f$.
We regard $z(\tau)$ as a vector in the standard Fock basis.
It solves
\begin{equation}\label{ct:eq-ode}
 \begin{aligned}
 z'(\tau)&=M(\tau)z(\tau),\\
 M(\tau)&=\mathcal U_\tau^\dagger
 \left[-\frac{H_L-E_*I}{2}
       -\frac12\left(\frac{d}{d\tau}\log Z_\tau(H_B)\right)I\right]
       \mathcal U_\tau
       -\mathcal U_\tau^\dagger\mathcal U'_\tau,
 \end{aligned}
\end{equation}
where $H_L$ is left multiplication by $H$, and
\[
 -\frac{d}{d\tau}\log Z_\tau(H_B)
       =\sum_{j=1}^{n_B}\frac{\lambda_j}{1+e^{\tau\lambda_j}}.
\]
The conjugated interaction $\mathcal U_\tau^\dagger V_L\mathcal U_\tau$
is independent of $\tau$ and is supported on the last $2n_A$
standard modes, since the chosen boundary annihilators are
independent of $\tau$. It preserves $\ran P_R$ because boundary
bits are unrestricted. All other terms in $M(\tau)$ are quadratic
or scalar, and a quadratic fermionic monomial changes at most
two bits. Thus
\begin{equation}\label{ct:eq-closure}
 M(\tau)\ran P_R\subseteq\ran\widehat P_R\qquad(0\le \tau\le\beta).
\end{equation}
This is where the small support of the interaction is used
algorithmically. In particular, arbitrarily complicated interactions
inside the unrestricted boundary do not enlarge the bath cutoff.

There is a known $B_0=e^{O(mS)}$ and a polynomial $B_1=\poly(n+1)$ such that
\begin{equation}\label{ct:eq-regularity}
 \norm{z(\tau)}\le B_0,\quad
 \norm{M(\tau)}\le B_1,\quad \norm{M'(\tau)}\le B_1^2,\quad
 \norm{z''(\tau)}\le2B_1^2B_0.
\end{equation}
Indeed $\norm H=O(n+1)$. In an eigenbasis of $E$, the
time-dependent bath rotations in $\mathcal U_\tau$ have angles
$\arctan(e^{-\tau\lambda_j/2})$. Their first two derivatives
are bounded by absolute constants, since $0\le\lambda_j\le1$.
The other basis changes are independent of $\tau$.
In particular, $\norm{\mathcal U_\tau^\dagger\mathcal U'_\tau}=O(n)$. The remaining bounds follow by
differentiating \eqref{ct:eq-ode}. The exact path obeys
\begin{equation}\label{ct:eq-z-tail}
 \norm{(I-P_R)z(\tau)}\le a_R,
 \qquad a_R=e^{O(mS)}2^{-8R}.
\end{equation}

Take $L=\lceil(1+\beta)2^{4R}\rceil$, $\Delta \tau=\beta/L$, and
$\tau_j=j\Delta \tau$ for $j\in\{0,\ldots,L\}$. For $\beta=0$ output the known
Gaussian maximally entangled state directly. For $\beta>0$, solve
the following ordinary linear least-squares problem:
\begin{equation}\label{ct:eq-least-squares}
 \min_{y_1,\ldots,y_L\in\ran P_R} \;
   \sum_{j=0}^{L-1}
       \left\|\frac{y_{j+1}-y_j}{\Delta \tau}-M(\tau_j)y_j\right\|^2,
 \qquad y_0=z(0).
\end{equation}
The initial vector is explicitly known:
\[
 z(0)=2^{-n_B/2}\mathcal U_0^\dagger\ket{\Id}_f.
\]
Its first $2n_B$ bits are all zero, while its boundary factor
is the vectorization of the identity on $\calF(\A)$ in the
chosen boundary basis. Thus $z(0)\in\ran P_R$, and only the
$2n_A$ unrestricted boundary bits need be expanded to specify it.
Every residual in \eqref{ct:eq-least-squares} lies in
$\ran\widehat P_R$, so this is an explicitly computable matrix problem
with $LD_R$ unknown coefficients and at most $L\widehat D_R$
residual coefficients. It involves no unknown exact thermal states.

For clarity, the norm in \eqref{ct:eq-least-squares} is taken in the
\emph{enlarged} span. Replacing $M(\tau_j)$ by $P_RM(\tau_j)P_R$ would
erase the leakage residual and would give a different algorithm.

\begin{proposition}[Error of the least-squares solution]\label{ct:prop-ls}
For $\beta>0$, $0<\delta_0\le1/2$, and sufficiently large
\[
 R=O\!\left(mS+\log\frac{(n+1)(1+\beta)}{\delta_0}\right),
\]
the vector $\mathcal U_\beta y_L$, after normalization, approximates
$\ket{\mathsf{TFD}_\beta(H)}$ to error $\delta_0$. The least-squares problem has
polynomial size in $n,1+\beta,\delta_0^{-1},e^{mS}$.
Before normalization it also satisfies
\begin{equation}\label{ct:eq-unnormalized-output}
 \norm{y_L-z(\beta)}\le\frac{\delta_0}{2C_{m,S}}.
\end{equation}
\end{proposition}
\begin{proof}
First bound its optimal residual. The trial values
$y_j=P_Rz(\tau_j)$ are not known to the algorithm, but prove existence
of a small objective. Taylor expansion
of the exact path gives the trial residual bound
\[
 \left\|\frac{y_{j+1}-y_j}{\Delta \tau}-M(\tau_j)y_j\right\|
 \le \frac{2a_R}{\Delta \tau}+B_1a_R+\Delta \tau\,B_1^2B_0.
\]
If $r_j$ are the residuals of a minimizing solution and
$\mathcal R=(\Delta \tau\sum_{j=0}^{L-1}\norm{r_j}^2)^{1/2}$, it follows that
\begin{equation}\label{ct:eq-objective-bound}
 \mathcal R\le\sqrt\beta
             \left(\frac{2a_R}{\Delta \tau}+B_1a_R+\Delta \tau\,B_1^2B_0\right).
\end{equation}

Interpolate the coefficients of the computed $y_j$ linearly in
the fixed basis $\{\ket{x}:x\in\mathcal X_R\}$, obtaining
$y(\tau)\in\ran P_R$. Put $r(\tau)=y'(\tau)-M(\tau)y(\tau)$.
This residual belongs to $\ran\widehat P_R$ almost everywhere.
On $[\tau_j,\tau_{j+1}]$, direct substitution and
\eqref{ct:eq-regularity} give
\begin{equation}\label{ct:eq-continuous-residual}
 \norm{r(\tau)}\le(1+\Delta \tau\,B_1)\norm{r_j}
                  +2\Delta \tau\,B_1^2\norm{y_j}.
\end{equation}
Let $\mathcal E=\sup_{0\le \tau\le\beta}\norm{y(\tau)-z(\tau)}$.
Duhamel's formula and Lemma~\ref{ct:lem-propagation}, applied only to
vectors in the residual span, imply
\[
 \mathcal E\le A_R\left[(1+\Delta \tau\,B_1)\sqrt\beta\,\mathcal R
                         +2\Delta \tau\,B_1^2\beta(B_0+\mathcal E)\right].
\]
Once $2A_R\Delta \tau\,B_1^2\beta\le1/2$, absorb the last occurrence of $\mathcal E$.
Using \eqref{ct:eq-objective-bound} then yields
\begin{equation}\label{ct:eq-error-bound}
 \mathcal E\le2A_R\left[(1+\Delta \tau\,B_1)
             (2a_RL+\beta B_1a_R+\beta \Delta \tau\,B_1^2B_0)
                         +2\Delta \tau\,B_1^2\beta B_0\right].
\end{equation}
Since $A_R\le e^{O(mS)}(n+1)2^{3R/2}$,
$a_R=e^{O(mS)}2^{-8R}$, and
$L\le2(1+\beta)2^{4R}$, both the absorption condition and
\eqref{ct:eq-error-bound} are controlled by
\[
 e^{O(mS)}\poly(n+1,1+\beta)\,2^{-5R/2}.
\]
Choose $R$ so this is at most $\delta_0/(2C_{m,S})$.
The lower bound $\norm{z(\beta)}\ge C_{m,S}^{-1}$ then makes the
error after normalization at most $\delta_0$.

Finally $D_R\le e^{O(mS)}2^R$, $\widehat D_R=O((n+1)^2D_R)$,
and $L\le2(1+\beta)2^{4R}$ are polynomial in the stated parameters.
The fixed retained set $\mathcal X_R$ can be enumerated in
$O((2n+1)D_R)$ operations by Lemma~\ref{lem:enum}, after omitting the
modes constrained to remain empty.
Matrix elements of $M(\tau_j)$ are obtained from its quadratic part and
the boundary interaction, using the finite-dimensional fermionic
matrix calculations described in Section~\ref{ge:algorithm}. Ordinary linear
least-squares computation has polynomial arithmetic cost in these
dimensions. It suffices to attain the objective to an inverse
polynomial residual tolerance, rather than to specify a particular
minimizer exactly.
There is also a polynomial conditioning bound for this discrete
problem. Apply the same residual estimate to a trial path with zero
initial value and compare it to the zero exact solution. After
absorption it gives
\[
 \left(\sum_{j=1}^L\norm{y_j}^2\right)^{1/2}
 \le 2A_R(1+\Delta \tau\,B_1)\beta
       \left(\sum_{j=0}^{L-1}\norm{r_j}^2\right)^{1/2}.
\]
Multiplying all residual equations by $\Delta \tau$ leaves the minimizer
unchanged. The resulting matrix has largest singular value at most
$2+\Delta \tau\,B_1$ and smallest singular value at least
$[2A_R(1+\Delta \tau\,B_1)L]^{-1}$. Its condition number is therefore polynomial
in the stated parameters. Sufficiently accurate scalar exponentials
and matrix entries require only polynomial numerical precision in
this arithmetic model.
\end{proof}

\subsection{Returning to an ungapped bath and the final Gaussian span}\label{ct:ungapped-section}

We complete the proof of Theorem~\ref{ct:main}, keeping the target
set $\mathcal X_{\beta,\delta}$ at the tolerance in its statement.
At $\beta=0$, use $\mathcal X_{0,\delta}=\{0^{2n}\}$ from
\eqref{tf:beta-zero} and output $c_{0^{2n}}=1$.
This is an exact normalized representation of $\ket{\mathsf{TFD}_0(H)}$.
Assume $\beta>0$ below. Here is the gap reduction needed for the
least-squares computation. Fix
\[
 \omega=\frac{\delta}{16m(1+\beta)},\qquad
 \widetilde e_j=\max(e_j,\omega),\qquad
 \widetilde H=\sum_{j=1}^n\widetilde e_j\beta_j^\dagger \beta_j+V.
\]
Let $P_<$ project onto the original energies below $\omega$, and set
$\A_<=P_<\Aimp$ and
$\B_<=P_<\mathbb C^n\ominus\A_<$.
Write the original one-particle matrix in this low-energy subspace as
\[
 \begin{pmatrix}h_{\A_<\A_<}&h_{\A_<\B_<}\\h_{\A_<\B_<}^\dagger&h_{\B_<\B_<}\end{pmatrix}
 \quad\text{on }\A_<\oplus\B_<.
\]
Define $H_{\rm eff}$ by replacing this block by
$\operatorname{diag}(\omega I_{\A_<},h_{\B_<\B_<})$.
The one-particle difference $D$ has rank at most $2m$.
Both low-energy matrices being subtracted are positive semidefinite
and at most $\omega I$. If they are denoted $X_<$ and $Y_<$, then
$-Y_<\preceq (X_<-Y_<)\preceq X_<$, hence $\norm D\le\omega$.
The norm of the corresponding quadratic Fock-space operator is at
most the sum of the absolute one-particle eigenvalues. Therefore
\[
 \norm{H_{\rm eff}-H}\le2m\omega.
\]
For arbitrary Hermitian Hamiltonians $A,B$, Duhamel's formula and
H\"older's inequality give
\begin{equation}\label{ct:eq-tfd-lipschitz}
 \norm{\ket{\mathsf{TFD}_\beta(A)}-\ket{\mathsf{TFD}_\beta(B)}}
       \le\frac\beta2\norm{A-B}.
\end{equation}
Indeed the derivative of the unnormalized heat operator has
Hilbert--Schmidt norm at most $\beta\norm{A-B}/2$ times its own
norm. Normalizing removes the parallel component of that derivative
and cannot increase this bound.

The modes in $\B_<$ commute with $V$ and factor as a free bath for
both $H_{\rm eff}$ and $\widetilde H$. Their energies are respectively
$h_{\B_<\B_<}$ and $\omega I$. A known Gaussian unitary $U_{\beta,\omega}$ on
their doubled modes maps the second free TFD to the first. Therefore
\begin{equation}\label{ct:eq-gap-error}
 \norm{\ket{\mathsf{TFD}_\beta(H)}-U_{\beta,\omega}\ket{\mathsf{TFD}_\beta(\widetilde H)}}
       \le\beta m\omega.
\end{equation}
Apply Proposition~\ref{ct:prop-ls} to $\widetilde H$ with tolerance
$\delta/8$, and then apply the Gaussian correction. The resulting
normalized vector $\ket{\psi_{\rm alg}}$ satisfies
\[
 \norm{\ket{\psi_{\rm alg}}-\ket{\mathsf{TFD}_\beta(H)}}
 \le\delta/8+\beta m\omega\le3\delta/16.
\]
This gives a computable Gaussian superposition approximating the
original TFD. Its cost follows from the parameter bounds
\[
 mS=O\!\left(m\log\frac{m(1+\beta)}{\delta}\right),\qquad
 R=O\!\left(m\log\frac{m(1+\beta)}{\delta}+\log(n+1)\right).
\]
Indeed $D_R\le e^{O(mS)}2^R$,
$\widehat D_R=O((n+1)^2D_R)$, and
$L\le2(1+\beta)2^{4R}$. Constructing and solving the least-squares
system therefore has the arithmetic cost in Theorem~\ref{ct:main}.
The preliminary ground-energy computation at accuracy
$1/[2(1+\beta)]$, the one-particle operations, and the Gaussian gap
correction fit within the same bound.

We now compute coefficients in the exact spanning set specified by
Theorem~\ref{ct:main}. Construct $U_\beta$ and
$\mathcal X_{\beta,\delta}$ for the original $H_0,\Aimp$ at tolerance
$\delta$, as in Theorem~\ref{thm:rank}, and use the projector
$\Pi_{\beta,\delta}$ from \eqref{tf:theorem-projector}.
For every $x\in\mathcal X_{\beta,\delta}$ define the ideal projected
coefficient
\begin{equation}\label{ct:eq-projected-coefficients}
 a_x=\langle x|U_\beta^\dagger|\psi_{\rm alg}\rangle.
\end{equation}
The vector with these coefficients is
$\Pi_{\beta,\delta}\ket{\psi_{\rm alg}}$.
Every intermediate Gaussian vector and every $U_\beta\ket{x}$ has
its phase fixed by its specified Bogoliubov transformation and ordered
Fock construction. Their overlaps can be computed by fermionic Gaussian
contractions, or Pfaffians of matrices of dimension $O(n)$. Summing
against the computed intermediate coefficients therefore evaluates
all $a_x$ in polynomial time in the two spanning-set sizes and $n$.

Compute approximations $\widehat a_x$ such that
\begin{equation}\label{ct:eq-coefficient-accuracy}
 \left(\sum_{x\in\mathcal X_{\beta,\delta}}
                  |\widehat a_x-a_x|^2\right)^{1/2}\le\delta/16,
 \qquad
 \ket{\widetilde\phi}=\sum_{x\in\mathcal X_{\beta,\delta}}
                  \widehat a_xU_\beta\ket{x}.
\end{equation}
For example, evaluating each coefficient to additive error
$\delta/(16\sqrt{|\mathcal X_{\beta,\delta}|})$ suffices.
The concentration theorem gives
$\norm{(I-\Pi_{\beta,\delta})\ket{\mathsf{TFD}_\beta(H)}}\le\delta$.
By contractivity of orthogonal projection and the intermediate error
bound above,
\begin{align}
 \norm{\ket{\widetilde\phi}-\ket{\mathsf{TFD}_\beta(H)}}
 &\le\delta/16+3\delta/16+\delta
 =5\delta/4.\label{ct:eq-final-projection-error}
\end{align}
Since $\delta\le1/2$, this implies $\norm{\widetilde\phi}\ge3/8$.
Thus normalization is well defined and stable. Output
\begin{equation}\label{ct:eq-output-coefficients}
 c_x=\frac{\widehat a_x}
 {\left(\sum_{y\in\mathcal X_{\beta,\delta}}|\widehat a_y|^2\right)^{1/2}}
 \quad(x\in\mathcal X_{\beta,\delta}),\qquad
 \ket\phi=\sum_{x\in\mathcal X_{\beta,\delta}}c_xU_\beta\ket x.
\end{equation}
Orthonormality gives $\sum_{x\in\mathcal X_{\beta,\delta}}|c_x|^2=1$.
For a nonzero vector $q$ and a normalized vector $\psi$, the reverse
triangle inequality gives
$\norm{q/\norm q-\psi}\le2\norm{q-\psi}$.
Applying this inequality to \eqref{ct:eq-final-projection-error},
\begin{equation}\label{ct:eq-final-state-error}
 \norm{\ket\phi-\ket{\mathsf{TFD}_\beta(H)}}\le5\delta/2\le3\delta.
\end{equation}
Consequently the theorem holds, for example, with the universal
constant $C=3$. No change from $\mathcal X_{\beta,\delta}$ to a set
with a different tolerance is made in this final step. Constructing
this set and computing the projected coefficients fit within the
runtime already established. This completes the proof of
Theorem~\ref{ct:main}.

The normalized output also specifies a reduced state
$\widehat\rho=\Tr_{\rm anc}\ket\phi\bra\phi$. Contractivity of the
trace norm under partial trace and the pure-state distance bound give
\[
 \norm{\widehat\rho-\rho_\beta(H)}_1
 \le2\norm{\ket\phi-\ket{\mathsf{TFD}_\beta(H)}}\le2C\delta.
\]

\begin{remark}[Role of the interaction]
The interaction preserves the retained span because all doubled
boundary occupations are unrestricted:
$[\mathcal P_R(\tau),V_L]=0$. The quadratic part of the generator maps
this span into the explicitly enumerable residual span. The
least-squares argument uses exact propagation of each residual and
therefore requires no stability bound for a sequence of projections
at different temperatures.
\end{remark}

\section{Approximating free energy and thermal expectation values}\label{ct:thermodynamics}

Thermal expectation values can be used to approximate partition
functions and free energies, as shown by Bravyi, Chowdhury, Gosset,
and Wocjan~\cite{bravyi2022quantum}.
Their reduction varies the
inverse temperature and reconstructs the partition function from
thermal energy estimates.  Here we use a variant that keeps the temperature fixed and varies
the strength of the interaction $V$, which gives a better
dependence of the runtime on the system size $n$.

Throughout this section, assume $\norm{V}\leq1$, and let $C\geq1$
be the universal constant in Theorem~\ref{ct:main}. We use $\delta$
for the input tolerance of the TFD approximation algorithm and
$\varepsilon$ for the desired additive error in a scalar quantity such as the free energy or  thermal expectation values.

\paragraph{Thermal expectation values.}
Let $P$ be a Hermitian observable, and let $P_L$ denote the
corresponding operator on the doubled system, acting as $P$
on the physical modes and trivially on the ancillary modes.
Since the physical reduced state of
$\ket{\mathsf{TFD}_\beta(H)}$ is $\rho_\beta(H)$, one has
\[
 \Tr\!\left(P\rho_\beta(H)\right)
 =\bra{\mathsf{TFD}_\beta(H)}P_L
       \ket{\mathsf{TFD}_\beta(H)}.
\]
Theorem~\ref{ct:main} computes a normalized state
\[
 \ket\phi=\sum_{x\in\mathcal X_{\beta,\delta}}
                  c_xU_\beta\ket{x},
 \qquad
 \norm{\ket\phi-\ket{\mathsf{TFD}_\beta(H)}}\leq C\delta.
\]
Normalization and the identity $\norm{P_L}=\norm P$ imply
\begin{equation}\label{ct:eq-observable-error}
 \left|\bra\phi P_L\ket\phi
       -\Tr\!\left(P\rho_\beta(H)\right)\right|
 \leq2C\delta\norm P.
\end{equation}
We have
\[
 \bra\phi P_L\ket\phi
 =\sum_{x,y\in\mathcal X_{\beta,\delta}}
      \overline{c_x}c_y
      \bra{x}U_\beta^\dagger P_L U_\beta\ket{y}.
\]
A state $U_\beta \ket{x}$ is Gaussian for any $x\in \{0,1\}^{2n}$.
A matrix element of any Majorana monomial between any pair of Gaussian states can be computed in time  $O(n^3)$ 
using the generalized Wick theorem of~\cite{bravyi2017complexity}. 
Accordingly, if $P$ is given as a sum of $M$ 
Majorana monomials, any matrix element   $\bra{x}U_\beta^\dagger P_L U_\beta\ket{y}$ can be computed in time $O(Mn^3)$.
Thus $\bra\phi P_L\ket\phi$ can be computed in time
\[
 O\!\left(Mn^3|\mathcal X_{\beta,\delta}|^2\right).
\]
Choose an orthonormal basis of $\Aimp$ and extend its
associated Majorana operators to a full Majorana basis.
In this basis, $V$ has at most $4^m$ Majorana monomials, while
$H_0$ has at most $O(n^2)$ monomials. Consequently 
the expectations $\langle \phi| V_L|\phi\rangle$
and $\langle \phi|(H_0)_L |\phi\rangle$
can be computed in time 
$\poly(n)4^m|\mathcal X_{\beta,\delta}|^2$.

\paragraph{Free energy.}
Below we write  $\mathsf F_\beta(H)$ for the free energy, where $\beta$ is the inverse temperature and $H$ is the Hamiltonian. 
We now prove Corollary~\ref{ct:free-energy}.
\begin{proof}
Fix $\beta>0$ and $0<\varepsilon\leq1/2$. For $0\leq\lambda\leq1$,
define
\[
 H_\lambda=H_0+\lambda V,
 \qquad
 g(\lambda)=\Tr\!\left(V\rho_\beta(H_\lambda)\right).
\]
Differentiating the free energy and using cyclicity of the
trace gives
\begin{equation}\label{ct:eq-free-energy-gradient}
 \frac{d}{d\lambda}\mathsf F_\beta(H_\lambda)=g(\lambda)
\end{equation}
(this identity does not require $H_0$ and $V$ to commute). Consequently,
\begin{equation}\label{ct:eq-coupling-integration}
 \mathsf F_\beta(H)
 =\mathsf F_\beta(H_0)+\int_0^1g(\lambda)\,d\lambda.
\end{equation}
The free contribution is
\begin{equation}\label{ct:eq-free-bath-energy}
 \mathsf F_\beta(H_0)
 =-\frac1\beta\sum_{k=1}^n\log(1+e^{-\beta e_k}),
\end{equation}
where $e_1,\ldots,e_n\in[0,1]$ are the bath excitation energies.

The Gibbs variational principle implies that
$\mathsf F_\beta(H_\lambda)$ is concave in $\lambda$, since it is
the minimum, over density matrices $\sigma$, of
\[
 \Tr(\sigma H_0)+\lambda\Tr(\sigma V)
       +\beta^{-1}\Tr(\sigma\log\sigma).
\]
Thus $d^2 \mathsf F_\beta(H_\lambda)/d\lambda^2\le 0$, that is,
$g(\lambda)$ is a nonincreasing function. Also, $|g(\lambda)|\leq\norm V\leq1$.
For
\[
 N=\lceil4/\varepsilon\rceil,
 \qquad \lambda_j=j/N\quad(0\leq j\leq N),
\]
the left and right Riemann sums therefore give
\begin{equation}\label{ct:eq-coupling-quadrature}
 0\leq\frac1N\sum_{j=0}^{N-1}g(\lambda_j)
          -\int_0^1g(\lambda)\,d\lambda
 \leq\frac{g(0)-g(1)}N
 \leq\frac2N\leq\frac\varepsilon2.
\end{equation}

At each grid point, apply Theorem~\ref{ct:main} to $H_{\lambda_j}$
with state tolerance $\delta=\varepsilon/(4C)$, obtaining state
$\ket{\phi_j}$ approximating $\ket{\mathsf{TFD}_\beta(H_{\lambda_j})}$ within error $C\delta$.
 By~\eqref{ct:eq-observable-error},
\[
 \left|\bra{\phi_j}V_L\ket{\phi_j}-g(\lambda_j)\right|
 \leq\varepsilon/2.
\]
Set $\widehat g_j=\bra{\phi_j}V_L\ket{\phi_j}$, which can be
computed using Wick's theorem.
Assume for simplicity that the free bath contribution  $\mathsf F_\beta(H_0)$ is computed exactly.
Return
\begin{equation}\label{ct:eq-free-energy-output}
 F_{\rm out}=\mathsf F_\beta(H_0)+
             \frac1N\sum_{j=0}^{N-1}\widehat g_j.
\end{equation}
Equations~\eqref{ct:eq-coupling-integration}--\eqref{ct:eq-free-energy-output}
give
\[
 |F_{\rm out}-\mathsf F_\beta(H)|
 \leq\varepsilon/2+\varepsilon/2
 =\varepsilon.
\]

Every $H_\lambda$ has the same bath and impurity support and satisfies
$\norm{\lambda V}\leq1$. There are $O(\varepsilon^{-1})$ calls to
Theorem~\ref{ct:main}, each with $\delta=\Theta(\varepsilon)$.
Thus the total runtime is
\[
 \poly(n,1+\beta,\varepsilon^{-1})
 \exp\!\left[O\!\left(m\log\frac{m(1+\beta)}{\varepsilon}\right)\right],
\]
as claimed.
\end{proof}

\section{BQP-completeness of simulating time evolution}\label{ct:BQP}
We now prove Theorem~\ref{thm:main-reduction}.  As described in the introduction, the construction encodes a quantum circuit into the initial
state of a time-independent quantum impurity model.  The quadratic part
\(H_0\) moves localized one-particle states through the system, while a
fixed quartic interaction \(V\), supported on \(24\) fermionic modes,
implements the computation when these states reach the interacting region.
 The proof proceeds in two steps.  We first formalize the sequential
computation described in the introduction and show how an arbitrary circuit
can be compiled into it.  We then construct the corresponding impurity
Hamiltonian and bound the error between its time evolution and the ideal
computation.

\subsection{Sequential computation model}
\label{sec:1}
We first define the computation model that will later be encoded into the fermionic system. 
This model has the same computational power as the quantum circuit model, but is more convenient for our purpose. 
The core idea is to use a fixed update rule throughout the computation, as in quantum Turing machines \cite{Deutsch1985,BernsteinVazirani1997}; 
closely related ideas, where the program is supplied as part of the initial state, also appear in programmable quantum cellular automata \cite{ShepherdFranzWerner2006} and quantum ring machines \cite{BauschCubittOzols2017}.

The model has three subsystem, which we call data, program, and processor.
The size of the processor is fixed, while those of the data and program will scale with the circuit that we want to encode. 
At each time step, a data or program register visits the processor, and a fixed unitary $U_*$ will be applied on the combined system.

\paragraph{Data, program, and processor}
Let \(D_1,\ldots,D_k\) be \(k\) data registers, each of which is a qubit $D_i \simeq \mathbb C^2$. 

The processor $Q$ contains two parts: an auxiliary qubit \(A\simeq\mathbb C^2\) and a five-dimensional command register \(C\simeq\mathbb C^5\).  
We write
$    Q:=A\otimes C$ 
and use the basis \(\ket{a,c}_Q\), where
\(a\in\{0,1\}\) and \(c\in\{0,1,2,3,4\}\).
 The value \(c\) specifies which operation is to be applied when the next
data qubit is presented.  Define
\begin{equation}\label{eq:1-2}
    G_0=I,\qquad
    G_1=H_D,\qquad
    G_2=T_D,\qquad
    G_3=\operatorname{SWAP}_{A,D},\qquad
    G_4=\operatorname{CZ}_{A,D}.
\end{equation}
Thus \(G_0,G_1,G_2\) act trivially on \(A\), while \(G_3,G_4\) act on
\(A\otimes D\). 

We also introduce program registers
\[
    P_1,\ldots,P_{K+1},
    \qquad
    P_r\simeq\mathbb C^5.
\]

\paragraph{Initialization and update rule}
Each \(P_r\) is initialized in a basis state \(\ket{p_r}\), with \(p_r\in\{0,1,2,3,4\}\).  
The sequence
\(p_1,\ldots,p_{K+1}\) will be chosen in the next section from the circuit to be simulated. 
Everything else, including the processor and the data registers, will be initialized in $\ket{0}$.

At each step, either a program register or a data register interacts with
\(Q\).  To describe both cases by one fixed unitary, let
\[
    X:=P\oplus D,
    \qquad
    P\simeq\mathbb C^5,\quad D\simeq\mathbb C^2.
\]
We denote the corresponding basis states by
\(\ket{P,p}\), \(p\in\{0,\ldots,4\}\), and
\(\ket{D,d}\), \(d\in\{0,1\}\). Define the fixed unitary \(U_*\) on \(Q\otimes X\) by
\begin{equation}\label{eq:1-5}
\begin{aligned}
U_*(\ket{a,c}_Q\ket{P,p}) &= \ket{a,p}_Q\ket{P,c},\\
U_*(\ket{a,c}_Q\ket{D,d}) &= G_c\bigl(\ket a_A\ket d_D\bigr)\otimes\ket c_C.
\end{aligned}
\end{equation}
The unitary \(U_*\) is the fixed update performed at every visit:
a program visit loads the next operation into \(C\) by a swap, while upon a data visit, the command currently stored in \(C\) is executed according to \cref{eq:1-2} via a controlled unitary.
The registers are presented in the order
\[
    P_1,D_1,P_2,D_2,\ldots,P_K,D_K,P_{K+1},
\]
where the data-register labels are extended periodically by
\[
    D_{r+k}:=D_r.
\]
Thus the data qubits are visited repeatedly in the order
\(D_1,D_2,\ldots,D_k,D_1,\ldots\), with one program register preceding
each data visit.

Since program visits only exchange
basis states with \(C\), all program registers and \(C\) remain in basis
states throughout the computation.

\subsection{Compiling circuits into programs}\label{sec:2}

Without loss of generality, let \(\mathcal C\) be a quantum circuit on \(k\ge 1\) qubits with \(g\)
gates from \(\{H,T,\operatorname{CZ}\}\), which is a universal gate set.
The input is
\(\ket{0^k}\), and we take \(D_k\) to be the output qubit.  Writing
\(\ket{\psi_{\mathcal C}}=\mathcal C\ket{0^k}\), let
\begin{equation}
    \label{eq:1-1}
    p_{\mathcal C}
    :=
    \expval{(\ketbra{0})_k}{\psi_{\mathcal C}} .
\end{equation}
Now we specify the values \(p_r\) so that the above fixed evolution in Section~\ref{sec:1} implements the given quantum circuit $\mathcal C$.

To do so, process the gates of \(\mathcal C\) in their
circuit order. For a gate \(H\) or \(T\) acting on \(D_i\), wait until the next visit of
\(D_i\), set the preceding program value $p$ to \(1\) or \(2\), respectively,
and use \(0\) at all intermediate visits. For a gate \(\operatorname{CZ}_{D_i,D_j}\), first swap the state of
\(D_i\) into \(A\), apply \(\operatorname{CZ}_{A,D_j}\) when \(D_j\)
next appears, and swap the state back when \(D_i\) appears again.  This
uses the identity
\begin{equation}
    \label{eq:2-2}
    I_A\otimes\operatorname{CZ}_{D_i,D_j}
    =
    \operatorname{SWAP}_{A,D_i}
    \operatorname{CZ}_{A,D_j}
    \operatorname{SWAP}_{A,D_i}.
\end{equation}
After this sequence, the state of \(A\) is again \(\ket 0\).  Hence \(A\)
serves only as temporary storage while implementing a two-qubit gate.
A one-qubit gate requires at most one complete pass through
\(D_1,\ldots,D_k\), and a two-qubit gate requires at most three.
Therefore the \(g\) gates of \(\mathcal C\) require at most \(3kg\)
data visits.  
Furthermore, after completing the circuit, we make one additional pass
and use identity operations until \(D_k\) is reached.  We pad if necessary
so that this final visit has index \(K\) with \(D_K=D_k\), and choose
\(p_K=3\) to swap the output qubit into \(A\).  Thus
\begin{equation}
    \label{eq:2-3}
    K\le (3g+1)k .
\end{equation}
Finally, we append \(p_{K+1}=0\).  The last program visit resets the
command register \(C\) to \(0\) and does not act on the data qubits. 

With \(p_0:=0\), the initial state of the sequential computation is
\begin{equation}
    \label{eq:omegain}
    \ket{\omega_{\mathcal C}^{\mathrm{in}}}
    =
    \ket{0,0}_Q\otimes\ket{0^k}_D
    \otimes
    \bigotimes_{r=1}^{K+1}\ket{p_r}_{P_r}.
\end{equation}
After all visits, the state is
\begin{equation}
    \label{eq:1-7}
    \ket{\omega_{\mathcal C}^{\mathrm{out}}}
    =
    \operatorname{SWAP}_{A,D_k}
    \left(
        \ket{0,0}_Q\otimes\ket{\psi_{\mathcal C}}_D
        \otimes
        \bigotimes_{r=1}^{K+1}\ket{p_{r-1}}_{P_r}
    \right).
\end{equation}
Indeed, each program visit exchanges the new value \(p_r\) with the value
previously stored in \(C\), so \(P_r\) contains \(p_{r-1}\) afterward.
The final choice \(p_{K+1}=0\) therefore returns \(C\) to \(0\). It follows immediately that measuring \(A\) in
\(\ket 0\) after the protocol gives the same probability as measuring
the output qubit \(D_k\) of the original circuit:

\begin{lemma}
    With the protocol introduced in Section~\ref{sec:1}, we have:
\begin{equation}
\label{eq:output-prob}
    \expval{(\ketbra{0})_A}{\omega_{\mathcal C}^{\mathrm{out}}} = \expval{(\ketbra{0})_{k}}{\psi_{\mathcal C}}=p_{\mathcal C}.
\end{equation}
\end{lemma}

\paragraph{Example}
As a short but complete example, take three data qubits, and take $\mathcal C = T_2 \operatorname{CZ}_{1,3} H_3$. 
The scan rule gives:
\begin{center}
\begin{tabular}{llll}
\toprule
$r$ & Command $p_r$ & Visiting register & Operation \\
\midrule
$1$ & $0$ & $D_1$ & $I$ \\
$2$ & $0$ & $D_2$ & $I$ \\
$3$ & $1$ & $D_3$ & $H_{D_3}$ \\
$4$ & $3$ & $D_1$ & $\operatorname{SWAP}_{A,D_1}$ \\
$5$ & $0$ & $D_2$ & $I$ \\
$6$ & $4$ & $D_3$ & $\operatorname{CZ}_{A,D_3}$ \\
$7$ & $3$ & $D_1$ & $\operatorname{SWAP}_{A,D_1}$ \\
$8$ & $2$ & $D_2$ & $T_{D_2}$ \\
$9$ & $3$ & $D_3$ & $\operatorname{SWAP}_{A,D_3}$ \\
$10$ & $0$ &  &  \\
\bottomrule
\end{tabular}
\end{center}
The corresponding visit schedule is
\begin{equation}
\begin{aligned}
&P_1[0],D_1,P_2[0],D_2,P_3[1],D_3,P_4[3],D_1,
P_5[0],D_2,P_6[4],D_3,P_7[3],D_1,P_8[2],D_2,\\
&P_9[3],D_3,P_{10}[0].
\end{aligned}
\end{equation}
The three registers in the second line are due to our convention to swap the output qubit to $A$ and reset $C$, as explained above.

\subsection{Encoding into quantum impurity models}

We now realize the sequential computation model through the time evolution of a quantum impurity model. Following the sketch, we first encode the registers as moving fermionic packets, then we design a quartic interaction to implement a fixed unitary $U_*$ on the processor and the passing register. 
Throughout, $K\le (3g+1)k$ is the number of data visits.

\subsubsection{Wave packets for data and program registers}

We arrange fermionic modes on five open chains labeled by $p\in\{0,1,2,3,4\}$ and two rings labeled by $d\in\{0,1\}$. 
Each ring has $L_D=80k\sigma$ sites, while each path has $2B_0+1$ sites indexed by $x=-B_0,\ldots,B_0$, where $B_0=100(K+k+1)\sigma$, and $\sigma\ge 2$ is an integer to be chosen later. 

Define a profile function for a wave packet starting at site $z$:
\begin{equation}
    \chi_{z}(x) = (-i)^x f_\sigma(x-z),
\end{equation}
where $f_\sigma(x)$ satisfies:
\begin{equation}
\label{eq:profile}
    f_\sigma(x)\propto
    \begin{cases}
        x^3(\sigma-x)^3, & 0<x<\sigma,\\
        0, & \text{otherwise}.
    \end{cases}
\end{equation}
We normalize $f_\sigma(x)$ as:
\begin{equation}
    \sum_{i=0}^{\sigma-1}  \abs{f_\sigma(i)}^2 = 1.
\end{equation}

\begin{figure}[H]
\centering
\includegraphics[width=0.65\linewidth]{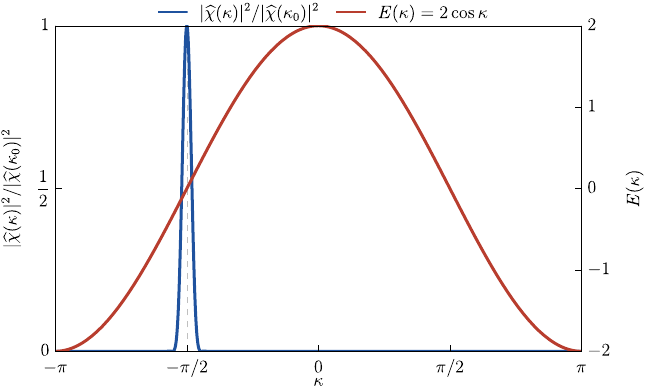}
\caption{Blue: Momentum distribution (relative to the peak at $\kappa_0=-\pi/2$) of the wave packet, evaluated for $\sigma=80$.
Red: dispersion relation for nearest-neighbor hopping. }
\label{fig:packet-dispersion}
\end{figure}

On a data ring, position labels are taken modulo $L_D$. Since $L_D$ is divisible by four, the phase $(-i)^x$ is consistent around each ring. 
The profile is designed such that its momentum concentrates near $\kappa_0=-\pi/2$. 
Recall that, the dispersion relation for unit nearest-neighbor hopping is
\begin{equation}
E(\kappa)=2\cos\kappa,\qquad E'(-\frac\pi2)=2, \qquad E''(-\frac\pi2)=0.
\end{equation}
Therefore, the group velocity is $2$ at $\kappa_0$.
The quadratic dispersion $E''(-\frac\pi2)$ is absent, so wave-packet deformation is governed by higher-order terms in the expansion of $E(\kappa)$ and is relatively smaller.
See \Cref{fig:packet-dispersion} for an illustration.
Conceptually, one may think this wave packet as a rigid particle moving with velocity 2; the propagation error is bounded in \Cref{sec:8-1}.

Let $z_i^D$ and $z_r^P$ denote the starting positions of wave packets for data register $D_i$ and program register $P_r$, respectively.
Here, the starting positions $z_i^D$ and $z_r^P$ are chosen in a equally spaced ($40\sigma$), interlaced way:
\begin{align}
\label{eq:7-2}
z_i^D&=-(80i+2)\sigma, &&i=1,\ldots,k,\\
\label{eq:7-3}
z_r^P&=-\bigl(80r-38\bigr)\sigma, &&r=1,\ldots,K+1,
\end{align}
with data positions taken modulo $L_D$. 

Let $f_{x,D,d}^\dagger$ and $f_{x,P,p}^\dagger$ be fermion creation operators at site $x$ in the indicated path or ring.
Using the profile $\chi_{z}$ defined above, we define the packet creation operators by
\begin{equation}\label{eq:createdp}
\begin{aligned}
\xi_{D_i,d}^\dagger
&=\sum_x\chi_{z_i^D}(x)f_{x,D,d}^\dagger,\\
\xi_{P_r,p}^\dagger
&=\sum_x\chi_{z_r^P}(x)f_{x,P,p}^\dagger.
\end{aligned}
\end{equation}
We will use these operators to create fermions in the corresponding paths or rings.
Since they are far enough apart, the wave packets have disjoint supports, so the corresponding one-particle orbitals created by \Cref{eq:createdp} are orthonormal.

In the approximation that each wave packet moves rigidly at velocity $2$ without spreading, program packet $P_r$ crosses the $0\leftrightarrow 1$ edge within the time window
\begin{equation}
\left[20\sigma(2r-1),\;20\sigma(2r-1)+2\sigma\right],
\end{equation}
while data packet $D_i$ makes its first crossing during
\begin{equation}
\left[40\sigma i,\;40\sigma i+2\sigma\right].
\end{equation}
Thus program and data packets visit the processor alternately without overlap.

\subsubsection{Processor fermion and impurity interaction}
Let $a_q$ be the fermionic operators for 10 different modes, $q=(a,c)\in\{0,1\}\times\{0,\ldots,4\}$. 
The ten modes will be placed in the middle of the $0\leftrightarrow 1$ edge, as shown in \Cref{fig:processor-architecture}(b).
We use the one-hot encoding to encode the 10 dimensional processor $Q$ as one fermion, so that the processor basis state $\ket{q}_Q=\ket{a,c}_Q$ correspond to the state created by $a_{(a,c)}^\dagger$. 

The total number of modes is therefore
\begin{equation}
\label{eq:mode-count}
    n=10+5(2B_0+1)+2L_D=O((K+k)\sigma).
\end{equation}
We use $\ket{\mathrm{vac}}$ to denote the Fock vacuum, i.e., the unique state that is annihilated by all $N$ annihilation operators defined above.

\paragraph{Hamiltonian.}\label{sec:5-2}
Let $\alpha,\beta$ range over the seven basis states of $X=P\oplus D$, and set $u_{q'\alpha,q\beta}=\matrixel{q',\alpha}{U_*}{q,\beta}$, where $U_*$ is defined in \Cref{eq:1-5}.
Define
\begin{equation}
\label{eq:5-3}
V_{\mathcal I}=\sum_{q,q',\alpha,\beta}u_{q'\alpha,q\beta}
f_{1,\alpha}^\dagger a_{q'}^\dagger a_q f_{0,\beta}+\mathrm{h.c.}
\end{equation}
The interaction $V_{\mathcal I}$ is Hermitian, quartic, and independent of the circuit $\mathcal C$ being encoded. Its coefficients lie in $\mathbb Q(i,\sqrt2)$, and its impurity support consists of the ten processor modes and the seven packet-channel modes at position 0 and 1:
\begin{equation}
\label{eq:5-4}
|\mathcal I|=10+7\times 2=24.
\end{equation}

Let $h$ be the unit-hopping matrix on the paths and rings defined above, with the seven links $0\leftrightarrow1$ removed and with zero hopping on the processor modes. Set
\begin{equation}
H_{0}=\sum_{x,y}\tilde h_{xy}f_x^\dagger f_y,
\qquad H=H_{0}+V_{\mathcal I}.
\end{equation}
It belongs to the class of quantum impurity models.
Here, $f_x$ ranges over the operators $f_{z,P,p}$, $f_{z,D,d}$.
The hopping $\tilde h$ is the nearest-neighbor hopping on each chain or ring, with the $0\leftrightarrow1$ edge dropped; 
namely, $\tilde h_{x,y}=1$ if and only if $x$ and $y$ are on the same ring (for data modes) or chain (for program modes), and their spatial distance is 1, and they do not lie on the $0\leftrightarrow1$ edge; and $\tilde h_{x,y}=0$ otherwise.
We impose open boundary condition for the program chain and periodic boundary condition for the data ring.

\begin{lemma}\label{lem:5-1}
The one-particle hopping matrix satisfies $\|\tilde h\|\le2$.
The Hamiltonian $H$ separately conserves the number of processor fermions, data fermions, and program fermions. 
On the sector with one processor fermion and one packet,
\begin{equation}\label{eq:Vrb}
V_{\mathcal I} a_q^\dagger f_{0,\beta}^\dagger\ket{\mathrm{vac}}
=\sum_{q',\alpha}u_{q'\alpha,q\beta}
a_{q'}^\dagger f_{1,\alpha}^\dagger\ket{\mathrm{vac}},
\end{equation}
with the adjoint transition from endpoint $1$ to endpoint $0$ and zero action on packets away from the edge.
\end{lemma}
\begin{proof}
    An application of Cauchy-Schwarz implies that
    \begin{equation}
        \norm{\tilde h}^2 \leq (\max_i \sum_j \abs{\tilde h_{ij}})(\max_j \sum_i \abs{\tilde h_{ij}}).
    \end{equation}
Since $\tilde h$ is Hermitian with absolute row sums at most two, we have $\|\tilde h\|\le2$.

Each term in $H$ preserves processor and total packet-fermion number; separate conservation of data and program number follows from $U_*(Q\otimes P)=Q\otimes P$ and $U_*(Q\otimes D)=Q\otimes D$.

The canonical anticommutation relations give
\begin{equation}
\label{eq:5-6}
f_{1,\alpha}^\dagger a_{q'}^\dagger a_q f_{0,\beta}
r_s^\dagger f_{0,\gamma}^\dagger\ket{\mathrm{vac}}
=\delta_{q,s}\delta_{\beta,\gamma}a_{q'}^\dagger f_{1,\alpha}^\dagger\ket{\mathrm{vac}}.
\end{equation}
Summing over $q,q',\alpha,\beta$ proves the \Cref{eq:Vrb}. 
The fact that $V_{\mathcal I}$ has zero action on away from the edge follows from its construction.
\end{proof}

\paragraph{Single-packet evolution.}\label{sec:4}

We verify that the interaction implements $U_*$ when a single packet passes the processor.
The full many-body time evolution is analyzed in \Cref{sec:8}.

Because the scattering is local and the packets remain far from the distant boundaries during one time window, it is convenient first to analyze the same problem on the infinite line, coupled to the processor through $V_{\mathcal I}$ at the edge $0\leftrightarrow1$.
Consider the Fock space sector where there is one data or program fermion, and one processor fermion. 
\Cref{lem:5-1} shows that this sector is invariant under time evolution.
We may identify this sector with $\ell^2(\mathbb Z)\otimes Q\otimes X$.
Define the uniform spatial hopping operators:
\begin{equation}
\label{eq:4-1}
\begin{aligned}
h&=\sum_{x\in\mathbb Z}(\ketbra{x+1}{x}+\mathrm{h.c.}),\\
h_{\mathrm f}&=h\otimes I_{Q\otimes X}.
\end{aligned}
\end{equation}
By \Cref{lem:5-1}, the full Hamiltonian on this sector has the position-space representation
\begin{equation}
h_U=h_{\mathrm f}
+\ketbra{1}{0}\otimes(U_*-I)
+\ketbra{0}{1}\otimes(U_*^\dagger-I).
\end{equation}
Write $P_L=\sum_{x\le0}\ketbra{x}$, $P_R=\sum_{x\ge1}\ketbra{x}$, and define
\begin{equation}
\label{eq:4-2}
\Gamma=P_L\otimes I+P_R\otimes U_*.
\end{equation}

\begin{lemma}
    The operator $\Gamma$ is unitary, and
\begin{equation}
\label{eq:4-3}
h_U=\Gamma h_{\mathrm f}\Gamma^\dagger,
\qquad
e^{-ih_Ut}=\Gamma e^{-ih_{\mathrm f}t}\Gamma^\dagger.
\end{equation}
\end{lemma}
\begin{proof}
    Since $P_L+P_R=I$ and $P_LP_R=0$, unitarity of $U_*$ gives
$\Gamma^\dagger\Gamma=\Gamma\Gamma^\dagger=I$.
For the distinguished edge, we have
\begin{equation}
\Gamma\bigl(\ketbra{1}{0}\otimes I\bigr)\Gamma^\dagger
=\ketbra{1}{0}\otimes U_*,
\end{equation}
and also its adjoint.
Every other hopping term has both endpoints on the same side and is unchanged by conjugation.
This proves $h_U=\Gamma h_{\mathrm f}\Gamma^\dagger$ and hence \Cref{eq:4-3}.
\end{proof}

\begin{lemma}[Single-packet propagation]\label{lem:single-packet-propagation}
Let $\chi,\chi_{\mathrm{out}}$ be two normalized spatial states and let $\zeta\in Q\otimes X$ be any normalized internal state. Then the following hold.
\begin{enumerate}
\item \emph{Exact evolution.}
If $P_L\chi=\chi$, then, writing
$\chi_L(t)=P_Le^{-iht}\chi$ and
$\chi_R(t)=P_Re^{-iht}\chi$, we have
\begin{equation}
\label{eq:4-4}
e^{-ih_Ut}(\chi\otimes\zeta)
=\chi_L(t)\otimes\zeta+\chi_R(t)\otimes U_*\zeta.
\end{equation}

\item \emph{Error bounds.}
Suppose
\begin{equation}
\|e^{-iht}\chi-\chi_{\mathrm{out}}\|\le\delta.
\end{equation}
If $P_L\chi=\chi$ and
$P_R\chi_{\mathrm{out}}=\chi_{\mathrm{out}}$, then
\begin{equation}
\label{eq:4-5}
\|e^{-ih_Ut}(\chi\otimes\zeta)
-\chi_{\mathrm{out}}\otimes U_*\zeta\|\le\delta.
\end{equation}
If instead, $P_s\chi=\chi$ and $P_s\chi_{\mathrm{out}}=\chi_{\mathrm{out}}$ for the same $s\in\{L,R\}$, then
\begin{equation}
\|e^{-ih_Ut}(\chi\otimes\zeta)
-\chi_{\mathrm{out}}\otimes\zeta\|\le\delta.
\end{equation}
\end{enumerate}
\end{lemma}

\begin{proof}
Since $P_L\chi=\chi$, we have $\Gamma^\dagger(\chi\otimes\zeta)=\chi\otimes\zeta$. Hence
\begin{align}
e^{-ih_Ut}(\chi\otimes\zeta)
&=\Gamma e^{-ih_{\mathrm f}t}\Gamma^\dagger(\chi\otimes\zeta)
=\Gamma e^{-ih_{\mathrm f}t}\bigl(\chi\otimes\zeta\bigr)
=\Gamma \bigl(e^{-iht}\chi\otimes\zeta\bigr)\\
&=(P_Le^{-iht}\chi)\otimes\zeta
 +(P_Re^{-iht}\chi)\otimes U_*\zeta,
\end{align}
which is \Cref{eq:4-4}.

Now suppose $\|e^{-iht}\chi-\chi_{\mathrm{out}}\|\le\delta$.
For left-to-right transfer, $P_L\chi=\chi$ and
$P_R\chi_{\mathrm{out}}=\chi_{\mathrm{out}}$, giving
\begin{equation}
\Gamma(\chi_{\mathrm{out}}\otimes\zeta)
=\chi_{\mathrm{out}}\otimes U_*\zeta.
\end{equation}
Therefore,
\begin{align}
\|e^{-ih_Ut}(\chi\otimes\zeta)
-\chi_{\mathrm{out}}\otimes U_*\zeta\|
=\|\Gamma((e^{-iht}\chi-\chi_{\mathrm{out}})\otimes\zeta)\|
=\|e^{-iht}\chi-\chi_{\mathrm{out}}\|
\le\delta,
\end{align}
which proves \Cref{eq:4-5}.
Here we have used the unitarity of $\Gamma$ and $\|\zeta\|=1$.

If both $\chi$ and $\chi_{\mathrm{out}}$ are supported on the left, the same calculation applies with
$\Gamma(\chi_{\mathrm{out}}\otimes\zeta)=\chi_{\mathrm{out}}\otimes\zeta$.
If both are supported on the right, then
\begin{align}
\Gamma^\dagger(\chi\otimes\zeta)
=\chi\otimes U_*^\dagger\zeta,\qquad
\Gamma(\chi_{\mathrm{out}}\otimes U_*^\dagger\zeta)
=\chi_{\mathrm{out}}\otimes\zeta.
\end{align}
Consequently,
\begin{align}
\|e^{-ih_Ut}(\chi\otimes\zeta)
-\chi_{\mathrm{out}}\otimes\zeta\|
=\|\Gamma((e^{-iht}\chi-\chi_{\mathrm{out}})
\otimes U_*^\dagger\zeta)\|
=\|e^{-iht}\chi-\chi_{\mathrm{out}}\|
\le\delta.
\end{align}
\end{proof}

For later convenience, we remark that the same conclusion and proof holds if the infinite chain is replaced by a finite open interval containing $0$ and $1$.
Moreover, let $\mathcal R$ be an arbitrary spectator space, and extend the operators $h_{\mathrm f},h_U,\Gamma,U_*$ by the identity on $\mathcal R$, and let $\zeta\in Q\otimes X\otimes\mathcal R$ be any normalized internal state.
The conclusion still holds, without a product-state assumption between the internal registers and the spectators.

\subsubsection{Initial state, readout, and simulation guarantee}\label{sec:6}

\paragraph{Initial state}

Given the quantum circuit $\mathcal C$, we compile it into a program $(p_r)_{r=1}^{K+1}$, following the procedure in \Cref{sec:2}.
Then, using the packet operators in \Cref{eq:createdp}, we define
\begin{equation}
\ket{\Psi_{\mathcal C}^{\mathrm{in}}}
=a_{(0,0)}^\dagger
\prod_{i=1}^k \xi_{D_i,0}^\dagger
\prod_{r=1}^{K+1} \xi_{P_r,p_r}^\dagger
\ket{\mathrm{vac}}.
\end{equation}
Here the products are ordered by register labels (within each product, operators with smaller labels are closer to $\ket{\mathrm{vac}}$). The occupied orbitals are orthonormal, so this is a normalized Slater determinant. 

Due to \Cref{lem:5-1}, the number of data fermions, program fermions, processor fermions will remain invariant throughout the entire process, which are $k$, $K+1$, and 1, repectively.
The total number of fermions is:
\begin{equation}\label{eq:defF}
    F=k+K+2.
\end{equation}

\paragraph{Readout}
Define the total simulation time as
\begin{equation}
\label{eq:simulation-time}
    T=20\sigma(2K+1)+2\sigma=2\sigma(20K+11).
\end{equation}
Starting from $\ket{\Psi_{\mathcal C}^{\mathrm{in}}}$, we time evolve it by time $T$ using the quantum impurity Hamiltonian $H$, then we measure register $A$ of the processor by measuring the occupation number of mode $(0,0)$.
Its expectation value is:
\begin{equation}
    \expval{e^{iHT}a_{(0,0)}^\dagger a_{(0,0)}e^{-iHT}}{\Psi_{\mathcal C}^{\mathrm{in}}}.
\end{equation}

\subsection{Error estimation and hardness reduction}\label{sec:8}

The main technical result of this note is the following.
\begin{theorem}\label{thm:simulation}
Let $\mathcal C$ be a $k$-qubit circuit with $g$ gates.
Let $p_{\mathcal C}=\expval{\mathcal C^\dagger (\ketbra{0})_k \mathcal C}{0^k}$.
For any $0<\epsilon\le1/24$, the packet width $\sigma$ can be chosen so that the construction above satisfies
\begin{equation}
\label{eq:6-3}
\left|\expval{e^{iHT} a_{(0,0)}^\dagger a_{(0,0)} e^{-iHT}}{\Psi_{\mathcal C}^{\mathrm{in}}}
-p_{\mathcal C}\right|
\le \epsilon.
\end{equation}
For fixed $\epsilon$, the instance satisfies $n,T=O((k+g)^{4})$.
\end{theorem}

\paragraph{Strategy of the proof.}
The overall idea is to compare the actual time evolution under $H$ with an idealized evolution: wave packets are rigid particles moving with a common constant speed, and each time a wave packet crosses the processor zone, a logical $U_*$ is applied.
We call the latter the reference evolution, and wave packets in it reference wave packets, etc.

Divide $[0,T]$ into $W=T/(2\sigma)=20K+11$ time windows, and let $\Phi_j$ be the reference state at time $t_j=2\sigma j$. 
$\Phi_0 = \Psi_{\mathcal C}^{\mathrm{in}}$. 
The reference output state $\Phi_W$ yields $p_{\mathcal C}$ exactly.
By unitarity, window errors only add up,
\begin{equation}
\|e^{-iHT}\Psi^{\mathrm{in}}_{\mathcal C}-\Phi_W\|
\le\sum_{j<W}\|e^{-2i\sigma H}\Phi_j-\Phi_{j+1}\|.
\end{equation}
To estimate each term, we cut the system into intervals, each containing at most one packet. Then we consider $H_j=\Pi_jH\Pi_j$ where the many-body projector $\Pi_j$ confines each of the packets to its own interval, and split
\begin{equation}
\|e^{-2i\sigma H}\Phi_j-\Phi_{j+1}\|
\le\underbrace{\|(e^{-2i\sigma H}-e^{-2i\sigma H_j})\Phi_j\|}_{\text{confinement}}
+\underbrace{\|e^{-2i\sigma H_j}\Phi_j-\Phi_{j+1}\|}_{\text{dispersion}}.
\end{equation}
Under $H_j$ the evolution factorizes exactly into single-body problems, each of which can be controlled directly.
The first term measures the error caused by confining the packets in their respective intervals; it can be bounded, ultimately since the packets finite propagation speed.
The second term captures the error due to wave-packet dispersion, which is bounded due to our choice of the wave packets profile (more precisely, its momentum distribution).

\subsubsection{Wave-packet profile and single-body estimation}
\paragraph{Free propagation in one time window}\label{sec:8-1}
We first consider the free propagation of a single wave packet without any quartic interaction. 
This is a single-body problem and we work in the first quantization picture.
We may also ignore the internal degree of freedom and focus on the spatial degree of freedom.

More precisely, we consider one wave packet on an infinite line with profile $\chi$, and consider its time evolution under the hopping Hamiltonian $h=\sum_x(\ketbra{x+1}{x}+\mathrm{h.c.})$.

The following lemma (similar to Theorem 1 in \cite{childs2013universal}) controls how well this can be approximated by the idealized spatial translation with velocity 2.
Here, for ease of notation, we identify $\chi_z$ with the single body wavefunction $\ket{\chi_z}=\sum \chi_z(x)\ket{x}$, and $e^{-iht}\chi_z$ is an abbreviation for $e^{-iht}\ket{\chi_z}$. 

\begin{lemma}\label{lem:8-1}
On the infinite chain, for every integer $\sigma\ge2$ and every integer $0\le t\le2\sigma$, we have
\begin{equation}
\label{eq:8-1}
\|e^{-iht}\chi_{z}-\chi_{z+2t}\|\le O(\frac{1}{\sigma^{2}}),
\end{equation}
\end{lemma}
\begin{proof}
In momentum space, the hopping Hamiltonian acts by multiplication:
\begin{equation}
\widehat{h\psi}(\kappa)=E(\kappa)\widehat\psi(\kappa),
\qquad E(\kappa)=2\cos\kappa.
\end{equation}

Define the Fourier transform of $f_\sigma$ as $\widehat {f_\sigma}(\kappa)=(2\pi)^{-1/2}\sum_x e^{-i\kappa x}f_\sigma(x)$, with momentum understood modulo $2\pi$. 
Denote $\kappa_0=-\pi/2$ and $q=\kappa-\kappa_0\in[-\pi,\pi]$.
Then the Fourier transforms of $\chi_z$ and $\chi_{z+2t}$ are
\begin{equation}
    \begin{aligned}
        \widehat{\chi_z}(\kappa_0+q) &= e^{-iqz}\widehat {f_\sigma} (q),\\
        \widehat{\chi_{z+2t}}(\kappa_0+q) &= e^{-iq(z+2t)}\widehat {f_\sigma} (q).
    \end{aligned}
\end{equation}
Note that $E(\kappa_0+q)=2\sin q$, we have, with norms in $L^2[-\pi,\pi]$,
\begin{equation}
\|e^{-iht}\chi_{z}-\chi_{z+2t}\|
=\|(e^{-2it\sin q}-e^{-2itq})\widehat{f_\sigma}\|
\le\frac t3\|q^3\widehat{f_\sigma}\|.
\end{equation}
In the last step, we used $|e^{ia}-e^{ib}|\le|a-b|$ and $|\sin q-q|\le|q|^3/6$.

Denote the backward difference for any function as $\nabla f(s)=f(s)-f(s-1)$.
We have
\begin{equation}
    \widehat{\nabla f} = (1-e^{-iq}) \widehat{f}.
\end{equation}
Applying it three times on $\widehat{f_\sigma}$ and using $|q|\le\frac\pi2|1-e^{-iq}|$ on $[-\pi,\pi]$, we get:
\begin{equation}
\|q^3\widehat{f_\sigma}\|\le\left(\frac\pi2\right)^3\|\nabla^3f_\sigma\|.
\end{equation}

It remains to bound $\|\nabla^3f_\sigma\|$.
Denote $g(x)=x^3(\sigma-x)^3$ for $0\le x\le\sigma$ and $g(x)=0$ otherwise, so that $f_\sigma=g/\|g\|$.
Note that $g$ is $C^2$ and $g'''$ is piecewise continuous, we have
\begin{equation}
    \nabla^3g(s)=\int_{[0,1]^3}g'''(s-u_1-u_2-u_3)\,du_1du_2du_3.
\end{equation}
By Cauchy--Schwarz and a scaling argument, we get:
\begin{equation}
    \|\nabla^3g\|^2\le\int_{\mathbb R}g'''(x)^2\,dx \propto \sigma^7.
\end{equation}
For the normalization, again by scaling argument, we get:
\begin{equation}
    \|g\|^2 = \Omega(\sigma^{13}).
\end{equation}
Combining the above with $t\le2\sigma$, we get
\begin{equation}
\|e^{-iht}\chi_{z}-\chi_{z+2t}\|
\lesssim \sigma \frac{\norm{\nabla^3g}}{\norm{g}} = O(\frac{1}{\sigma^2}).
\end{equation}
\end{proof}

We note that, the hopping Hamiltonian acts only on position and leaves the internal registers unchanged.
The same error bound therefore applies even when the packet's internal state is entangled with other registers.

\paragraph{Truncating to intervals}\label{sec:8-2}

As a comparison model, keep a finite interval $\Lambda=\{a,\ldots,b\}$ and set the hopping amplitudes across $(a-1,a)$ and $(b,b+1)$ to zero. For the finite interval $\Lambda$, this gives
\begin{equation}
P_\Lambda=\sum_{x=a}^b\ketbra{x},
\qquad
h_\Lambda=P_\Lambda hP_\Lambda
=\sum_{x=a}^{b-1}(\ketbra{x+1}{x}+\mathrm{h.c.}).
\end{equation}
We compare evolution under $h_\Lambda$ with evolution under the infinite-chain Hamiltonian $h$ for a wave packet initially supported on $\Lambda$.

Intuitively, due to \Cref{lem:8-1}, a packet moves only $4\sigma$ sites during one time window, so if it is initially $10\sigma$ sites away from the ends of its interval, it essentially never reaches them, and cutting the hopping there changes almost nothing. 
The following lemma makes this intuition precise.

\begin{lemma}[Truncation, single-body version]\label{lem:boundary-truncation}
Let $\sigma\ge1$ and $0\le t\le2\sigma$.
\begin{enumerate}
\item For every unit vector $\phi\in\ell^2(\mathbb Z)$ supported on $[a+10\sigma,\,b-10\sigma]$,
\begin{equation}
\label{eq:truncation}
\|(e^{-iht}-e^{-ih_\Lambda t})\phi\|\le2^{-5\sigma}.
\end{equation}
\item With a processor, i.e.\ for $h_U$ on $\ell^2(\mathbb Z)\otimes Q\otimes X$ from \Cref{sec:4}, for every unit vector $\phi$ supported on $[a+10\sigma,\,b-10\sigma]$,
\begin{equation}
\|(e^{-ih_Ut}-e^{-iP_\Lambda h_UP_\Lambda t})\phi\|\le2^{-5\sigma},
\end{equation}
where $P_\Lambda$ stands for $P_\Lambda\otimes I_{Q\otimes X}$.
\end{enumerate}
\end{lemma}

The lemma extends directly when $\ell^2(\mathbb Z)$ is replaced by $\ell^2(\mathbb Z)\otimes\mathcal K$ for any Hilbert space $\mathcal K$.
In particular, $\phi$ may be arbitrarily entangled with other registers.

\begin{proof}
We prove (i) first. 
The two evolutions differ only when the particle hops out through an endpoint. We show that the amplitude at the endpoints is tiny and then integrate over time.
Define
\begin{equation}
    \psi(u)=e^{-ih_\Lambda u}\phi,
\end{equation}
which stays supported on $\Lambda$.

By Duhamel's formula and unitarity of $e^{-ih(t-u)}$,
\begin{equation}\label{eq:Duhamel}
    \|(e^{-iht}-e^{-ih_\Lambda t})\phi\|
=\Bigl\|\int_0^te^{-ih(t-u)}(h-h_\Lambda)\psi(u)\,du\Bigr\|
\le\int_0^t\|(h-h_\Lambda)\psi(u)\|\,du.
\end{equation}
Since $\psi(u)$ is supported on $\Lambda$, only the two cut bonds contribute:
\begin{equation}
\begin{aligned}
(h-h_\Lambda)\psi(u)&=\ket{a-1}\!\bra{a}\psi(u)+\ket{b+1}\!\bra{b}\psi(u),\\
\|(h-h_\Lambda)\psi(u)\|^2&=\|E_a\psi(u)\|^2+\|E_b\psi(u)\|^2.
\end{aligned}
\end{equation}
Here, we denote $E_x=\ketbra{x}$.

Denote by $D$ the diagonal operator $D\ket x=2^{a+10\sigma-x}\ket x$ on $\ell^2(\Lambda)$.
The weight equals $2^{10\sigma}$ at $a$ and is at most $1$ on the support of $\phi$, so $\|D\phi\|\le1$ and
\begin{equation}
\|E_a\psi(u)\|\le2^{-10\sigma}\|D\psi(u)\|
\le2^{-10\sigma}\|De^{-ih_\Lambda u}D^{-1}\|
=2^{-10\sigma}\|e^{uA}\|,
\qquad
A=-iDh_\Lambda D^{-1}.
\end{equation}
Write $h_\Lambda=S+S^\dagger$ where $S$ is the right shift on $\Lambda$. Then $DSD^{-1}=S/2$ and $DS^\dagger D^{-1}=2S^\dagger$, so
\begin{equation}
\tfrac12(A+A^\dagger)=\tfrac{3i}{4}(S-S^\dagger),
\qquad
\bigl\|\tfrac12(A+A^\dagger)\bigr\|\le\tfrac32.
\end{equation}
By the logarithmic-norm bound $\|e^{uA}\|\le e^{u\,\lambda_{\max}((A+A^\dagger)/2)}$ for $u\ge0$, this implies
\begin{equation}
    \|e^{uA}\|\le e^{3u/2}
\end{equation}
and hence
\begin{equation}
    \|E_a\psi(u)\|\le2^{-10\sigma}e^{3u/2}.
\end{equation}
Similar argument gives the same bound for $\|E_b\psi(u)\|$.
Plugging the bounds into \Cref{eq:Duhamel}, we get:
\begin{equation}
\|(e^{-iht}-e^{-ih_\Lambda t})\phi\|
\le\sqrt2\,2^{-10\sigma}\int_0^{2\sigma}e^{3u/2}\,du
\le2^{-5\sigma}.
\end{equation}

For (ii), $h_U=\Gamma h_{\mathrm f}\Gamma^\dagger$ by \Cref{eq:4-3}. Since $\Gamma$ is diagonal in position, it commutes with $P_\Lambda$, so $P_\Lambda h_UP_\Lambda=\Gamma(h_\Lambda\otimes I)\Gamma^\dagger$, and $\Gamma^\dagger\phi$ has the same support as $\phi$. Hence same argument gives:
\begin{equation}
\|(e^{-ih_Ut}-e^{-iP_\Lambda h_UP_\Lambda t})\phi\|
=\|\bigl((e^{-iht}-e^{-ih_\Lambda t})\otimes I\bigr)\Gamma^\dagger\phi\|
\le2^{-5\sigma}.
\end{equation}
\end{proof}

\begin{corollary}[Finite-interval propagation]\label{cor:finite-interval-propagation}
Let $\sigma\ge 2$. Suppose the support of $\chi_z$ is at least $10\sigma$ sites from both endpoints of $\Lambda$. Let $h_{U,\Lambda}=P_\Lambda h_UP_\Lambda$ denote the restriction of $h_U$ to $\Lambda$, and let $\zeta$ be any normalized internal state. Then for every integer $0\le t\le2\sigma$:
\begin{enumerate}
\item[(a)] If $\{0,1\}\not\subset\Lambda$, then $h_{U,\Lambda}=h_\Lambda\otimes I$ and
\begin{equation}
\label{eq:8-4}
\left\|e^{-ih_{U,\Lambda}t}(\chi_z\otimes\zeta)-\chi_{z+2t}\otimes\zeta\right\|
= O(\sigma^{-2}).
\end{equation}
\item[(b)] If $\{0,1\}\subset\Lambda$ and $\chi_z$ lies entirely on one side of the $0\leftrightarrow 1$ edge, then
\begin{equation}
\label{eq:8-4b}
\left\|e^{-ih_{U,\Lambda}t}(\chi_z\otimes\zeta)
-\Gamma(\chi_{z+2t}\otimes\eta)\right\|
= O(\sigma^{-2}),
\end{equation}
where $\eta=\zeta$ for left-side input and $\eta=U_*^\dagger\zeta$ for right-side input.
\end{enumerate}
All operators act as the identity on any other registers entangled with $\zeta$.
\end{corollary}

\begin{proof}
By the triangle inequality,
\begin{equation}
\|e^{-ih_\Lambda t}\chi_z-\chi_{z+2t}\|
\le\|(e^{-ih_\Lambda t}-e^{-iht})\chi_z\|+\|e^{-iht}\chi_z-\chi_{z+2t}\|.
\end{equation}
By \Cref{lem:boundary-truncation}~(i), the first term is $O(2^{-5\sigma})$. By \Cref{lem:8-1}, the second term is $O(\sigma^{-2})$. 
Therefore,
\begin{equation}
\label{eq:free-interval}
\|e^{-ih_\Lambda t}\chi_z-\chi_{z+2t}\|= O(\sigma^{-2}).
\end{equation}

(a) If $\{0,1\}\not\subset\Lambda$, then $P_\Lambda\ketbra{1}{0}P_\Lambda=P_\Lambda\ketbra{0}{1}P_\Lambda=0$, so $h_{U,\Lambda}=h_\Lambda\otimes I$ and \Cref{eq:8-4} follows from \Cref{eq:free-interval}.

(b) Since $\Gamma$ is diagonal in position, it commutes with $P_\Lambda$, so \Cref{eq:4-3} implies $e^{-ih_{U,\Lambda}t}=\Gamma(e^{-ih_\Lambda t}\otimes I)\Gamma^\dagger$. Since $\chi_z$ lies on one side of the edge, $\Gamma^\dagger(\chi_z\otimes\zeta)=\chi_z\otimes\eta$. As $\Gamma$ is unitary, \Cref{eq:8-4b} follows from \Cref{eq:free-interval}.
\end{proof}

\subsubsection{Many-body estimation}\label{sec:8-3}

\paragraph{Reference states.}\label{sec:5-3}
We define a family of reference states, as the ideal fermionic many-body states assuming the wave packets are independent and rigid particles with constant speed.

Recall that we have divided $[0,T]$ into a total number of $W$ windows each of width $2\sigma$.
At time $t_j=2\sigma j$, all reference wave packets have moved $4\sigma j$ sites. 
Write the position of them as
\begin{equation}
x_{D_i,j}=z_i^D+4\sigma j,\qquad
x_{P_r,j}=z_r^P+4\sigma j,
\end{equation}
with data positions taken modulo $L_D$. Let $\chi_j$ denote these profiles, which is a translation of $\chi$.
Let $\xi_{D_i,d}^\dagger[j]$ and $\xi_{P_r,p}^\dagger[j]$ denote the operators in \Cref{eq:createdp} with $z_i^D$ and $z_r^P$ replaced by $x_{D_i,j}$ and $x_{P_r,j}$, respectively. 
The label $[j]$ simply specifies the spatial translation of the packet profiles; in particular, $\xi^\dagger[0]=\xi^\dagger$. 
The encoding is
\begin{equation}
\label{eq:5-5}
\mathcal{J}_j\bigl(\ket{q}\otimes\ket{d_1,\ldots,d_k}\otimes\ket{s_1,\ldots,s_{K+1}}\bigr)
=a_q^\dagger\prod_{i=1}^k \xi_{D_i,d_i}^\dagger[j]
\prod_{r=1}^{K+1} \xi_{P_r,s_r}^\dagger[j]\ket{\mathrm{vac}}.
\end{equation}
Disjointness of the normalized profiles ensures that $\mathcal{J}_j$ is an isometry.

Starting from $\ket{\omega_{\mathcal C}^{\mathrm{in}}}$ defined in \Cref{eq:omegain}, let $\ket{\omega_{\mathcal C,j}}$ be the logical state after the scheduled visits completed by time $t_j$, and define $\ket{\Phi_j}=\mathcal{J}_j\ket{\omega_{\mathcal C,j}}$. 
In particular,
\begin{equation}
\label{eq:6-1}
\ket{\Phi_0}=\ket{\Psi_{\mathcal C}^{\mathrm{in}}},\qquad
\ket{\Phi_W}=\mathcal{J}_W\ket{\omega_{\mathcal C}^{\mathrm{out}}}.
\end{equation}

\paragraph{The auxiliary Hamiltonian.}
Label the packets by $\mu=1,\ldots,F-1$ with $F=k+K+2$ defined in \Cref{eq:defF}.
During window $j$, we define an interval that contains packet $\mu$ for each $\mu$
\begin{equation}
\Lambda_{\mu,j}=[x_{\mu,j}-10\sigma,\;x_{\mu,j}+15\sigma].
\end{equation}
Each interval leaves a margin of at least $10\sigma$ around its reference packet throughout the window. 
The packet spacing is large enough so that these regions are still disjoint, with at most one containing the processor. 

Let $\Pi_j$ be the orthogonal projector onto states with exactly one register fermion of the assigned type in each $\Lambda_{\mu,j}$, exactly one processor fermion, and no other fermions.
Define
\begin{equation}
H_j=\Pi_jH\Pi_j.
\end{equation}
Importantly, since each interval contains exactly one packet and $H_j$ has no hopping between intervals, its evolution on $\operatorname{Im}\Pi_j$ factors exactly into single-packet free evolutions and, if present, one joint evolution of the processor and its nearby packet.

More precisely, since the intervals are disjoint, we have an isometric identification:
\begin{equation}
    \operatorname{Im}\Pi_j \cong \Big(\bigotimes_\mu\bigl(\ell^2(\Lambda_{\mu,j})\otimes X_\mu\bigr)\Big)\otimes Q,
\end{equation}
where $X_\mu=P$ or $D$ according to the carrier type.
Under this identification,
\begin{equation}
\label{eq:Hj-factor}
H_j=\sum_\mu H_{\mu,j},
\end{equation}
where $H_{\mu,j}=h_{\Lambda_{\mu,j}}\otimes I$ if $\Lambda_{\mu,j}$ does not contain the edge $0\leftrightarrow1$, and $H_{\mu_*,j}=P_{\Lambda_{\mu_*,j}}h_UP_{\Lambda_{\mu_*,j}}$ for the packet $\mu_*$ whose interval contains it, if such interval exists. 
The terms $H_{\mu,j}$ act on disjoint sets of tensor factors and mutually commute.

\begin{lemma}[Auxiliary evolution in one time window]\label{lem:auxiliary-window}
For every $0\le j<W$,
\begin{equation}
\label{eq:8-5}
\|e^{-2i\sigma H_j}\Phi_j-\Phi_{j+1}\|
= O\bigl((k+K+1)\sigma^{-2}\bigr).
\end{equation}
\end{lemma}

\begin{proof}
Write $V_\mu=e^{-2i\sigma H_{\mu,j}}$, so that $e^{-2i\sigma H_j}=V_1\cdots V_{F-1}$ by \Cref{eq:Hj-factor}. Let $S$ be the shift $\ket{x}\mapsto\ket{x+4\sigma}$, and let $S_\mu=P_{\Lambda_{\mu,j}}SP_{\Lambda_{\mu,j}}$ act on $\ell^2(\Lambda_{\mu,j})$. Define the single-interval reference maps
\begin{equation}
R_\mu=\begin{cases}
S_\mu\otimes I,&\mu\neq\mu_*,\\
\Gamma(S_\mu\otimes I)\Gamma^\dagger,&\mu=\mu_* ~(\text{if exists}).
\end{cases}
\end{equation}
Each $R_\mu$ has norm at most $1$ and $R_\mu$ commutes with $V_{\mu'}$ and $R_{\mu'}$ for $\mu'\neq\mu$.
Every packet moves $4\sigma$ sites during one window. By design of the time window, the reference packet of $\mu_*$ does not intersect with the $0\leftrightarrow 1$ edge at both $t_j$ and $t_{j+1}$, and $R_{\mu_*}$ applies $U_*$ exactly when it crosses from left to right. Hence
\begin{equation}
\Phi_{j+1}=R_1\cdots R_{F-1}\,\Phi_j.
\end{equation}

Since the support of every packet in $\Phi_j$ is $10\sigma$ sites from the ends of its interval, \Cref{eq:8-4,eq:8-4b} with $t=2\sigma$ give
\begin{equation}
\|(V_\mu-R_\mu)\Phi_j\|=O(\sigma^{-2}),\qquad \mu=1,\ldots,F-1.
\end{equation}
Telescoping and using $\|V_\mu\|=1$, $\|R_\mu\|\le1$,
\begin{equation}
\begin{aligned}
\|e^{-2i\sigma H_j}\Phi_j-\Phi_{j+1}\|
\le\sum_{\mu=1}^{F-1}\|(V_\mu-R_\mu)\Phi_j\|=O\bigl((k+K+1)\sigma^{-2}\bigr).
\end{aligned}
\end{equation}
\end{proof}

The following is the many-body version of \Cref{lem:boundary-truncation}, and the proof is similar.

\begin{lemma}[Truncation, many-body version]\label{lem:many-body-truncation}
For every $0\le j<W$,
\begin{equation}
\label{eq:many-body-truncation}
\|(e^{-2i\sigma H}-e^{-2i\sigma H_j})\Phi_j\|\le(k+K+1)\,2^{-5\sigma}.
\end{equation}
\end{lemma}

\begin{proof}
Define
\begin{equation}
    \Psi(u)=e^{-iH_ju}\Phi_j,
\end{equation}
which stays in $\operatorname{Im}\Pi_j$. By Duhamel's formula and unitarity of $e^{-iH(2\sigma-u)}$,
\begin{equation}\label{eq:Duhamel2}
\|(e^{-2i\sigma H}-e^{-2i\sigma H_j})\Phi_j\|
\le\int_0^{2\sigma}\|(H-H_j)\Psi(u)\|\,du.
\end{equation}
Every term of $H$ moves at most one packet by one site, and the quartic terms on $0\leftrightarrow1$ act inside $\Lambda_{\mu_*,j}$. 
Since $\Psi(u)\in\operatorname{Im}\Pi_j$, $(H-H_j)\Psi(u)$ only receives contribution from terms that hop a packet from an interval to the outside.
Writing $\Lambda_{\mu,j}=[a_\mu,b_\mu]$ and $E^{(\mu)}_x$ for the projector onto packet $\mu$ at site $x$, then
\begin{equation}
\|(H-H_j)\Psi(u)\|
\le\sum_\mu\Bigl(\|E^{(\mu)}_{a_\mu}\Psi(u)\|^2+\|E^{(\mu)}_{b_\mu}\Psi(u)\|^2\Bigr)^{1/2}.
\end{equation}

Fix a packet $\mu$. By \Cref{eq:Hj-factor}, $\Psi(u)=e^{-iH_{\mu,j}u}\prod_{\mu'\neq\mu}e^{-iH_{\mu',j}u}\,\Phi_j$, and the unitaries with $\mu'\neq\mu$ act on other factors, hence commute with $E^{(\mu)}_{x}$. Therefore
\begin{equation}
\|E^{(\mu)}_{a_\mu}\Psi(u)\|=\|E^{(\mu)}_{a_\mu}e^{-iH_{\mu,j}u}\Phi_j\|.
\end{equation}
This is a single-body quantity.
Indeed, for $\mu\neq\mu_*$, $H_{\mu,j}$ acts as $h_{\Lambda_{\mu,j}}$ on the position of packet $\mu$;
for $\mu=\mu_*$, it acts as $P_{\Lambda_{\mu,j}}h_UP_{\Lambda_{\mu,j}}$ on $\ell^2(\Lambda_{\mu,j})\otimes Q\otimes X_\mu$, with $h_U$ restricted to the carrier type $X_\mu$.
Since the packet of $\mu$ in $\Phi_j$ is $10\sigma$ sites from the ends of $\Lambda_{\mu,j}$, the proof of \Cref{lem:boundary-truncation} (with the remark following it, and part (ii) for $\mu=\mu_*$) gives
\begin{equation}
\|E^{(\mu)}_{a_\mu}\Psi(u)\|,\ \|E^{(\mu)}_{b_\mu}\Psi(u)\|\le2^{-10\sigma}e^{3u/2}.
\end{equation}

Plugging the bounds into \Cref{eq:Duhamel2}, we get:
\begin{equation}
\|(e^{-2i\sigma H}-e^{-2i\sigma H_j})\Phi_j\|
\le(k+K+1)\sqrt2\,2^{-10\sigma}\int_0^{2\sigma}e^{3u/2}\,du
\le(k+K+1)\,2^{-5\sigma}.
\end{equation}
\end{proof}

\subsubsection{BQP hardness}\label{sec:8-4}

Now we can compare the time evolution under the physical Hamiltonian $H$ and the reference time evolution.

\begin{lemma}[One-step error]\label{lem:physical-window}
For every $0\le j<W$,
\begin{equation}
\label{eq:8-8}
\|e^{-2i\sigma H}\Phi_j-\Phi_{j+1}\|
= O\bigl((k+K+1)\sigma^{-2}\bigr).
\end{equation}
\end{lemma}
\begin{proof}
By the triangle inequality,
\begin{equation}
\|e^{-2i\sigma H}\Phi_j-\Phi_{j+1}\|
\le\|(e^{-2i\sigma H}-e^{-2i\sigma H_j})\Phi_j\|
 +\|e^{-2i\sigma H_j}\Phi_j-\Phi_{j+1}\|.
\end{equation}
By \Cref{lem:many-body-truncation}, the first term is at most $(k+K+1)2^{-5\sigma}$. By \Cref{lem:auxiliary-window}, the second term is $O((k+K+1)\sigma^{-2})$. 
This proves \Cref{eq:8-8}.
\end{proof}

\begin{corollary}[Total error]\label{cor:total-error}
The physical state at time $T$ satisfies
\begin{equation}
\label{eq:8-9}
\|e^{-iHT}\Psi_{\mathcal C}^{\mathrm{in}}-\Phi_W\|
= O\bigl(W(k+K+1)\sigma^{-2}\bigr).
\end{equation}
\end{corollary}
\begin{proof}
Write $\Psi_j=e^{-it_jH}\Psi_{\mathcal C}^{\mathrm{in}}$ and let
$\delta_{\mathrm{w}}=O((k+K+1)\sigma^{-2})$ be the bound in \Cref{eq:8-8}. Then
\begin{align}
\|\Psi_{j+1}-\Phi_{j+1}\|
&\le\|e^{-2i\sigma H}(\Psi_j-\Phi_j)\|
    +\|e^{-2i\sigma H}\Phi_j-\Phi_{j+1}\|\\
&\le\|\Psi_j-\Phi_j\|+\delta_{\mathrm{w}}.
\end{align}
Since $\Psi_0=\Phi_0$, induction gives
\begin{equation}
\|\Psi_W-\Phi_W\|\le W\delta_{\mathrm{w}}.
\end{equation}
\end{proof}

This error analysis proves the main technical result \Cref{thm:simulation}.
The theorem on BQP-hardness then follows straightforwardly.

\begin{proof}[Proof of \Cref{thm:simulation,thm:main-reduction}]

To obtain state error at most $\epsilon$, by \Cref{eq:8-9} it suffices to take $\sigma$ to be the least integer such that
\begin{equation}
\label{eq:9-1}
\sigma \ge C\left(\frac{W(k+K+1)}{\epsilon}\right)^{1/2}
\end{equation}
for a suitable absolute constant $C$.
Then \Cref{eq:8-9} gives
\begin{equation}
\label{eq:6-2}
\|e^{-iHT}\Psi_{\mathcal C}^{\mathrm{in}}-\Phi_W\|\le\epsilon.
\end{equation}
By \Cref{eq:6-1,eq:1-7}, the output can be measured by the occupation number on mode $(0,0)$ in the processor
\begin{equation}
    \expval{a_{(0,0)}^\dagger a_{(0,0)}}{\Phi_W}=p_{\mathcal C}.
\end{equation}
Since $a_{(0,0)}^\dagger a_{(0,0)}$ is a projector, its expectation values in two normalized states differ by at most their distance, so \Cref{eq:6-2} implies
\begin{equation}
\left|\expval{e^{iHT} a_{(0,0)}^\dagger a_{(0,0)} e^{-iHT}}{\Psi_{\mathcal C}^{\mathrm{in}}}
-p_{\mathcal C}\right|
\le \epsilon,
\end{equation}
which is exactly the desired \Cref{eq:6-3}.
    
Using $K\le (3g+1)k$ and $W=20K+11$, the parameters chosen for the impurity model satisfy
\begin{equation}
\label{eq:9-2}
\sigma=O\!\left(K^{1/2}(k+K)^{1/2}\epsilon^{-1/2}\right),\quad
n=O((k+K)\sigma),\quad T=O(K\sigma).
\end{equation}
For fixed error, $K,k,W=O((k+g)^2)$, so $\sigma=O((k+g)^{2})$ and $n,T=O((k+g)^{4})$.

Other requirements of \Cref{thm:main-reduction} also hold: $\|\tilde h\|\le2$ by \Cref{lem:5-1}; the orbitals are normalized with disjoint supports, so $Z^\dagger Z=I$.
The time $T$, every entry of $\tilde h$, and the creation operators that define the initial state $\Psi_{\mathcal C}^{\mathrm{in}}$ are all given by explicit formulas in $k$, $g$ and the gate list, and are effectively computable in time polynomial in $k+g$.
\end{proof}

\bibliographystyle{alpha}
\bibliography{references}

\end{document}